\documentclass[12pt,reqno]{amsart}

\newcommand\version{May 11, 2015}

\usepackage{amsmath,amsfonts,amsthm,amssymb,amsxtra}

\newtheorem{theorem}{Theorem}[section]
\newtheorem{proposition}[theorem]{Proposition}
\newtheorem{lemma}[theorem]{Lemma}
\newtheorem{corollary}[theorem]{Corollary}

\theoremstyle{definition}

\newtheorem{assumption}[theorem]{Assumption}

\theoremstyle{remark}

\numberwithin{equation}{section}

\renewcommand{\epsilon}{\varepsilon}
\newcommand{\I}{\mathbb{I}}

\newcommand{\N}{\mathbb{N}}

\renewcommand{\phi}{\varphi}
\newcommand{\R}{\mathbb{R}}
\newcommand{\Rp}{\text{Re\,}}

\newcommand{\Z}{\mathbb{Z}}

\DeclareMathOperator{\im}{Im}

\DeclareMathOperator{\re}{Re}

\DeclareMathOperator{\tr}{Tr}

\begin{document}

\title[Derivation of an effective evolution equation --- \version]{Derivation of an effective evolution equation\\ for a strongly coupled polaron}

\author{Rupert L. Frank}
\address{Rupert L. Frank, Mathematics 253-37, Caltech, Pasadena, CA 91125, USA}
\email{rlfrank@caltech.edu}

\author{Zhou Gang}
\address{Zhou Gang, Mathematics 253-37, Caltech, Pasadena, CA 91125, USA}
\email{gzhou@caltech.edu}

\begin{abstract}
Fr\"ohlich's polaron Hamiltonian describes an electron coupled to the quantized phonon field of an ionic crystal. We show that in the strong coupling limit the dynamics of the polaron is approximated by an effective non-linear partial differential equation due to Landau and Pekar, in which the phonon field is treated as a classical field.
\end{abstract}

%\subjclass[2000]{Primary 35P15, Secondary 35J10, 47F05}

\maketitle

\renewcommand{\thefootnote}{${}$} \footnotetext{\copyright\, 2015 by
  the authors. This paper may be reproduced, in its entirety, for
  non-commercial purposes.}

\section{Introduction and Main result}

\subsection{Setting of the problem}

In this paper we are interested in the dynamics of a strongly coupled polaron. A polaron is a model of an electron in an ionic lattice interacting with its surrounding polarization field. In 1937 Fr\"ohlich \cite{Fr} proposed a quantum-mechanical Hamiltonian, given in \eqref{eq:frohlich} below, in order to describe the dynamics of a polaron. In this model the phonon field is treated as a quantum field. The Fr\"ohlich Hamiltonian depends on a single parameter $\alpha>0$ which describes the strength of the coupling between the electron and the phonon field. In 1948 Landau and Pekar \cite{LaPe} proposed a system of non-linear PDEs, see \eqref{eq:defParticle}, \eqref{eq:defField} below, to describe the dynamics of a polaron and used this in their famous computation of the effective polaron mass (see \cite{Spohn1987} for an alternative approach). They treat the phonons as a classical field. The derivation of their equations is phenomenological and they do not comment on the relation between their equations and the dynamics generated by Fr\"ohlich's Hamiltonian. Our purpose in this paper is to establish a connection between the two dynamics and to rigorously derive the Landau--Pekar equations from the Fr\"ohlich dynamics in the strong coupling limit $\alpha\to\infty$ for a natural class of initial conditions and on certain time scales.

In order to describe this result in detail, we recall that the Fr\"ohlich Hamiltonian acts in $\mathcal{L}^2(\mathbb{R}^3)\otimes \mathcal{F}$, where $\mathcal{L}^2(\mathbb{R}^3)$ corresponds to the electron and $\mathcal{F}=\mathcal{F}(\mathcal{L}^2(\mathbb{R}^3))$, the bosonic Fock space over $\mathcal{L}^2(\mathbb{R}^3)$, corresponds to the phonon field. The Hamiltonian is given by
\begin{align}\label{eq:frohlich}
p^2+\sqrt{\alpha} \int_{\mathbb{R}^3} \left[e^{-ik\cdot x}a_k+e^{ik\cdot x}a_k^*\right] \frac{dk}{|k|}  + \int_{\mathbb{R}^3} a_k^* a_k \,dk \,,
\end{align} where $p:=-i\nabla_x$ and $x$ are momentum and position of the electron and $a_k$ and $a_k^*$ are annihilation and creation operators in $\mathcal{F}$ satisfying the commutation relations
\begin{align}
[a_{k},\ a_{k^{'}}^*]=\delta(k-k^{'}),\ [a_k,\ a_{k^{'}}]=0,\ \text{and} \ [a_{k}^{*},\ a_{k^{'}}^*]=0 \,.
\end{align} 
As mentioned before, the scalar $\alpha>0$ describes the strength of the coupling between the electron and the phonon field and will be large in our study.

To facilitate later discussions we rescale the variables, as in \cite{FrankSchlein2013},
\begin{align}
x\mapsto \alpha^{-1} x\,,\qquad k\mapsto \alpha k \,,
\end{align}
and find that the Hamiltonian in \eqref{eq:frohlich} is unitarily equivalent to $\alpha^2 \tilde{H}_{\alpha}^{F}$, where the new Hamiltonian $\tilde{H}_{\alpha}^{F}$, acting again in $\mathcal{L}^2(\mathbb{R}^3)\otimes \mathcal{F}$, is defined as
\begin{align}
\tilde{H}_{\alpha}^{F}:=p^2+\int_{\mathbb{R}^3} \left[e^{-ik\cdot x}b_k+e^{ik\cdot x}b_k^*\right] \frac{dk}{|k|} +\int_{\mathbb{R}^3} b_k^* b_k \,dk \,.
\end{align}
The new annihilation and creation operators $b_k:=\alpha^{\frac{1}{2}} a_{\alpha k}$ and $b_k^*:=\alpha^{\frac{1}{2}} b^*_{\alpha k}$ satisfy the commutation relations
\begin{align}\label{eq:commB}
[b_{k},\ b_{k^{'}}^*]=\alpha^{-2}\delta(k-k^{'}),\ [b_k,\ b_{k^{'}}]=0,\ \text{and} \ [b_{k}^{*},\ b_{k^{'}}^*]=0 \,.
\end{align}
We emphasize the $\alpha$-dependence in \eqref{eq:commB}.

We will discuss the dynamics generated by $\tilde{H}_{\alpha}^{F}$ for initial conditions of the product form
\begin{align}
\psi_0\otimes W(\alpha^2\phi_0)\Omega \,.
\end{align}
Here, $\Omega$ denotes the vacuum in $\mathcal{F}$ and $W(f)$ denotes the Weyl operator,
\begin{align}\label{eq:defW}
W(f):=e^{b^*(f)-b(f)} \,,
\end{align}
so that $W(\alpha^2\phi)\Omega$ is a coherent state. This particular choice of initial conditions is motivated by Pekar's approximation \cite{pekar1946polaron,pekar1954polaron} to the ground state energy, which uses exactly states of this form. Pekar's approximation was made mathematically rigorous by Donsker and Varadhan \cite{MR709647} (see \cite{MR1462224} for an alternative approach).

Clearly, the time-evolved state $e^{-i \tilde{H}_{\alpha}^{F}t}\psi_0\otimes W(\alpha^2\phi_0)\Omega$ with $t\neq 0$ will in general no longer have an exact product structure. However, we will see that for large $\alpha$ (and $t$ of order one, or even larger) it can be approximated, in a certain sense, by a state of the product form $\psi_t\otimes W(\alpha^2\phi_t)\Omega$, where $\psi_t$ and $\phi_t$ solve the Landau--Pekar equations
\begin{align}
\label{eq:defParticle}
i\partial_t \psi_t(x) & = \left[ -\Delta + \int_{\R^3} \left[ e^{-ik\cdot x} \phi_t(k)+e^{ik\cdot x} \overline{\phi_t(k)} \right] \frac{dk}{|k|} \right] \psi_t(x) \,, \\
\label{eq:defField}
i\alpha^2 \partial_t \phi_t(k) & = \phi_t(k) + |k|^{-1} \int_{\R^3} |\psi_t(x)|^2 e^{ik\cdot x} \,dx
\end{align}
with initial data $\psi_0$ and $\phi_0$. Using standard methods one can show that for any $\psi_0\in\mathcal H^1(\R^3)$, $\phi_0\in \mathcal L^2(\R^3)$ and $\alpha>0$ the system \eqref{eq:defParticle}, \eqref{eq:defField} has a global solution $(\psi_t,\phi_t)$, which satisfies
$$
\|\psi_t\|_{\mathcal L^2(\R^3)} = \|\psi_0\|_{\mathcal L^2(\R^3)}
\qquad\text{and}\qquad
\mathcal E(\psi_t,\phi_t)=\mathcal E(\psi_0,\phi_0)
\qquad\text{for all}\ t\in\R
$$
with the energy
\begin{align}
\label{eq:energy}
\mathcal E(\psi,\phi) := & \int_{\R^3} |\nabla\psi|^2\,dx + \int_{\R^3} |\psi(x)|^2 \int_{\R^3} \left( e^{-ik\cdot x} \phi(k) + e^{ik\cdot x} \bar\phi(k)\right)\frac{dk}{|k|} \,dx \nonumber \\
& + \int_{\R^3} |\phi(k)|^2\,dk \,.
\end{align}
We refer to Lemma \ref{wellposedenergy} and Proposition \ref{THM:wellposedness} for more details about the solution $(\psi_t,\phi_t)$. In the original work of Landau and Pekar the equations are given in a different, but equivalent form, and we explain this connection in Subsection \ref{sec:equiv}.

%%%%%%%%%%%%%%%%%%%%%%%%%%

\subsection{Main result}

In order to prove our main result we need the following regularity and decay assumptions on the initial data. We denote by $\mathcal H^m(\R^3)$ the Sobolev space of order $m$ and by
\begin{equation}
\label{eq:weighted}
\mathcal L^2_{(m)}(\R^3) := \mathcal L^2(\R^3,(1+k^2)^m\,dk)
\end{equation}
a weighted $\mathcal L^2$ space with norm
$$
\|\phi\|_{\mathcal L^2_{(m)}} = \left( \int_{\R^3} (1+k^2)^m |\phi(k)|^2\,dk \right)^{1/2} \,.
$$
Our main result will be valid under

\begin{assumption}\label{ass}
We assume $\psi_0\in\mathcal H^4(\R^3)$ and $\phi_0\in \mathcal L^2_{(3)}(\R^3)$ with $\|\psi_0\|_{\mathcal{L}^2(\mathbb{R}^3)}=1$.
\end{assumption}

A first version of our main result concerns the approximation of the reduced density matrices of $e^{-i \tilde{H}_{\alpha}^{F}t}\psi_0\otimes W(\alpha^2\phi_0)\Omega$ in the trace norm.

\begin{theorem}\label{reduceddensity}
Assume that $\psi_0$ and $\phi_0$ satisfy Assumption \ref{ass} and let $(\psi_t,\phi_t)$ be the solution of \eqref{eq:defParticle}, \eqref{eq:defField} with inital condition $(\psi_0,\phi_0)$. Define
\begin{align*}
\gamma_t^\mathrm{particle} & := \tr_\mathcal F \left| e^{-i \tilde{H}_{\alpha}^{F}t}\psi_0\otimes W(\alpha^2\phi_0)\Omega\right\rangle\left\langle e^{-i \tilde{H}_{\alpha}^{F}t}\psi_0\otimes W(\alpha^2\phi_0)\Omega \right|, \\
\gamma_t^\mathrm{field} & := \tr_{\mathcal{L}^2(\mathbb{R}^3)} \left| e^{-i \tilde{H}_{\alpha}^{F}t}\psi_0\otimes W(\alpha^2\phi_0)\Omega\right\rangle\left\langle e^{-i \tilde{H}_{\alpha}^{F}t}\psi_0\otimes W(\alpha^2\phi_0)\Omega \right|.
\end{align*}
Then, for all $\alpha\geq 1$ and all $t\in [-\alpha,\alpha]$,
\begin{align*}
\tr_{\mathcal{L}^2(\mathbb{R}^3)} \left| \gamma_t^\mathrm{particle} - |\psi_t\rangle\langle\psi_t| \right| & \leq C \alpha^{-2}(1+t^2) \,, \\
\tr_\mathcal F \left| \gamma_t^\mathrm{field} - |W(\alpha^2\phi_t)\Omega\rangle\langle W(\alpha^2\phi_t)\Omega| \right| & \leq C \alpha^{-2} (1+t^2) \,.
\end{align*}
\end{theorem}

Note that $\gamma_t^\mathrm{particle}$, $\gamma_t^\mathrm{field}$, $|\psi_t\rangle\langle\psi_t|$ and $|W(\alpha^2\phi_t)\Omega\rangle\langle W(\alpha^2\phi_t)\Omega|$ all have trace norm equal to one (in fact, they are non-negative operators with trace one) and therefore Theorem \ref{reduceddensity} gives a non-trivial approximation up to times $t=o(\alpha)$. Already the approximation up to times of order one is significant since this is the time scale on which $\psi_t$ changes. It is a bonus that the same approximation is in fact valid for much longer times.

We emphasize that the Landau--Pekar approximation to the Fr\"ohlich dynamics depends on $\alpha$ (through \eqref{eq:defField}). As we will explain in Subsection \ref{sec:comp}, without allowing for an $\alpha$-dependence one can not approximate $\gamma_t^\mathrm{particle}$ with accuracy $\alpha^{-2}$ for times of order one.

We next present a more precise result which comes at the expense of a more complicated formulation. We approximate the state 
$e^{-i \tilde{H}_{\alpha}^{F}t}\psi_0\otimes W(\alpha^2\phi_0)\Omega$ itself in $\mathcal L^2(\R^3)\otimes\mathcal F$, and not only its reduced density matrices. However, it turns out that up to the desired order $\alpha^{-2}$ this is \emph{not} possible in terms of simple product states. Instead, we need to include an explicit non-product state of order $\alpha^{-1}$ which takes correlations between the particle and the field into account. The key observation is that this term satisfies an almost orthogonality condition, so that it does not contribute to the reduced density matrices to order $\alpha^{-1}$. For the statement we need the real scalar function $\omega$ defined as
\begin{align}
\label{eq:defC1}
\omega(t):=\alpha^2\im(\phi_t,\partial_t\phi_t) + \|\phi_t\|^2 \,.
\end{align}
It will follow from Lemma \ref{wellposedenergy} below that this function is uniformly bounded in $t\in\R$.

The following is our main result.

\begin{theorem}\label{THM:main}
Assume that $\psi_0$ and $\phi_0$ satisfy Assumption \ref{ass} and let $(\psi_t,\phi_t)$ be the solution of \eqref{eq:defParticle}, \eqref{eq:defField} with initial condition $(\psi_0,\phi_0)$. Then there is a decomposition
\begin{align}
e^{-i\tilde{H}_{\alpha}^{F}t}\psi_0\otimes W(\alpha^2\phi_0)\Omega=e^{-i\int_0^t \omega(s)\,ds} \psi_t\otimes W(\alpha^2 \phi_t)\Omega+R(t)
\label{eq:defD}
\end{align}
and a constant $C>0$ such that for all $\alpha\geq 1$ and all $t\in [-\alpha,\alpha]$,
\begin{align}
\left\|\left\langle \Omega, \ W^*(\alpha^2\phi_t)R(t)\right\rangle_{\mathcal{F}}\right\|_{\mathcal{L}^2(\mathbb{R}^3)}& \leq C\alpha^{-2}|t|\left(1+|t|\right) \label{eq:true1}\\
\left\|\left\langle \psi_t,\ W^*(\alpha^2\phi_t)R(t)\right\rangle_{\mathcal{L}^2(\mathbb{R}^3)}\right\|_{\mathcal{F}}& \leq C\alpha^{-2}|t|\left(1+|t|\right)\label{eq:true2}
\end{align}
and
\begin{align}\label{eq:diff}
\left\|R(t)\right\|_{\mathcal{L}^2(\mathbb{R}^3)\otimes\mathcal{F}}\leq C\alpha^{-1}\left(1+|t|\right) \,.
\end{align}
More precisely, \eqref{eq:defD} holds with $R(t)=R_1(t)+R_2(t)$ and with the following bounds
\begin{align}
\left\|\left\langle \Omega, \ W^*(\alpha^2\phi_t)R_1(t)\right\rangle_{\mathcal{F}}\right\|_{\mathcal{L}^2(\mathbb{R}^3)} & \leq C\alpha^{-2}t^2 \label{eq:true1rem}\\
\left\|\left\langle \psi_t,\ W^*(\alpha^2\phi_t)R_1(t)\right\rangle_{\mathcal{L}^2(\mathbb{R}^3)}\right\|_{\mathcal{F}} & \leq C\alpha^{-2}t^2 \label{eq:true2rem}
\end{align}
and
\begin{align}\label{eq:diffrem}
\left\|R_2(t)\right\|_{\mathcal{L}^2(\mathbb{R}^3)\otimes\mathcal{F}}\leq C\alpha^{-2}|t|\left(1+|t|\right) \,,\qquad
\left\|R_1(t)\right\|_{\mathcal{L}^2(\mathbb{R}^3)\otimes\mathcal{F}}\leq C\alpha^{-1}\left(1+|t|\right) \,.
\end{align}
\end{theorem}

Similarly as before we note that for $t=o(\alpha)$ the term $R(t)$ is of lower order than the main term $e^{-i\int_0^t \omega(s)\,ds} \psi_t\otimes W(\alpha^2 \phi_t)\Omega$, which has constant norm equal to one.

The message of Theorem \ref{THM:main} is that, while $R(t)$ is in general not of order $\alpha^{-2}$ (for times of order one), it can be split into a piece which is, namely $R_2(t)$, and a piece which satisfies almost orthogonality conditions, so that it does not contribute to the reduced particle or field density matrices at order $\alpha^{-1}$ either. The term $R_1(t)$ is given explicitly in \eqref{eq:psi1} below. 

Theorem \ref{THM:main} implies Theorem \ref{reduceddensity} by a simple abstract argument, which we explain in Appendix \ref{sec:reducedabstract}. In the following we concentrate on proving Theorem \ref{THM:main}.

In Subsection \ref{sec:comp} we compare Theorem \ref{THM:main} with a similar approximation in \cite{FrankSchlein2013} where $\phi_t$ is independent of $t$. In Lemma \ref{opt} we show that this simpler approximation does not yield the same accuracy in terms of powers of $\alpha^{-1}$ as Theorem \ref{THM:main}. In this sense Theorem \ref{THM:main} derives the Landau--Pekar dynamics from the Fr\"ohlich dynamics and answers an open question in \cite{FrankSchlein2013}.

While it is necessary to take the time dependence of $\phi_t$ into account, this dependence is still weak for times of order $\alpha$ as considered in our theorems. The field $\phi_t$ changes by order one only on times of order $\alpha^2$, and it would be desirable to extend Theorems \ref{reduceddensity} and \ref{THM:main} to this time scale, at least for a certain class of initial conditions. This remains an open problem.

The almost orthogonality relations \eqref{eq:true1} and \eqref{eq:true2} clearly play an important role in our proof. Let us discuss their origin in more detail. We introduce the function
\begin{equation}
\label{eq:tildepsi}
\tilde\psi_t := e^{-i\int_0^t\omega(s)\,ds}\,\psi_t
\end{equation}
and consider the problem of approximating $e^{-i\tilde{H}_{\alpha}^{F}t}\psi_0\otimes W(\alpha^2\phi_0)\Omega$ by a function of the form $\tilde\psi_t\otimes W(\alpha^2 \phi_t)\Omega$. (We do \emph{not} assume at this point that $\tilde\psi_t$ and $\phi_t$ satisfy an equation.) Since $W(\alpha^2\phi_t)$ is unitary, this is the same as the problem of choosing $\tilde\psi_t$ and $\phi_t$ so as to minimize the norm of the vector
\begin{align}
W^*(\alpha^2\phi_t) e^{-i\tilde{H}_{\alpha}^{F}t}\psi_0\otimes W(\alpha^2\phi_0)\Omega - \tilde\psi_t\otimes \Omega\,.\label{eq:sec}
\end{align}
Clearly, for given $\psi_0$, $\phi_0$ and $\phi_t$, the optimal choice for $\tilde\psi_t$ is
\begin{align}\label{eq:trial1}
\tilde\psi_t=\left\langle \Omega, \ W^*(\alpha^2\phi_t)e^{-i\tilde{H}_{\alpha}^{F}t} \psi_0\otimes W(\alpha^2\phi_0)\Omega\right\rangle_{\mathcal{F}}.
\end{align}
In order to determine $\phi_t$ we only solve the simpler problem of minimizing the norm of the projection of \eqref{eq:sec} onto the subspace $\mathrm{span} \{\tilde\psi_t\}\otimes\mathcal F$. This norm could be made zero if we could achieve
\begin{align}\label{eq:trial2}
\Omega = \left\langle \tilde\psi_t,\ W^*(\alpha^2\phi_t)e^{-i\tilde{H}_{\alpha}^{F}t} \psi_0\otimes W(\alpha^2\phi_0)\Omega \right\rangle_{\mathcal{L}^2}.
\end{align}
While it may not be possible to have exact equalities in \eqref{eq:trial1} and \eqref{eq:trial2}, we will see that the Landau--Pekar equations yield almost equalities. In fact, the almost orthogonality relations \eqref{eq:true1} and \eqref{eq:true2} in our main theorem state exactly that
\begin{align}
\tilde\psi_t-\left\langle \Omega, \ W^*(\alpha^2\phi_t)e^{-i\tilde{H}_{\alpha}^{F}t} \psi_0\otimes W(\alpha^2\phi_0)\Omega\right\rangle_{\mathcal{F}}=&\, O_{\mathcal L^2}\left(\alpha^{-2}|t|\left(1+|t|\right)\right)\label{eq:true11}\\
\Omega - \left\langle \tilde\psi_t,\ W^*(\alpha^2\phi_t)e^{-i\tilde{H}_{\alpha}^{F}t} \psi_0\otimes W(\alpha^2\phi_0)\Omega\right\rangle_{\mathcal{L}^2}=&\, O_{\mathcal F}\left(\alpha^{-2}|t|\left(1+|t|\right)\right) .\label{eq:true12}
\end{align}

%%%%%%%%%%%%%%%%%%%%%%%%%%%%%%%%%%%%%%%%%%%%%

\subsection{Comparison with earlier results}\label{sec:comp}

The problem of approximating the Fr\"ohlich dynamics of a polaron was studied before in \cite{FrankSchlein2013} and it was shown that
\begin{equation}
\label{eq:fs}
\left\| e^{-i \tilde{H}_{\alpha}^{F}t}\psi_0\otimes W(\alpha^2\phi_0)\Omega - e^{-i\|\phi_0\|_2^2 t} \zeta_t \otimes W(\alpha^2\phi_0) \right\|_{\mathcal L^2\otimes\mathcal F} \leq C \alpha^{-1} |t|^{1/2} e^{C|t|} \,,
\end{equation}
where $\zeta_t$ denotes the solution of the \emph{linear} equation
$$
i\partial_t\zeta_t(x) = \left[-\Delta + \int_{\R^3} \left[ e^{-ikx}\phi_0(k) + e^{ik\cdot x} \overline{\phi_0(k)} \right]\frac{dk}{|k|}\right] \zeta_t(x)
$$
with initial condition $\psi_0$. We stress that in this approximation, $\phi_0$ does not evolve in time.

Our Theorem \ref{THM:main} improves upon this result by exhibiting an approximation which is valid for longer times. Namely, \eqref{eq:diff} says that
$$
\left\| e^{-i \tilde{H}_{\alpha}^{F}t}\psi_0\otimes W(\alpha^2\phi_0)\Omega - e^{-i\int_0^t\omega(s)\,ds} \psi_t \otimes W(\alpha^2\phi_t)\Omega \right\|_{\mathcal L^2\otimes\mathcal F} \leq C \alpha^{-1} \left(1+|t|\right) \,.
$$
(In \cite{FrankSchlein2013} weaker regularity and decay assumptions are imposed on $\psi_0$ and $\phi_0$, but we emphasize that \eqref{eq:diff} is also valid under weaker assumptions than those in Assumption \ref{ass}. In fact, the latter assumption is needed to bound $R_2(t)$, whereas for \eqref{eq:diff} one can avoid the use of Duhamel's principle in Proposition \ref{decomp}.)

More importantly, even for times of order one the bounds from \cite{FrankSchlein2013} do not allow one to approximate $\gamma_t^\mathrm{field}$ as precisely as in Theorem \ref{reduceddensity}. In fact, \eqref{eq:fs} gives, using \eqref{eq:traceineq} and possibly changing the value of $C$,
\begin{align*}
\tr_{\mathcal L^2} \left| \gamma_t^\mathrm{particle} - |\zeta_t\rangle\langle\zeta_t| \right| & \leq C \alpha^{-1} |t|^{1/2} e^{C|t|} \,, \\
\tr_{\mathcal L^2} \left| \gamma_t^\mathrm{field} - |W(\alpha^2\phi_0)\Omega\rangle\langle W(\alpha^2\phi_0)\Omega| \right| & \leq C \alpha^{-1} |t|^{1/2} e^{C|t|} \,.
\end{align*}
The next result shows that in the approximation of $\gamma_t^\mathrm{field}$ the order $\alpha^{-1}$ (for times of order one) cannot be improved in general. 
In contrast, Theorem \ref{reduceddensity} says that $|W(\alpha^2\phi_t)\Omega\rangle\langle W(\alpha^2\phi_t)\Omega|$ provides an approximation to order $\alpha^{-2}$. This gain of a factor of $\alpha^{-1}$ comes from the time dependence of $\phi_t$ through the Landau--Pekar equations.

\begin{lemma}\label{opt}
In addition to Assumption \ref{ass} suppose that $\phi_0\not\equiv-\sigma_{\psi_0}$ in the notation \eqref{eq:forcing}. Then there are $\epsilon>0$, $C>0$ and $c>0$ such that for all $|t|\in[C\alpha^{-1},\epsilon]$ and all $\alpha\geq C/\epsilon$,
$$
\tr_{\mathcal F} \left|\gamma^\mathrm{field} - |W(\alpha^2\phi_0)\Omega\rangle\langle W(\alpha^2\phi_0\Omega| \right| \geq c \alpha^{-1} |t| \,.
$$
\end{lemma}

Since Theorem \ref{reduceddensity} is a consequence of Theorem \ref{THM:main} and since we showed that one cannot replace $\phi_t$ by $\phi_0$ in Theorem \ref{reduceddensity}, the same applies to Theorem \ref{THM:main}.

\medskip

Let us consider our problem from a wider perspective. We have a composite quantum system $\mathcal H_1\otimes\mathcal H_2$ and a Hamiltonian which couples the two subsystems. Each system has an effective `Planck constant' and the characteristic feature of the problem is that the Planck constant of one system goes to zero, whereas that of the other system remains fixed. Thus, one of the system becomes classical, whereas the other one remains quantum-mechanical, and Ginibre, Nironi and Velo \cite{GiNiVe} used the term `partially classical limit' in a closely related context. (For us, the `Planck constant' of the phonons is $\alpha^{-2}$, as can be seen from the commutation relations, whereas that of the electron is of order one.) A prime example of such a problem is the Born--Oppenheimer approximation, where the inverse square root of the nuclear mass plays the role of the small Planck constant.

Here, however, we consider the case where $\mathcal H_1\otimes\mathcal H_2$ has infinitely many degrees of freedom. As is well known, our Hamiltonian is the Wick quantization of an energy functional on an infinite-dimensional phase space and the notion of `Planck constant' has a well-defined meaning through the commutation relations of the fields. (We emphasize that in our problem we can imagine that we have also a field $\Psi$ for the electrons, but that we only consider the sector of a single electron.) 

Although there is an enormous literature concerning the classical limit, starting with Hepp's work \cite{He}, and although we believe that the question of a partially classical limit is a very natural one which appears in many models, we are only aware of the single work \cite{GiNiVe} prior to \cite{FrankSchlein2013} on this question. The paper \cite{GiNiVe} studies fluctuation dynamics. Closer to our focus here are the works \cite{Fa,AmFa} about the Nelson model with a cut-off where, however, a classical limit on \emph{both} systems is taken. On the level of results one obtains equations similar to the Landau--Pekar equations (without the factor $\alpha^2$ in \eqref{eq:defField}), but the proofs are completely different, as \cite{AmFa} relies on the Wigner measure approach from \cite{AmNi1,AmNi2}. The polaron model, in contrast to the Nelson model, does not require a cut-off, although this is not obvious since the operator $\int e^{ik\cdot x} b_k |k|^{-1}\,dk$ and its adjoint are not bounded relative to the number operator. Lieb and Yamazaki \cite{LiYa} devised a method to deal with this problem in the stationary case, but it is not clear to us how to apply their argument in a dynamical setting and we consider our solution of this problem as a technical novelty in this paper. Our methods apply equally well to a partially classical limit in the cut-off Nelson model and, in fact, the proofs in that case would be considerably shorter.

%%%%%%%%%%%%%%%%%%%%%%%%%%%%%%%%%%

\subsection{An equivalent form of the Landau--Pekar equations}\label{sec:equiv}

Often the Landau--Pekar equations are stated in the form
\begin{align}
\label{eq:defParticleequiv}
i\partial_t\psi_t & = \left( -\Delta + |x|^{-1}*P_t \right) \psi_t \,,\\
\label{eq:defFieldequiv}
\alpha^4 \partial_t^2 P_t & = - P_t - (2\pi)^2 |\psi_t|^2
\end{align}
for a real-valued polarization field $P_t$; see, e.g., \cite{LaPe,PolaronReview}. Let us show that this pair of equations is equivalent to the pair of equations that we discussed so far. In fact, assume that $\psi_t$ and $\phi_t$ solve \eqref{eq:defParticle} and \eqref{eq:defField} and define
$$
P_t(x) := (2\pi)^{-1} \re \int_{\R^3} |k| \phi_t(k) e^{-ik\cdot x}\,dk \,,
$$
as well as the auxiliary function
$$
Q_t(x) := (2\pi)^{-1} \im \int_{\R^3} |k| \phi_t(k) e^{-ik\cdot x}\,dk \,.
$$
If we multiply \eqref{eq:defField} by $|k|$ and integrate with respect to $e^{-ik\cdot x}$, we obtain
$$
i\alpha^2\partial_t (P_t+iQ_t) = P_t+iQ_t + (2\pi)^2 |\psi_t|^2 \,.
$$
Since $P_t$ and $Q_t$ are real, this equation is equivalent to the pair of equations
$$
\alpha^2\partial_t P_t = Q_t \,,
\qquad
\alpha^2\partial_t Q_t = - P_t- (2\pi)^2 |\psi_t|^2 \,. 
$$
Here we can eliminate $Q_t$ by differentiating the first equation and arrive at \eqref{eq:defFieldequiv}.

Moreover, the inversion formula
$$
\phi_t(k) = (2\pi)^{-2} |k|^{-1} \int_{\R^3} (P_t+iQ_t) e^{ik\cdot x}\,dx
$$
implies
$$
\int_{\R^3} \left( e^{-ik\cdot x} \phi_t(k) + e^{ik\cdot x} \overline{\phi_t(k)} \right) \frac{dk}{|k|}
= |x|^{-1} * P_t \,,
$$
which yields \eqref{eq:defParticleequiv}.

\subsection*{Acknowledgements}

The authors are grateful to J. Fr\"ohlich, M. Lewin, B. Schlein and R. Seiringer for their helpful remarks at various stages of this project. Support through NSF grants PHY--1347399 and DMS--1363432 (R.L.F.) and DMS--1308985 and DMS--1443225 (Z.G.) is acknowledged.

%%%%%%%%%%%%%%%%%%%%%%%%%%%%%%%%%%%%%%%%%%%%%%%
%%%%%%%%%%%%%%%%%%%%%%%%%%%%%%%%%%%%%%%%%%%%%%%

\section{Outline of the proof}

\subsection{Well-posedness of the Landau--Pekar equations}

We begin by discussing the well-posedness of the equations for $\psi_t$ and $\phi_t$ in \eqref{eq:defParticle} and \eqref{eq:defField}. We use the following abbreviations for the coupling terms in these equations,
\begin{equation}
\label{eq:effpot}
V_\phi(x) := \int_{\R^3} \left[ e^{-ik\cdot x} \phi(k)+e^{ik\cdot x} \overline{\phi(k)} \right] \frac{dk}{|k|}
\end{equation}
and
\begin{equation}
\label{eq:forcing}
\sigma_\psi(k) := |k|^{-1} \int_{\R^3} |\psi_t(x)|^2 e^{ik\cdot x} \,dx \,.
\end{equation}
The following lemma, which is proved in Appendix \ref{sec:wellposedness}, states global well-posedness in the energy space $\mathcal H^1(\R^3)\times \mathcal L^2(\R^3)$.

\begin{lemma}\label{wellposedenergy}
For any $(\psi_0,\phi_0)\in \mathcal H^1(\R^3)\times \mathcal L^2(\R^3)$ there is a unique global solution $(\psi_t,\phi_t)$ of \eqref{eq:defParticle}, \eqref{eq:defField}. One has the conservation laws
$$
\|\psi_t\|_{\mathcal L^2}=\|\psi_0\|_{\mathcal L^2} \qquad\text{and}\qquad \mathcal E(\psi_t,\phi_t) = \mathcal E(\psi_0,\phi_0)
\qquad\text{for all}\ t\in\R \,.
$$
Moreover, for all $\alpha>0$ and all $t\in\R$,
\begin{equation}
\label{eq:energycons}
\|\psi_t\|_{\mathcal H^1} \lesssim 1 \,,
\qquad
\|\phi_t\|_{\mathcal L^2} \lesssim 1 \,,
\end{equation}
and
\begin{equation}
\label{eq:slowPhi1}
\|\partial_{t}\phi_t\|_{\mathcal L^2} \lesssim \alpha^{-2} \,,
\qquad
\|\phi_t-\phi_s\|_{\mathcal L^2} \lesssim \alpha^{-2}|t-s| \,,
\qquad
\|\sigma_{\psi_t}\|_{\mathcal L^2} \lesssim 1 \,.
\end{equation}
\end{lemma}

In the proof of our main result we need to go beyond the energy space $\mathcal H^1(\R^3)\times \mathcal L^2(\R^3)$. The following proposition states that if the initial conditions have more regularity and decay then, at least for a certain (long) time interval, we have bounds on the solution in the corresponding spaces. We will also need some bounds on the auxiliary functions $g_{s,t}:\ \mathbb{R}^3\rightarrow \mathbb{C}$ defined by
\begin{align}\label{eq:gts}
g_{s,t}(x):=\int_{\mathbb{R}^3} \left[\overline{\phi_t(k)}-\overline{\phi_s(k)}\right]  e^{ik\cdot x}\, \frac{dk}{|k|}
\end{align}
and $g_s:\R^3\to\mathbb C$ defined by
\begin{align}
\label{eq:gs}
g_s(x) := -\partial_s g_{s,t}(x) = \int_{\mathbb{R}^3} e^{ik\cdot x} \overline{\partial_s\phi_s(k)} \, \frac{dk}{|k|}\,.
\end{align}
The following proposition will also be proved in Appendix \ref{sec:wellposedness}.

\begin{proposition}\label{THM:wellposedness}
Let $\tau>0$. If $(\psi_0,\phi_0)$ satisfies Assumption \ref{ass}, then for all $\alpha>0$ and for all $t, s\in [-\tau\alpha^2,\tau\alpha^2]$ we have
\begin{align}
\|\psi_t\|_{\mathcal{H}^4}\lesssim_\tau 1, \qquad
\|\phi_t\|_{\mathcal L^2_{(3)}} \lesssim_\tau 1 \,.\label{eq:psiH4}
\end{align}
Moreover,
\begin{align}
\label{eq:psiH4der}
\|\partial_t\psi_t\|_{\mathcal H^2} \lesssim_\tau 1 \,,
\qquad
\|\partial_t \sigma_{\psi_t}\|_{\mathcal L^2} \lesssim_\tau 1
\end{align}
and
\begin{align}
\|g_{s,t}\|_{\infty}\lesssim_\tau \alpha^{-2}|t-s|\,,
\qquad
\|g_s\|_\infty \lesssim \alpha^{-2} \,.\label{eq:slowPhi}
\end{align}
\end{proposition}

%%%%%%%%%%%%%%%%%%%%%%%%%%%%%%%%%%%%%%%%%

\subsection{Decomposition of the solution}\label{sec:decomp}

In this subsection we decompose the solution $e^{-i\tilde{H}_{\alpha}^{F}t}\psi_0\otimes W(\alpha^2\phi_0)\Omega$ as claimed in Theorem \ref{THM:main}. In order to state this, we need to introduce some notations.

It will be convenient to work with the function $\tilde\psi_t$ from \eqref{eq:tildepsi}. Clearly, the bounds from Lemma \ref{wellposedenergy} and Proposition \ref{THM:wellposedness} hold for $\tilde\psi_t$ as well. (For the bounds on $\partial_t\tilde\psi_t$ we use the fact that $|\omega(t)|\lesssim 1$ by Lemma \ref{wellposedenergy}.) Moreover, we note that $\tilde\psi_t$ and $\phi_t$ satisfy the modified equations
\begin{align}
\label{eq:defParticlemod}
i\partial_t \tilde\psi_t(x) & = \left[ -\Delta + \int_{\R^3} \left[ e^{-ik\cdot x} \phi_t(k)+e^{ik\cdot x} \bar\phi_t(k) \right] \frac{dk}{|k|} + \omega(t) \right] \tilde\psi_t(x) \,, \\
\label{eq:defFieldmod}
i\alpha^2 \partial_t \phi_t(k) & = \phi_t(k)+ |k|^{-1} \int_{\R^3} |\tilde\psi_t(x)|^2 e^{ik\cdot x} \,dx \,.
\end{align}

Next, we define for $\psi\in\mathcal L^2(\R^3)$ with $\|\psi\|=1$ the orthogonal projections in $\mathcal L^2(\R^3)$
\begin{align*}
P_\psi := |\psi\rangle\langle\psi| \,,
\qquad P_\psi^\bot := 1 - P_\psi = 1- |\psi\rangle\langle\psi| \,.
\end{align*}
The effective Schr\"odinger operator $H_\phi$ in $\mathcal L^2(\R^3)$ is defined by
\begin{equation}
\label{eq:effso}
H_\phi := -\Delta + V_\phi + \int_{\R^3} |\phi(k)|^2\,dk
\end{equation}
with $V_\phi$ from \eqref{eq:effpot}. Moreover, let us introduce the operator
\begin{equation}
\label{eq:tildehphi}
\tilde H_\phi := W^*(\alpha^2 \phi) \tilde H^F_\alpha W(\alpha^2 \phi)
\end{equation}
in $\mathcal L^2(\R^3)\otimes\mathcal F$. Using the commutation relations (see Lemma \ref{LM:commRe}) we find that
\begin{equation}
\label{eq:tildehphialt}
\tilde H_\phi = H_\phi + \int_{\mathbb{R}^3} \left[e^{ik\cdot x}b_k^*+ e^{-ik\cdot x}b_k \right]\frac{dk}{|k|} +\int_{\mathbb{R}^3} \left[\phi(k) b_{k}^*+\bar\phi(k) b_{k} \right] dk +\int_{\mathbb{R}^3} b_k^* b_k \,dk \,.
\end{equation}

Finally, we introduce the vector
\begin{equation}
\label{eq:fts}
F_{t,s} := P_{\tilde\psi_s}^\bot \int_{\R^3} \left( e^{ik\cdot x} W^*(\alpha^2 \phi_t)W(\alpha^2\phi_s) b_k^*\ \tilde\psi_s\otimes\Omega\right) \frac{dk}{|k|}
\end{equation}
and define
$$
D_0 := \int_0^t e^{iH_{\phi_t}s} F_{t,s} \,ds
$$
and
\begin{align*}
D_1:=&\int_0^{t} \int_{0}^{t-s} \int_{\mathbb{R}^3} \left( e^{i\tilde{H}_{\phi_t}(s+s_1)} e^{ik\cdot x}b_k^*\  e^{-iH_{\phi_t}s_1} F_{t,s} \right)\frac{dk}{|k|} \,ds_1\,ds \,, \\
D_2:=&\int_0^{t} \int_{0}^{t-s} \int_{\mathbb{R}^3} \left( e^{i\tilde{H}_{\phi_t}(s+s_1)} e^{-ik\cdot x}b_k\  e^{-iH_{\phi_t}s_1} F_{t,s} \right) \frac{dk}{|k|} \,ds_1\,ds \,, \\
D_3:=&\int_0^{t} \int_{0}^{t-s} \int_{\mathbb{R}^3} \left( e^{i\tilde{H}_{\phi_t}(s+s_1)} \phi_{t}(k) b_{k}^*\  e^{-iH_{\phi_t}s_1} F_{t,s} \right) dk \,ds_1\,ds \,, \\
D_4:=&\int_0^{t} \int_{0}^{t-s} \int_{\mathbb{R}^3} \left( e^{i\tilde{H}_{\phi_t}(s+s_1)} \overline{\phi_{t}(k)} b_{k} \ e^{-iH_{\phi_t}s_1} F_{t,s} \right) dk \,ds_1\,ds \,, \\
D_5:=&\int_0^{t} \int_{0}^{t-s} \int_{\mathbb{R}^3} \left( e^{i\tilde{H}_{\phi_t}(s+s_1)} b_k^* b_k\  e^{-iH_{\phi_t}s_1}  F_{t,s} \right) dk \,ds_1\,ds \,.
\end{align*}
With these notations the promised representation formula for the solution looks as follows.

\begin{proposition}\label{decomp}
Assume that $(\tilde\psi_t,\phi_t)$ satisfy \eqref{eq:defParticlemod}, \eqref{eq:defFieldmod} with initial conditions $(\psi_0,\phi_0)$ where $\|\psi_0\|^2=1$. Then for any $t\in\R$ one has the decomposition
$$
e^{-i\tilde{H}_{\alpha}^{F}t}\psi_0\otimes W(\alpha^2\phi_0)\Omega= \tilde\psi_t\otimes W(\alpha^2 \phi_t)\Omega+R_1(t)+R_2(t)
$$
with
$$
R_1(t) := -i W(\alpha^2\phi_t) e^{-i H_{\phi_t}t} D_0
$$
and
$$
R_2(t) := - W(\alpha^2\phi_t) e^{-i\tilde H_{\phi_t} t} \left( D_1+D_2+D_3+D_4+D_5 \right).
$$
\end{proposition}

Clearly, in terms of the original function $\psi_t$, the term $R_1$ is explicitly given by
\begin{align}
\label{eq:psi1}
R_1(t) = & -i W(\alpha^2 \phi_t) \int_0^{t} \left[ e^{-iH_{\phi_t}(t-s)-i\int_0^s \omega(s_1)\,ds_1} \right. \notag\\
& \qquad\qquad\qquad\quad \left. P_{\psi_s}^\bot \int_{\R^3} \left( e^{ik\cdot x} W^*(\alpha^2 \phi_t)W(\alpha^2\phi_s) b_k^*\ \psi_s\otimes\Omega\right) \frac{dk}{|k|} \right] ds \,.
\end{align}

The proof of Proposition \ref{decomp} makes use of equations \eqref{eq:defParticlemod}, \eqref{eq:defFieldmod} for $(\tilde\psi_t,\phi_t)$ as well as the Duhamel formula. We single out the use of the equations in the following lemma.

\begin{lemma}\label{LM:newform}
Assume that $(\tilde\psi_t,\phi_t)$ satisfy \eqref{eq:defParticlemod}, \eqref{eq:defFieldmod} with initial conditions $(\psi_0,\phi_0)$ where $\|\psi_0\|^2=1$. Then for any $t\in\R$ one has
\begin{align}
e^{-i\tilde H_{\alpha}^{F}t} \psi_0\otimes W(\alpha^2\phi_0)\Omega
=\tilde\psi_t\otimes W(\alpha^2\phi_t) \Omega
-i \int_0^{t} e^{-i\tilde{H}_{\alpha}^{F}(t-s)} W(\alpha^2\phi_t) F_{t,s} \,ds \,.\label{eq:newform}
\end{align}
\end{lemma}

\begin{proof}[Proof of Lemma \ref{LM:newform}]
Applying the operator $e^{i\tilde{H}_{\alpha}^{F}t}$ to both sides of \eqref{eq:newform} we see that we need to prove
\begin{align*}
\psi_0\otimes W(\alpha^2\phi_0)\Omega
=e^{i\tilde{H}_{\alpha}^{F}t}\tilde\psi_t\otimes W(\alpha^2\phi_t)\Omega
-i  \int_0^{t} e^{i\tilde{H}_{\alpha}^{F}s} W(\alpha^2\phi_t) F_{t,s} \,ds \,.
\end{align*}
This is clearly true at $t=0$ and therefore we only need to show that the time derivatives of both sides coincide for all $t$, that is, in view of definition \eqref{eq:fts} of $F_{t,s}$,
\begin{align*}
0=e^{i\tilde{H}_{\alpha}^{F}t} & \left[ i\tilde{H}_{\alpha}^{F} \tilde\psi_t\otimes W(\alpha^2\phi_t)\Omega+\partial_{t}\tilde\psi_t\otimes W(\alpha^2\phi_t)\Omega+\tilde\psi_t\otimes \partial_{t}W(\alpha^2\phi_t)\Omega \right. \\
& \quad \left. - i  W(\alpha^2\phi_t) P_{\tilde\psi_t}^\bot \int_{\R^3} \left( e^{ik\cdot x} b_k^* \tilde\psi_t \otimes\Omega\right) \frac{dk}{|k|} \right] .
\end{align*}
This is, of course, the same as
\begin{align}
\label{eq:tderiv}
& i\tilde{H}_{\alpha}^{F} \tilde\psi_t\otimes W(\alpha^2\phi_t)\Omega+\partial_{t}\tilde\psi_t\otimes W(\alpha^2\phi_t)\Omega+\tilde\psi_t\otimes \partial_{t}W(\alpha^2\phi_t)\Omega \nonumber \\
& \qquad = i  W(\alpha^2\phi_t) P_{\tilde\psi_t}^\bot \int_{\R^3} \left( e^{ik\cdot x} b_k^* \tilde\psi_t \otimes\Omega\right) \frac{dk}{|k|} \,,
\end{align}
which is what we are going to show now.

We begin by rewriting the first term on the left side. Using \eqref{eq:tildehphi} and \eqref{eq:tildehphialt} we obtain
\begin{align*}
& i\tilde{H}_{\alpha}^{F} \tilde\psi_t\otimes W(\alpha^2\phi_t)\Omega \\
& \qquad =i H_{\phi_t}\tilde\psi_t\otimes W(\alpha^2\phi_t)\Omega +\tilde\psi_t\ W(\alpha^2\phi_t)\left[i b^*(\phi_t)+i\int_{\mathbb{R}^3} e^{ik\cdot x} b_k^*\frac{dk}{|k|}\right]\Omega \,.
\end{align*}
In order to rewrite the third term on the left side of \eqref{eq:tderiv} we use the formula for $\partial_{t}W(\alpha^2\phi_t)$ from \eqref{eq:firPW} below and find
\begin{align*}
\tilde\psi_t\otimes \partial_{t}W(\alpha^2\phi_t)\Omega = & i\alpha^2 \left(\im (\phi_t,\partial_t\phi_t) \right) \tilde\psi_t \otimes W(\alpha^2\phi_t) \Omega
+ \alpha^2 \tilde\psi_t \otimes W(\alpha^2\phi_t) \ b^*(\partial_t\phi_t)\Omega \,.
\end{align*}
Thus, recalling the definition of $\omega$ in \eqref{eq:defC1}, we have shown that
\begin{align}
& i\tilde{H}_{\alpha}^{F} \tilde\psi_t\otimes W(\alpha^2\phi_t)\Omega+\partial_{t}\tilde\psi_t\otimes W(\alpha^2\phi_t)\Omega+\tilde\psi_t\otimes \partial_{t}W(\alpha^2\phi_t)\Omega \nonumber \\
& \qquad = \left[\partial_t+i\left( -\Delta + V_{\phi_t} + \omega(t) \right) \right]\tilde\psi_t\otimes W(\alpha^2\phi_t)\Omega\label{eq:particle}\\
& \qquad \quad + W(\alpha^2\phi_t)\left[\alpha^2 b^*(\partial_t\phi_t)+i b^*(\phi_t)+i\int_{\mathbb{R}^3} e^{ik\cdot x} b_k^* \,\frac{dk}{|k|} \right]\left(\tilde \psi_t\otimes \Omega\right) .\label{eq:field}
\end{align}
At this point in the proof we use the equations for $\tilde\psi_t$ and $\phi_t$. It follows from \eqref{eq:defParticlemod} that line \eqref{eq:particle} vanishes identically. For line \eqref{eq:field} we use \eqref{eq:defFieldmod} to obtain
\begin{align}
\label{eq:newformproof}
& i\tilde{H}_{\alpha}^{F} \tilde\psi_t\otimes W(\alpha^2\phi_t)\Omega+\partial_{t}\tilde\psi_t\otimes W(\alpha^2\phi_t)\Omega+\tilde\psi_t\otimes \partial_{t}W(\alpha^2\phi_t)\Omega \nonumber \\
& \qquad = i W(\alpha^2\phi_t)\left[
\int_{\R^3} \left( -\int_{\R^3} |\tilde\psi_t(y)|^2 e^{ik\cdot y} \,dy + e^{ik\cdot x} \right) b_k^* \,\frac{dk}{|k|}
\right] \left( \tilde \psi_t\otimes \Omega\right) \nonumber \\
& \qquad = i W(\alpha^2\phi_t)  P_{\tilde\psi_t}^\bot \int_{\mathbb{R}^3} \left(e^{ik\cdot x} b_k^* \tilde\psi_t\otimes \Omega \right) \frac{dk}{|k|} \,.
\end{align}
Here we used the fact that $\|\tilde\psi_t\|=\|\psi_0\|=1$ by assumption and Lemma \ref{wellposedenergy}, and therefore
$$
P_{\tilde\psi_t}^\bot = 1- |\tilde\psi_t\rangle\langle\tilde\psi_t| \,.
$$
Equation \eqref{eq:newformproof} proves \eqref{eq:tderiv} and completes the proof.
\end{proof}

Having proved Lemma \ref{LM:newform} we turn to the proof of Proposition \ref{decomp}.

\begin{proof}[Proof of Proposition \ref{decomp}]
It follows from Lemma \ref{LM:newform} and \eqref{eq:tildehphi} that
$$
e^{-i\tilde H_{\alpha}^{F}t} \psi_0\otimes W(\alpha^2\phi_0)\Omega
=\tilde\psi_t\otimes W(\alpha^2\phi_t) \Omega
-i W(\alpha^2\phi_t) \int_0^{t} e^{-i\tilde{H}_{\phi_t}(t-s)} F_{t,s} \,ds \,.
$$
In the time integral on the right side we use Duhamel's principle and \eqref{eq:tildehphialt},
\begin{align*}
e^{-i\tilde{H}_{\phi_t}(t-s)} = & e^{-i H_{\phi_t}(t-s)} \\
& -i \int_0^{t-s} e^{-i\tilde{H}_{\phi_t}(t-s-s_1)} \left( 
\int_{\mathbb{R}^3} \left[e^{ik\cdot x}b_k^*+ e^{-ik\cdot x}b_k \right]\frac{dk}{|k|}  +\int_{\mathbb{R}^3} b_k^* b_k \,dk \right. \\
& \qquad\qquad\qquad\qquad\qquad\qquad \left. +\int_{\mathbb{R}^3} \left[\phi_{t}(k) b_{k}^*+\bar\phi_{t}(k) b_{k} \right] dk
\right) e^{-iH_{\phi_t} s_1} \,ds_1 \,.
\end{align*}
Proposition \ref{decomp} now follows easily from the definition of $D_0$, $\ldots$, $D_5$.
\end{proof}

%%%%%%%%%%%%%%%%%%%%%%%%%%%%%%%%%%%%%%%%%%%
%%%%%%%%%%%%%%%%%%%%%%%%%%%%%%%%%%%%%%%%%%%

\subsection{Reduction of the proof of the main result}

In the remainder of this paper we will prove the following

\begin{theorem}\label{PRO:D05}
Assume that $\psi_0$ and $\phi_0$ satisfy Assumption \ref{ass}, let $(\tilde\psi_t,\phi_t)$ be the solution of \eqref{eq:defParticlemod}, \eqref{eq:defFieldmod} with initial condition $(\psi_0,\phi_0)$ and let $D_0,\ldots,D_5$ be as in Proposition \ref{decomp}. Then there is a constant $C>0$ such that for all $\alpha\geq 1$ and $t\in [0,\alpha^2]$
\begin{align}
\|D_0\|_{\mathcal{L}^2\otimes \mathcal{F}}&\leq C\alpha^{-1}\left(1+t\right), \label{eq:D0}\\
\|D_1\|_{\mathcal{L}^2\otimes \mathcal{F}}&\leq C\alpha^{-2} t\left(1+t\right) \,, \label{eq:D1}\\
\|D_2\|_{\mathcal{L}^2\otimes \mathcal{F}}&\leq \alpha^{-2} t\left( 1+ t\right) \left(1+\alpha^{-1}t\right), \label{eq:D2}\\
\|D_3\|_{\mathcal{L}^2\otimes \mathcal{F}}&\leq C\alpha^{-2}t\left(1+t\right)\left(1+\alpha^{-1} t\right), \label{eq:D3}\\
\|D_4\|_{\mathcal{L}^2\otimes \mathcal{F}}&\leq C\alpha^{-2}t^2\left(1+\alpha^{-1}t\right), \label{eq:D4}\\
\|D_5\|_{\mathcal{L}^2\otimes \mathcal{F}}&\leq C\alpha^{-3}t\left(1+t\right)\left(1+\alpha^{-2} t^2\right), \label{eq:D5}\\
\left\| \left\langle\Omega, \  e^{-iH_{\phi_t}t} D_0\right\rangle_{\mathcal{F}} \right\|_{\mathcal L^2(\R^3)}& \leq C\alpha^{-2} t^2 \,, \label{eq:vaToOne} \\
\left\| \left\langle\tilde\psi_t, \  e^{-iH_{\phi_t}t} D_0\right\rangle_{\mathcal L^2(\R^3)} \right\|_{\mathcal{F}}& \leq C\alpha^{-2} t^2 \left(1+ \alpha^{-2} t^2 \right). \label{eq:psiPerpprop}
\end{align}
\end{theorem}

This theorem (and its analogue for $t\in[-\alpha^2,0]$), together with the decomposition from Proposition \ref{decomp} and the fact that the operators $W(\alpha^2\phi_t)$, $e^{-iH_{\phi_t}t}$ and $e^{-i\tilde H_{\phi_t} t}$ are unitary, implies Theorem \ref{THM:main}. In fact, \eqref{eq:D0} implies the second bound in \eqref{eq:diffrem}, \eqref{eq:D1}--\eqref{eq:D5} imply the first bound in \eqref{eq:diffrem}, \eqref{eq:vaToOne} implies \eqref{eq:true1rem} and \eqref{eq:psiPerpprop} implies \eqref{eq:true2rem}.

We emphasize that Theorem \ref{PRO:D05} is valid up to times $\alpha^2$. (In fact, since the proof only relies on Proposition \ref{THM:wellposedness}, it is valid up to times $\tau\alpha^2$ for an arbitrary $\tau>0$ with $C$ depending on $\tau$.) Consequently, the bounds in Theorem \ref{THM:main} are also valid up to times $\alpha^2$. However, since the evolved state and the main term in the approximation have both norm one, the bounds are only meaningful for times up to $\epsilon\alpha$ for some small $\epsilon>0$.

The basic intuition behind the bounds on $D_k$, $k=0,\ldots,5$, is that each annihilation or creation operator is of order $\alpha^{-1}$ and therefore $D_0$, which contains only one creation operator, is of order $\alpha^{-1}$, $D_1,D_2,D_3,D_4$, which contain two creation or annihilation operators, are of order $\alpha^{-2}$ and $D_5$, which contains three creation or annihilation operators, is of order $\alpha^{-3}$. We illustrate this intuition in more detail in Subsection \ref{sec:warmup} with the simplest possible terms.

While this basic principle is true, it is oversimplifying the situation considerably as is does not take the slow-decaying terms $|k|^{-1}$ into account. The operator $\int e^{ik\cdot x}b_k^* |k|^{-1}dk$ and its adjoint are \emph{not} bounded relative to the number operator $\int b_k^* b_k\,dx$. In fact, the treatment of these operators is the major difficulty that we have to overcome here.

At this point we have reduced the proof of Theorem \ref{THM:main} to the proof of Theorem~\ref{PRO:D05}, and the remainder of the paper is concerned with this. We bound $D_0$ in Section \ref{sec:remainder}, $D_1$ in Section~\ref{sec:D15} and $D_2$ in Section \ref{sec:remainder2}. The terms $D_3$, $D_4$ and $D_5$ which are easier to bound than $D_1$ and $D_2$, are briefly discussed in Section \ref{sec:remainder3}. Finally, the bounds \eqref{eq:vaToOne} and \eqref{eq:psiPerpprop} will be proved in Subsections \ref{sub:true1} and \ref{sub:true2}, respectively.

%%%%%%%%%%%%%%%%%%%%%%%%%%%%%%%%%%%%%%%%%%%%%
%%%%%%%%%%%%%%%%%%%%%%%%%%%%%%%%%%%%%%%%%%%%%

\subsection{A further decomposition}

Using the fact that $P_{\tilde\psi_t}^\bot=1-|\tilde\psi_t\rangle\langle\tilde\psi_t|$ (see the proof of Lemma \ref{LM:newform}), we decompose
$$
F_{t,s} = F_{t,s}^{(1)} - F_{t,s}^{(2)} \,,
$$
where
$$
F_{t,s}^{(1)} := \int_{\R^3} \left( e^{ik\cdot x} W^*(\alpha^2 \phi_t)W(\alpha^2\phi_s) b_k^*\ \tilde\psi_s\otimes\Omega\right) \frac{dk}{|k|} 
$$
and, with the notation $\sigma_{\psi}$ from \eqref{eq:forcing}, 
$$
F_{t,s}^{(2)} := \tilde\psi_s \otimes W^*(\alpha^2\phi_t)W(\alpha^2\phi_s) b^*(\sigma_{\tilde\psi_s})\Omega \,.
$$
Correspondingly, we define
$$
D_k = D_{k1} - D_{k2}
\qquad\text{for}\ k=0,1,2,3,4,5 \,.
$$
In general, the terms $D_{k2}$ are easier to deal with than the terms $D_{k1}$. The reason for this is that $e^{ik\cdot x}|k|^{-1} \not\in L^2(\R^3)$, whereas $\sigma_{\tilde\psi_t}\in L^2(\R^3)$ by Lemma \ref{wellposedenergy}, so the operator $\int e^{ik\cdot x} b_k^* |k|^{-1}\,dk$ in $F_{t,s}^{(1)}$ is harder to control than the operator $b^*(\sigma_{\tilde\psi_s})$ in $F_{t,s}^{(2)}$.

For $k=1,\ldots,5$ both operators $D_{k1}$ and $D_{k2}$ involve an operator $b_k^*$, $b_k$ or $b_k^*b_k$ to the left of $F_{t,s}^{(1)}$ or $F_{t,s}^{(2)}$, which in turn involves an operator $W^*(\alpha^2\phi_t) W(\alpha^2\phi_s)$. We now decompose
$$
D_{kj}= D_{kj1} + D_{kj2}
\qquad\text{for}\ k=1,2,3,4,5
\ \text{and}\ j=1,2 \,,
$$
where $D_{kj1}$ denotes the expression with $b_k$, $b_k^*$ or $b_k^*b_k$ commuted through the operator $W^*(\alpha^2\phi_t) W(\alpha^2\phi_s)$ and $D_{kj2}$ denotes the expression coming from the commutator. To be explicit, we display some exemplary cases,
\begin{align}
\label{eq:d111def}
D_{111} & =\int_0^{t} \int_{0}^{t-s} \int_{\R^3} \int_{\R^3} e^{i\tilde{H}_{\phi_t}(s+s_1)} e^{ik\cdot x} e^{-iH_{\phi_t}s_1} e^{i k'\cdot x} W^*(\alpha^2\phi_t)W(\alpha^2\phi_s) b_k^* b_{k'}^* \notag \\
& \qquad\qquad\qquad\qquad \times\tilde\psi_s\otimes\Omega \,\frac{dk'}{|k'|}\,\frac{dk}{|k|}\,ds_1\,ds \,,\\
\label{eq:d121def}
D_{121} & =\int_0^{t} \int_{0}^{t-s} \int_{\R^3} e^{i\tilde{H}_{\phi_t}(s+s_1)} e^{ik\cdot x} e^{-iH_{\phi_t}s_1}
 W^*(\alpha^2\phi_t)W(\alpha^2\phi_s) b_k^* b^*(\sigma_{\tilde\psi_s}) \notag\\
 &\qquad\qquad\qquad\qquad\times\tilde\psi_s\otimes\Omega \,\frac{dk}{|k|}\,ds_1\,ds \,, \\
\label{eq:d211def}
D_{211} & =\int_0^{t} \int_{0}^{t-s} \int_{\R^3} \int_{\R^3} e^{i\tilde{H}_{\phi_t}(s+s_1)} e^{ik\cdot x} e^{-iH_{\phi_t}s_1} e^{i k'\cdot x} W^*(\alpha^2\phi_t)W(\alpha^2\phi_s) b_k b_{k'}^* \notag\\
&\qquad\qquad\qquad\qquad\times\tilde\psi_s\otimes\Omega \,\frac{dk'}{|k'|}\,\frac{dk}{|k|}\,ds_1\,ds \,,\\
\label{eq:d221def}
D_{221} & =\int_0^{t} \int_{0}^{t-s} \int_{\R^3} e^{i\tilde{H}_{\phi_t}(s+s_1)} e^{ik\cdot x} e^{-iH_{\phi_t}s_1}
 W^*(\alpha^2\phi_t)W(\alpha^2\phi_s) b_k b^*(\sigma_{\tilde\psi_s}) \notag\\
 &\qquad\qquad\qquad\qquad\times\tilde\psi_s\otimes\Omega \,\frac{dk}{|k|}\,ds_1\,ds \,.
\end{align}
The commutator terms can be computed with the help of Corollary \ref{cor:bbb}. Recalling the definition of the function $g_{s,t}$ in \eqref{eq:gts}, we have for instance
\begin{align}
\label{eq:d112def}
D_{112} & = - \int_0^{t} \int_{0}^{t-s} \int_{\R^3} e^{i\tilde{H}_{\phi_t}(s+s_1)} g_{s,t} W^*(\alpha^2\phi_t)W(\alpha^2\phi_s) e^{-iH_{\phi_t}s_1} e^{ik\cdot x} b_{k}^* \tilde\psi_s\otimes\Omega \,\frac{dk}{|k|}\,ds_1\,ds \,, \\
\label{eq:d122def}
D_{122} & = - \int_0^{t} \int_{0}^{t-s} e^{i\tilde{H}_{\phi_t}(s+s_1)} g_{s,t} W^*(\alpha^2\phi_t)W(\alpha^2\phi_s) e^{-iH_{\phi_t}s_1} b^*(\sigma_{\tilde\psi_s}) \tilde\psi_s\otimes\Omega \,ds_1\,ds \,, \\
\label{eq:d212def}
D_{212} & = \int_0^{t} \int_{0}^{t-s} \int_{\R^3} e^{i\tilde{H}_{\phi_t}(s+s_1)} \overline{g_{s,t}} W^*(\alpha^2\phi_t)W(\alpha^2\phi_s) e^{-iH_{\phi_t}s_1} e^{ik\cdot x} b_{k}^* \tilde\psi_s\otimes\Omega \,\frac{dk}{|k|}\,ds_1\,ds \,,\\
\label{eq:d222def}
D_{222} &= \int_0^{t} \int_{0}^{t-s} e^{i\tilde{H}_{\phi_t}(s+s_1)} \overline{g_{s,t}} W^*(\alpha^2\phi_t)W(\alpha^2\phi_s) e^{-iH_{\phi_t}s_1} b^*(\sigma_{\tilde\psi_s}) \tilde\psi_s\otimes\Omega \,ds_1\,ds \,.
\end{align}

%%%%%%%%%%%%%%%%%%%%%%%%%%%%%%%%%%%%%%%%%%%%%

\subsection{Some warm-up bounds}\label{sec:warmup}

In order to prepare for the rather technical sections that follow, we will first focus on the terms that do not include a term of the form $|k|^{-1}$, that is, on the terms $D_{02}$, $D_{32}$, $D_{42}$ and $D_{52}$. We hope that this explains the underlying mechanism of our proof and the intuition that each annihilation or creation operator is of size $\alpha^{-1}$.

%%%%%%%%%%%%%%%%%%%%%%%%%%%%%%%%%%%%%%%%%%%%

\subsection*{Bound on $D_{02}$}

We recall that
$$
D_{02} = \int_0^t \left( e^{iH_{\phi_t}s} \tilde\psi_s \right) \otimes \left( W^*(\alpha^2\phi_t) W(\alpha^2\phi_s) b^*(\sigma_{\tilde\psi_s})\Omega\right) ds
$$
and, therefore, by Lemma \ref{wellposedenergy},
\begin{equation}
\label{eq:d02}
\|D_{02}\|_{\mathcal L^2\otimes\mathcal F} \leq \int_0^t \|\tilde\psi_s\|_2 \|b^*(\sigma_{\tilde\psi_s})\Omega\|_{\mathcal F} \,ds = \alpha^{-1} \int_0^t \|\sigma_{\tilde\psi_s}\|_2 \,ds \lesssim \alpha^{-1} t \,.
\end{equation}

%%%%%%%%%%%%%%%%%%%%%%%%%%%%%%%%%%%%%%%%%%%%

\subsection*{Bound on $D_{32}$}

We have
$$
D_{321} = \int_0^t \int_0^{t-s} e^{i\tilde H_{\phi_t}(s+s_1)} \left( e^{-iH_{\phi_t}s_1} \tilde\psi_s\right) \otimes \left( W^*(\alpha^2\phi_t) W(\alpha^2\phi_s) b^*(\phi_t) b^*(\sigma_{\tilde\psi_s})\Omega\right) ds_1\,ds
$$
and, according to Corollary \ref{cor:bbb},
\begin{align*}
D_{322} & = - \int_0^t \int_0^{t-s} e^{i\tilde H_{\phi_t}(s+s_1)} \left( e^{-iH_{\phi_t}s_1} \tilde\psi_s\right) \\ & \qquad\qquad\qquad\otimes \left( (\phi_t-\phi_s,\phi_t) W^*(\alpha^2\phi_t) W(\alpha^2\phi_s) b^*(\sigma_{\tilde\psi_s})\Omega\right) ds_1\,ds \,.
\end{align*}
By the bounds from Lemma \ref{wellposedenergy} we have
$$
\left\| b^*(\phi_t) b^*(\sigma_{\tilde\psi_s}) \Omega \right\|_{\mathcal F} = \alpha^{-2} \left( \|\phi_t\|_2^2 \|\sigma_{\tilde\psi_s}\|_2^2 + |(\phi_t,\sigma_{\tilde\psi_s})|^2 \right)^{1/2} 
\lesssim \alpha^{-2} \,,
$$
and therefore, using also the conservation of the $\mathcal L^2$-norm of $\tilde\psi_s$,
$$
\|D_{321}\|_{\mathcal L^2\otimes\mathcal F} \lesssim \alpha^{-2} t^2 \,. 
$$
On the other hand, the bounds from Lemma \ref{wellposedenergy} imply
$$
\left\| b^*(\sigma_{\tilde\psi_s})\Omega \right\|_{\mathcal F} = \alpha^{-1} \|\sigma_{\tilde\psi_s}\|_2 \lesssim \alpha^{-1}
\qquad
|(\phi_t-\phi_s,\phi_t)|\lesssim \alpha^{-2}|t-s| \,,
$$
and therefore, using again the conservation of the $\mathcal L^2$-norm of $\tilde\psi_s$,
$$
\|D_{322}\|_{\mathcal L^2\otimes\mathcal F} \lesssim \alpha^{-3} t^3 \,. 
$$
Thus, we have shown that
\begin{equation}
\label{eq:d32}
\| D_{32} \|_{\mathcal L^2\otimes\mathcal F} \lesssim \alpha^{-2} t^2 \left( 1+ \alpha^{-1} t \right) .
\end{equation}

%%%%%%%%%%%%%%%%%%%%%%%%%%%%%%%%%%%%%%%%%%%%

\subsection*{Bound on $D_{42}$}

We have
$$
D_{421} = \int_0^t \int_0^{t-s} e^{i\tilde H_{\phi_t}(s+s_1)} \left( e^{-iH_{\phi_t}s_1} \tilde\psi_s\right) \otimes \left(W^*(\alpha^2\phi_t) W(\alpha^2\phi_s)  b(\phi_t) b^*(\sigma_{\tilde\psi_s})\Omega\right) ds_1\,ds
$$
and, according to Corollary \ref{cor:bbb},
\begin{align*}
D_{422} & = \int_0^t \int_0^{t-s} e^{i\tilde H_{\phi_t}(s+s_1)} \left( e^{-iH_{\phi_t}s_1} \tilde\psi_s\right) \\
 & \qquad\qquad\qquad\otimes \left( (\phi_t,\phi_t-\phi_s) W^*(\alpha^2\phi_t) W(\alpha^2\phi_s) b^*(\sigma_{\tilde\psi_s})\Omega\right) ds_1\,ds \,.
\end{align*}
We commute once again and obtain
$$
D_{421} = \int_0^t \int_0^{t-s} e^{i\tilde H_{\phi_t}(s+s_1)} \left( e^{-iH_{\phi_t}s_1} \tilde\psi_s\right) \otimes \left( \alpha^{-2} (\phi_t,\sigma_{\tilde\psi_s}) W^*(\alpha^2\phi_t) W(\alpha^2\phi_s) \Omega\right) ds_1\,ds \,.
$$
According to Lemma \ref{wellposedenergy} we have $|(\phi_t,\sigma_{\tilde\psi_s})|\lesssim 1$. This and computations similarly to those in the bound of $D_{32}$ yield
$$
\|D_{421}\|_{\mathcal L^2\otimes\mathcal F} \lesssim \alpha^{-2} t^2 \,,
\qquad
\|D_{422}\|_{\mathcal L^2\otimes\mathcal F} \lesssim \alpha^{-3} t^3 \,.
$$
Thus, we have shown that
\begin{equation}
\label{eq:d42}
\| D_{42} \|_{\mathcal L^2\otimes\mathcal F} \lesssim \alpha^{-2} t^2 \left( 1+\alpha^{-1} t \right) .
\end{equation}

%%%%%%%%%%%%%%%%%%%%%%%%%%%%%%%%%%%%%%%%%%%%

\subsection*{Bound on $D_{52}$}

To simplify the notation, let us introduce
\begin{align}\label{eq:number}
\mathcal N:=\int_{\mathbb{R}^3} b_k^* b_k \,dk \,.
\end{align}
We have
$$
D_{521} = \int_0^t \int_0^{t-s} e^{i\tilde H_{\phi_t}(s+s_1)} \left( e^{-iH_{\phi_t}s_1} \tilde\psi_s\right) \otimes \left( W^*(\alpha^2\phi_t) W(\alpha^2\phi_s) \mathcal N\, b^*(\sigma_{\tilde\psi_s})\Omega\right) ds_1\,ds \,.
$$
Moreover, by Corollary \ref{cor:bbb},
\begin{align*}
\left[\mathcal N,\ W^*(\alpha^2\phi_t)W(\alpha^2\phi_s)\right]
= &-W^*(\alpha^2\phi_t)W(\alpha^2\phi_s) \left( b(\phi_t)-b(\phi_s) \right) \\
&+W^*(\alpha^2\phi_t)W(\alpha^2\phi_s) \left( b^*(\phi_t) - b^*(\phi_s) \right) \\
&-W^*(\alpha^2\phi_t)W(\alpha^2\phi_s) \|\phi_t-\phi_2\|_2^2 \,,
\end{align*}
so
\begin{align*}
D_{522}= & \int_0^t \int_0^{t-s} e^{i\tilde H_{\phi_t}(s+s_1)} \left( e^{-iH_{\phi_t}s_1} \tilde\psi_s\right) \otimes \left( W^*(\alpha^2\phi_t)W(\alpha^2\phi_s) \right. \\
&  \qquad\qquad\qquad \times \left. \left( - b(\phi_t-\phi_s) + b^*(\phi_t- \phi_s) - \|\phi_t-\phi_2\|_2^2 \right) b^*(\sigma_{\tilde\psi_s})\Omega \right)ds_1\,ds \,.
\end{align*}
We use $\mathcal N\, b^*(\sigma_{\tilde\psi_s}) = b^*(\sigma_{\tilde\psi_s})\mathcal N + \alpha^{-2} b^*(\sigma_{\tilde\psi_s})$ and obtain
$$
D_{521} = \alpha^{-2} \int_0^t \int_0^{t-s} e^{i\tilde H_{\phi_t}(s+s_1)} \left( e^{-iH_{\phi_t}s_1} \tilde\psi_s\right) \otimes \left( W^*(\alpha^2\phi_t) W(\alpha^2\phi_s) b^*(\sigma_{\tilde\psi_s})\Omega\right) ds_1\,ds \,.
$$
Therefore, similarly as before,
$$
\| D_{521} \|_{\mathcal L^2\otimes\mathcal F} \lesssim \alpha^{-3}t^2 \,.
$$
For $D_{522}$ we commute again to get
\begin{align*}
D_{522} &= \int_0^t \int_0^{t-s} e^{i\tilde H_{\phi_t}(s+s_1)} \left( e^{-iH_{\phi_t}s_1} \tilde\psi_s\right) \otimes \left( W^*(\alpha^2\phi_t)W(\alpha^2\phi_s) \right. \\
& \ \ \times \left. \left( - \alpha^{-2} (\phi_t-\phi_s,\sigma_{\tilde\psi_s}) \Omega + b^*(\phi_t- \phi_s) b^*(\sigma_{\tilde\psi_s})\Omega - \|\phi_t-\phi_2\|_2^2\, b^*(\sigma_{\tilde\psi_s})\Omega \right) \right)ds_1\,ds \,.
\end{align*}
For the second term on the right side we compute
$$
\left\| b^*(\phi_t-\phi_s) b^*(\sigma_{\tilde\psi_s})\Omega \right\|_{\mathcal F} = \alpha^{-2}  \left( \|\phi_t-\phi_s\|_2^2 \|\sigma_{\tilde\psi_s}\|_2^2 + |(\phi_t-\phi_s,\sigma_{\tilde\psi_s})|^2 \right)^{1/2} \,.
$$
Using the bounds from Lemma \ref{wellposedenergy} for $\|\phi_t-\phi_s\|_2$ we obtain that
$$
\| D_{522} \|_{\mathcal L^2\otimes\mathcal F} \lesssim \alpha^{-4} t^3 \left( 1+ \alpha^{-1} t \right)  \,.
$$
Thus, we have shown that
\begin{equation}
\label{eq:d52}
\| D_{52} \|_{\mathcal L^2\otimes\mathcal F} \lesssim \alpha^{-3} t^2 \left( 1 + \alpha^{-2} t^2 \right) .
\end{equation}

%%%%%%%%%%%%%%%%%%%%%%%%%%%%%%%%%%%%%%%%%%%%%
%%%%%%%%%%%%%%%%%%%%%%%%%%%%%%%%%%%%%%%%%%%%%

\section{Bound on $D_0$}\label{sec:remainder}

%%%%%%%%%%%%%%%%%%%%%%%%%%%%%%%%%%%%%%%%%%%%%

We have already controlled $D_{02}$ in \eqref{eq:d02}, so it remains to consider $D_{01}$.

\subsection*{Bound on $D_{01}$}

We recall that
$$
D_{01} = \int_0^t e^{iH_{\phi_t}s} \int_{\R^3} \left( e^{ik\cdot x} W^*(\alpha^2 \phi_t)W(\alpha^2\phi_s) b_k^*\ \tilde\psi_s\otimes\Omega\right) \frac{dk}{|k|} \,ds
$$
The main difficulty here, which we will encounter in various forms throughout this paper, is the unboundedness of the operator $\int e^{ik\cdot x} b_{k}^{*}|k|^{-1}\,dk$ (for any fixed $x\in\R^3$), since $e^{ik\cdot x}|k|^{-1}\not\in\mathcal L^2(\R^3)$.

To overcome this difficulty we make use of the oscillatory behavior of $e^{ik\cdot x}$ via the formula
\begin{align}\label{eq:standard}
e^{ik\cdot x}=\frac{1-ik\cdot \nabla_x}{1+|k|^2}\ e^{ik\cdot x}
\end{align}
and aim at integrating by parts with respect to $x$. However, this integration by parts creates a new difficulty: the resulting operator $\nabla_x$ is unbounded and has to be controlled.

To overcome this new difficulty, it will be desirable to have an operator $(-\Delta+1)^{-1}$ somewhere in the expression of $D_{01}$ so that we can use it to control $\nabla_x,$ since obviously $\nabla_x (-\Delta+1)^{-1}$ is bounded. It is equivalent and technically more convenient to work with $(H_{\phi_t}+M)^{-1}$, where $M>0$ is a large constant (independent of $\alpha$ and $t$), instead of $(-\Delta+1)^{-1}$. In order to create this term we first integrate by parts in $s$ and make use of the identity
\begin{align}
e^{i H_{\phi_t}s}= -i \left(H_{\phi_t}+M\right)^{-1} e^{-iMs} \partial_{s}\left[e^{i\left(H_{\phi_t}+M\right)s}\right] .
\end{align}
We obtain, using the fact that $H_{\phi_t}$ commutes with $W(\alpha^2\phi_s)$,
\begin{align*}
D_{01}=&-i e^{iH_{\phi_t}t} \left(H_{\phi_t}+M\right)^{-1} \int_{\R^3} e^{ik\cdot x} b_k^* \tilde\psi_t\otimes \Omega \,\frac{dk}{|k|} \\
&+i W^*(\alpha^2\phi_t) W(\alpha^2\phi_0) \left( H_{\phi_t}+M\right)^{-1} \int_{\mathbb{R}^3} e^{ik\cdot x} b_k^*\tilde\psi_0\otimes \Omega\,\frac{dk}{|k|} \\
&+M\int_{0}^{t} e^{iH_{\phi_t}s} W^*(\alpha^2\phi_t) W(\alpha^2\phi_s)  \left(H_{\phi_t}+M\right)^{-1} \int_{\mathbb{R}^3} e^{ik\cdot x} b_k^*\tilde\psi_s\otimes \Omega \,\frac{dk}{|k|}\,ds \\
&+i\int_{0}^{t} e^{iH_{\phi_t}s} W^*(\alpha^2\phi_t) W(\alpha^2\phi_s)  \left(H_{\phi_t}+M\right)^{-1} \int_{\mathbb{R}^3} e^{ik\cdot x} b_k^* \partial_{s}\tilde\psi_s\otimes \Omega\, \frac{dk}{|k|} \,ds \\
&+i\int_{0}^{t} e^{iH_{\phi_t}s} W^*(\alpha^2\phi_t) \left(\partial_s W(\alpha^2\phi_s)\right) \left(H_{\phi_t}+M\right)^{-1} \int_{\mathbb{R}^3} e^{ik\cdot x} b_k^* \tilde\psi_s\otimes\Omega \,\frac{dk}{|k|} \,ds \\
=& D_{011}+D_{012}+D_{013}+D_{014}+D_{015} \,,
\end{align*}
where the terms $D_{01k}$ are defined in a natural way. We will prove the following lemma.

\begin{lemma}\label{singular}
For $u\in \mathcal H^1(\R^3)$ and $f\in \mathcal L^2(\R^3)$,
$$
\left\| \left( -\Delta+1\right)^{-1/2} \int_{\R^3} e^{ik\cdot x} b_k^* u\otimes \Omega \,\frac{dk}{|k|} \right\|_{\mathcal L^2\otimes\mathcal F} \lesssim \alpha^{-1} \|u\|_{\mathcal H^1}
$$
and
$$
\left\| \left( -\Delta+1\right)^{-1/2} \int_{\R^3} e^{ik\cdot x} b^*(f) b_k^* u\otimes \Omega \,\frac{dk}{|k|} \right\|_{\mathcal L^2\otimes\mathcal F} \lesssim \alpha^{-2} \|u\|_{\mathcal H^1} \|f\|_{2} \,.
$$
\end{lemma}

We defer the proof of this lemma to the end of this section and first show how to use it to control $D_{01}$. We know from Corollary \ref{formopdom} and Lemma \ref{wellposedenergy} that we can choose $M$ large enough so that $(H_{\phi_t}+M)^{-1/2} (-\Delta+1)^{1/2}$ is bounded uniformly in $t\in\R$. Moreover, by Proposition \ref{THM:wellposedness}, $\tilde\psi_t$ and $\partial_t\tilde\psi_t$ belong to $H^1(\R^3)$ and have uniformly bounded norms for $t\in [0,\alpha^2]$; see also the remark at the beginning of Subsection \ref{sec:decomp} concerning the bounds on $\partial_t\tilde\psi_t$. These facts, together with the unitarity of $e^{iH_{\phi_t}s}$, $W^*(\alpha^2\phi_t)$ and $W(\alpha^2\phi_s)$, imply that
$$
\| D_{011} \|_{\mathcal L^2\otimes\mathcal F} \lesssim \alpha^{-1} \,,
\qquad
\| D_{012} \|_{\mathcal L^2\otimes\mathcal F} \lesssim \alpha^{-1} \,,
$$
and
$$
\| D_{013} \|_{\mathcal L^2\otimes\mathcal F} \lesssim \alpha^{-1}t \,,
\qquad
\| D_{014} \|_{\mathcal L^2\otimes\mathcal F} \lesssim \alpha^{-1}t \,,
$$
In order to deal with the term $D_{015}$ we make use of \eqref{eq:firPW} and find
\begin{align*}
D_{015}= & -\int_{0}^{t} \left( \im  (\phi_s,\alpha^2\partial_s\phi_s)\right) e^{iH_{\phi_t}s} W^*(\alpha^2\phi_t) W(\alpha^2\phi_s) \\
& \qquad\qquad\qquad \times \left(H_{\phi_t}+M\right)^{-1} \int_{\mathbb{R}^3} e^{ik\cdot x} b_k^* \tilde\psi_s\otimes\Omega \,\frac{dk}{|k|} \,ds \\
& + i\int_{0}^{t} e^{iH_{\phi_t}s} W^*(\alpha^2\phi_t) W(\alpha^2\phi_s) \\
& \qquad\qquad\qquad \times\left(H_{\phi_t}+M\right)^{-1} \int_{\mathbb{R}^3} e^{ik\cdot x}  b^*(\alpha^2\partial_s\phi_s) b_k^* \tilde\psi_s\otimes\Omega \,\frac{dk}{|k|} \,ds \\
& - i\int_{0}^{t} e^{iH_{\phi_t}s} W^*(\alpha^2\phi_t) W(\alpha^2\phi_s) \\
&\qquad\qquad\qquad\times \left(H_{\phi_t}+M\right)^{-1} \int_{\mathbb{R}^3} e^{ik\cdot x} b(\alpha^2\partial_s\phi_s) b_k^* \tilde\psi_s\otimes\Omega \,\frac{dk}{|k|} \,ds \\
= & D_{0151}+D_{0152}+D_{0153} \,.
\end{align*}
From Lemma \ref{wellposedenergy} we know that $|(\phi_s,\alpha^2\partial_s\phi_s)|\lesssim 1$ and $\|\alpha^2\partial_s\phi_s\|\lesssim 1$. Thus, the first and the second bound in Lemma \ref{singular} imply, respectively,
$$
\| D_{0151} \|_{\mathcal L^2\otimes\mathcal F} \lesssim \alpha^{-1} t \,,
\qquad
\| D_{0152} \|_{\mathcal L^2\otimes\mathcal F} \lesssim \alpha^{-2} t \,.
$$
For $D_{0153}$ we use the commutation relations to rewrite it as
$$
D_{0153} = - i \int_{0}^{t} e^{iH_{\phi_t}s} W^*(\alpha^2\phi_t) W(\alpha^2\phi_s) \left(H_{\phi_t}+M\right)^{-1} g_s \tilde\psi_s\otimes\Omega \,ds
$$
with $g_s$ from \eqref{eq:gs}. Therefore, Proposition \ref{THM:wellposedness} yields
$$
\| D_{0153} \|_{\mathcal L^2\otimes\mathcal F} \lesssim \alpha^{-2} t \,.
$$
To summarize, we have shown that
\begin{equation}
\label{eq:d01}
\left\| D_{01} \right\|_{\mathcal L^2\otimes\mathcal F} \lesssim \alpha^{-1} \left(1+t\right).
\end{equation}

It remains to give the

\begin{proof}[Proof of Lemma \ref{singular}]
For any $\gamma\in \mathcal{L}^2(\mathbb{R}^3)\otimes \mathcal{F}$ and $(\Phi_k)_{k\in\R^3} \subset\mathcal F$ we use \eqref{eq:standard} to find
\begin{align*}
& \left\langle \gamma,\  \left( -\Delta+1\right)^{-1/2} \int_{\R^3} e^{ik\cdot x} u \otimes \Phi_k\, \frac{dk}{|k|} \right\rangle_{\mathcal{L}^2\otimes \mathcal{F}} \\
& \qquad = \left\langle \nabla \left(-\Delta+1\right)^{-1/2} \gamma, \int_{\mathbb{R}^3} \frac{ik e^{ik\cdot x}}{|k|(1+|k|^2)} u\otimes \Phi_k \,dk \right\rangle_{\mathcal{L}^2\otimes \mathcal{F}}\\
&\qquad\qquad+\left\langle  \left(-\Delta+1\right)^{-1/2} \gamma, \int_{\mathbb{R}^3} \frac{ik e^{ik\cdot x}}{|k|(1+|k|^2)} (\nabla u)\otimes \Phi_k \,dk \right\rangle_{\mathcal{L}^2\otimes \mathcal{F}}\\
&\qquad\qquad+\left\langle  \left(-\Delta+1\right)^{-1/2} \gamma, \int_{\mathbb{R}^3} \frac{e^{ik\cdot x}}{|k|(1+|k|^2)} u\otimes \Phi_k\,dk \right\rangle_{\mathcal{L}^2\otimes \mathcal{F}} \,.
\end{align*}
Clearly, $\|\nabla(-\Delta+1)^{-1/2}\gamma\|_{\mathcal L^2\otimes\mathcal F}\leq \|\gamma\|_{\mathcal L^2\otimes\mathcal F}$ and $\|(-\Delta+1)^{-1/2}\gamma\|_{\mathcal L^2\otimes\mathcal F}\leq \|\gamma\|_{\mathcal L^2\otimes\mathcal F}$, so
\begin{align*}
\left\| \left( -\Delta+1\right)^{-1/2} \int_{\R^3} e^{ik\cdot x} u\otimes \Phi_k \,\frac{dk}{|k|} \right\|_{\mathcal L^2\otimes\mathcal F} \lesssim
\|u\|_{\mathcal H^1} & \sup_{x\in\R^d} \left( 
\left\| \int_{\mathbb{R}^3} \frac{ik e^{ik\cdot x}}{|k|(1+|k|^2)} \Phi_k \,dk \right\|_{\mathcal{F}} \right. \\
& \qquad\left. + \left\|\int_{\mathbb{R}^3} \frac{e^{ik\cdot x}}{|k|(1+|k|^2)} \Phi_k\,dk \right\|_{\mathcal{F}} \right).
\end{align*}
If $\Phi_k=b_k^*\Omega$, we use the fact that $\frac{1}{|k|(1+|k|^2)}$, $\frac{k}{|k|(1+|k|^2)}\in \mathcal{L}^2(\mathbb{R}^3)$ to conclude that, uniformly in $x\in\R^3$,
\begin{align*}
\left\| \int_{\mathbb{R}^3} \frac{ik e^{ik\cdot x}}{|k|(1+|k|^2)} b_k^* \Omega \,dk \right\|_{\mathcal{F}}\lesssim\alpha^{-1} \,,\qquad
\left\|\int_{\mathbb{R}^3} \frac{e^{ik\cdot x}}{|k|(1+|k|^2)} b_k^* \Omega\,dk \right\|_{\mathcal{F}}\lesssim \alpha^{-1} \,.
\end{align*}
This proves the first bound in the lemma. If $\Phi_k=b^*(f)b_k^*\Omega$, one can similarly show that
\begin{align*}
\left\| \int_{\mathbb{R}^3} \frac{ik e^{ik\cdot x}}{|k|(1+|k|^2)} b^*(f)b_k^* \Omega \,dk \right\|_{\mathcal{F}}\lesssim \frac{\|f\|_2}{\alpha^2} \,,\quad
\left\|\int_{\mathbb{R}^3} \frac{e^{ik\cdot x}}{|k|(1+|k|^2)} b^*(f) b_k^* \Omega\,dk \right\|_{\mathcal{F}}\lesssim \frac{\|f\|_2}{\alpha^2} \,.
\end{align*}
This proves the second bound in the lemma.
\end{proof}

%%%%%%%%%%%%%%%%%%%%%%%%%%%%%%%%%%%%%%%%%%%%
%%%%%%%%%%%%%%%%%%%%%%%%%%%%%%%%%%%%%%%%%%%%

\section{Bound on $D_1$}\label{sec:D15}

\subsection*{Bound on $D_{111}$}

We recall equation \eqref{eq:d111def} for $D_{111}$. In this equation, we commute $e^{ik\cdot x}$ with $e^{-iH_{\phi_t}s}$. Thus, if we introduce the operator
\begin{align}\label{eq:hphitk}
H_{\phi}(k):=e^{ik\cdot x} H_{\phi}e^{-ik\cdot x}=\left(i\nabla_x+k\right)^2+V_{\phi}+\int_{\mathbb{R}^3} |\phi(k)|^2 \,dk \,,
\end{align}
we obtain
\begin{align*}
D_{111}&=\int_0^{t} \int_{0}^{t-s} \int_{\R^3} \int_{\R^3} e^{i\tilde{H}_{\phi_t}(s+s_1)} e^{-iH_{\phi_t}(k)s_1} e^{i(k+k')\cdot x}
 W^*(\alpha^2\phi_t)W(\alpha^2\phi_s) b_k^* b_{k'}^* \\
 &\qquad\qquad\qquad\qquad \times \tilde\psi_s\otimes\Omega \,\frac{dk'}{|k'|}\,\frac{dk}{|k|}\,ds_1\,ds \,.
\end{align*}
Controlling $D_{111}$ is harder than controlling $D_{01}$ because there are two slowly decaying terms $|k|^{-1}$ and $|k'|^{-1}$. The beginning of the proof, however, is similar, namely, for a large constant $M>0$ to be specified, independent of $t$ and $\alpha$, we integrate by parts in $s$ using
$$
e^{i\tilde{H}_{\phi_t}s}=-i \left(\tilde{H}_{\phi_t}+M\right)^{-1} e^{-iMs} \left[\partial_{s}e^{i\left(\tilde{H}_{\phi_t}+M\right)s} \right] \,.
$$
In this way we obtain
\begin{align*}
D_{111}=& - i \int_{0}^{t} \int_{\R^3}\int_{\R^3} e^{i\tilde{H}_{\phi_t}t} \left(\tilde{H}_{\phi_t}+M\right)^{-1} e^{-iH_{\phi_t}(k)s_1} e^{i(k'+k)\cdot x} \\
& \quad \times W^*(\alpha^2\phi_t)W(\alpha^2\phi_{t-s_1}) b_k^* b_{k'}^* \tilde\psi_{t-s_1} \otimes \Omega \,\frac{dk'}{|k'|}\,\frac{dk}{|k|}\,ds_1 \\
& + i \int_{0}^{t} \int_{\R^3}\int_{\R^3} e^{i\tilde{H}_{\phi_t}s_1} \left(\tilde{H}_{\phi_t}+M\right)^{-1} e^{-iH_{\phi_t}(k)s_1} e^{i(k'+k)\cdot x} \\
& \quad \times W^*(\alpha^2\phi_t)W(\alpha^2\phi_0) b_k^* b_{k'}^* \tilde\psi_{0} \otimes \Omega \,\frac{dk'}{|k'|}\,\frac{dk}{|k|}\,ds_1 \\
 &+ M\int_0^{t} \int_{0}^{t-s} \int_{\R^3} \int_{\R^3} e^{i\tilde{H}_{\phi_t}(s+s_1)} \left(\tilde{H}_{\phi_t}+M\right)^{-1} e^{-iH_{\phi_t}(k)s_1} e^{i(k'+k)\cdot x} \nonumber\\
 & \quad \times W^*(\alpha^2\phi_t)W(\alpha^2\phi_s) b_k^* b_{k'}^* \tilde\psi_s\otimes \Omega \,\frac{dk'}{|k'|}\,\frac{dk}{|k|}\,ds_1\,ds\nonumber\\
& +i\int_0^{t} \int_{0}^{t-s} \int_{\R^3} \int_{\R^3} e^{i\tilde{H}_{\phi_t}(s+s_1)} \left(\tilde{H}_{\phi_t}+M\right)^{-1}  e^{-iH_{\phi_t}(k)s_1} e^{i(k'+k)\cdot x} \nonumber\\
 & \quad \times W^*(\alpha^2\phi_t)W(\alpha^2\phi_s) b_k^* b_{k'}^* [\partial_s\tilde\psi_s]\otimes \Omega \,\frac{dk'}{|k'|}\,\frac{dk}{|k|}\,ds_1\,ds \nonumber  \\
&+ i\int_0^{t} \int_{0}^{t-s} \int_{\R^3} \int_{\R^3} e^{i\tilde{H}_{\phi_t}(s+s_1)} \left(\tilde{H}_{\phi_t}+M\right)^{-1}  e^{-iH_{\phi_t}(k)s_1} e^{i(k'+k)\cdot x} \nonumber\\
& \quad \times W^*(\alpha^2\phi_t)[\partial_{s}W(\alpha^2\phi_s)] b_k^* b_{k'}^* \tilde\psi_s\otimes \Omega \,\frac{dk'}{|k'|}\,\frac{dk}{|k|}\,ds_1\,ds \nonumber \,.
\end{align*}
We now use \eqref{eq:tildehphi}, which implies
\begin{align*}
\left(\tilde{H}_{\phi_t}+M\right)^{-1} W^*(\alpha^2\phi_t)W(\alpha^2\phi_s) & = W^*(\alpha^2\phi_t) \left( \tilde H_\alpha^F + M \right)^{-1} W(\alpha^2\phi_s) \\
& = W^*(\alpha^2\phi_t) W(\alpha^2\phi_s) \left( \tilde H_{\phi_s} + M \right)^{-1} \,,
\end{align*}
in order to commute $\left(\tilde{H}_{\phi_t}+M\right)^{-1}$ to the right through $W^*(\alpha^2\phi_t)W(\alpha^2\phi_s)$. Moreover, we use Lemma \ref{LM:derivative} to compute $\partial_s W(\alpha^2\phi_s)$. In this way we obtain
\begin{align*}
D_{111}=& - i \int_{0}^{t} e^{i\tilde{H}_{\phi_t}t} W^*(\alpha^2\phi_t)W(\alpha^2\phi_{s}) Q_1 \,ds \\
& + i \int_{0}^{t} e^{i\tilde{H}_{\phi_t}s_1} W^*(\alpha^2\phi_t)W(\alpha^2\phi_0) Q_2 \,ds_1 \\
&+ M\int_0^{t} \int_{0}^{t-s} e^{i\tilde{H}_{\phi_t}(s+s_1)} W^*(\alpha^2\phi_t)W(\alpha^2\phi_s) Q_3 \,ds_1\,ds\nonumber\\
& +i\int_0^{t} \int_{0}^{t-s} e^{i\tilde{H}_{\phi_t}(s+s_1)} W^*(\alpha^2\phi_t)W(\alpha^2\phi_s) Q_4 \,ds_1\,ds \nonumber  \\
&+ i\int_0^{t} \int_{0}^{t-s} e^{i\tilde{H}_{\phi_t}(s+s_1)} W^*(\alpha^2\phi_t) W(\alpha^2\phi_s) Q_5 \,ds_1\,ds
\end{align*}
with
\begin{align*}
Q_1 & := \left(\tilde{H}_{\phi_s}+M\right)^{-1} \int_{\R^3}\int_{\R^3} e^{-iH_{\phi_t}(k)(t-s)} e^{i(k'+k)\cdot x} b_k^* b_{k'}^* \tilde\psi_{s} \otimes \Omega \,\frac{dk'}{|k'|}\,\frac{dk}{|k|} \,,\\
Q_2 & := \left(\tilde{H}_{\phi_0}+M\right)^{-1} \int_{\R^3}\int_{\R^3} e^{-iH_{\phi_t}(k)s_1} e^{i(k'+k)\cdot x} b_k^* b_{k'}^* \tilde\psi_{0} \otimes \Omega \,\frac{dk'}{|k'|}\,\frac{dk}{|k|} \,,\\
Q_3 & := \left(\tilde{H}_{\phi_s}+M\right)^{-1} \int_{\R^3} \int_{\R^3} e^{-iH_{\phi_t}(k)s_1} e^{i(k'+k)\cdot x} b_k^* b_{k'}^* \tilde\psi_s\otimes \Omega \,\frac{dk'}{|k'|}\,\frac{dk}{|k|} \,,\\
Q_4 & := \left(\tilde{H}_{\phi_s}+M\right)^{-1} \int_{\R^3} \int_{\R^3} e^{-iH_{\phi_t}(k)s_1} e^{i(k'+k)\cdot x} b_k^* b_{k'}^* [\partial_s\tilde\psi_s]\otimes \Omega \,\frac{dk'}{|k'|}\,\frac{dk}{|k|} \,,\\
Q_5 & := \left(\tilde{H}_{\phi_s}+M\right)^{-1} \int_{\R^3} \int_{\R^3}  e^{-iH_{\phi_t}(k)s_1} e^{i(k'+k)\cdot x}
\left( b^*(\alpha^2\partial_s\phi_s) - b(\alpha^2\partial_s\phi_s) \right. \\
& \qquad\qquad\qquad\qquad\qquad\qquad \left. + i \im (\phi_s,\alpha^2\partial_s\phi_s)\right) b_k^* b_{k'}^* \tilde\psi_s\otimes \Omega \,\frac{dk'}{|k'|}\,\frac{dk}{|k|} \,.\qquad\end{align*}
(Here, we suppress the dependence on $t$, $s$ and $s_1$ in the notation of the $Q_j$'s.)

In the remainder of this section we shall show that, uniformly for $0\leq s,s_1\leq t \leq \alpha^2$,
\begin{equation}
\label{eq:qbounds}
\| Q_j \|_{\mathcal L^2\otimes\mathcal F} \lesssim \alpha^{-2} \quad\text{if}\ j=1,2,3,4,5 \,.
\end{equation}
This will imply that
\begin{equation}
\label{eq:d111}
\left\| D_{111} \right\|_{\mathcal L^2\otimes\mathcal F} \lesssim \alpha^{-2} t(1+t) \,.
\end{equation}

Since the operator $(\tilde{H}_{\phi_s}+M)^{-1} \left( -\Delta+\mathcal N+M\right)$ is \emph{not} bounded, bounding the $Q_j$ is rather involved. (Here $\mathcal N$ was introduced in \eqref{eq:number}.) With the notation
$$
Z_\phi :=  V_\phi + \int_{\R^3} |\phi(x)|^2\,dk + \int_{\R^3} \left( e^{-ik\cdot x} b_k + e^{ik\cdot x} b_k^*\right)\frac{dk}{|k|} + b(\phi) + b^*(\phi)
$$
abbreviate \eqref{eq:tildehphialt} as
$$
\tilde H_{\phi} = -\Delta + \mathcal N + Z_\phi \,.
$$
Defining
$$
\tilde Z_\phi := \left( -\Delta+\mathcal N+M\right)^{-1/2} Z_\phi \left(-\Delta+\mathcal N + M\right)^{-1/2}
$$
we have
\begin{align*}
& \left( H_\phi +M \right)^{-1} =\left( -\Delta+\mathcal N+M\right)^{-1/2} \left( 1+ \tilde Z_\phi \right)^{-1} \left( -\Delta+\mathcal N+M\right)^{-1/2} \\
& \quad= \left(-\Delta+\mathcal N+M\right)^{-1} \\
& \quad\qquad - \left( -\Delta+\mathcal N+M\right)^{-1/2} \left( 1+\tilde Z_\phi \right)^{-1} \left( -\Delta+\mathcal N+M\right)^{-1/2} Z_\phi \left( -\Delta+\mathcal N+M\right)^{-1}.
\end{align*}
It is not difficult to see that for every $\epsilon>0$ and $A>0$ there is an $M$ such that
\begin{equation}
\label{eq:zinverse}
\left\| \tilde Z_\phi \right\|_{\mathcal L^2\otimes\mathcal F\mapsto\mathcal L^2\otimes F} \leq \epsilon
\end{equation}
for all $\phi$ with $\|\phi\|_{\mathcal L^2} \leq A$; for details of this argument we refer to \cite{FrankSchlein2013}. Thus, using the bound on $\|\phi_s\|_{\mathcal L^2}$ from Lemma \ref{wellposedenergy}, we can choose $M$ in such a way that
$$
\left\| \tilde Z_{\phi_s} \right\|_{\mathcal L^2\otimes\mathcal F\mapsto\mathcal L^2\otimes F} \leq \frac12
\qquad\text{for all}\ s>0 \,.
$$
Therefore, the operator $1+\tilde Z_{\phi_s}$ in the above formula for $\left(H_{\phi_s}+M\right)^{-1}$ is invertible. We use this formula to decompose
\begin{align}
\label{eq:q1decomp}
Q_1 = & \left( 1 -\left( -\Delta+\mathcal N+M\right)^{-1/2} \left( 1+\tilde Z_{\phi_s} \right)^{-1} \left( -\Delta+\mathcal N+M\right)^{-1/2} \right. \nonumber \\
& \qquad\qquad \left. \times \left( V_{\phi_s} + \int_{\R^3} |\phi_s(x)|^2\,dk + b(\phi_s) + b^*(\phi_s) \right) \right) Q_{10} \nonumber\\
& - \left( -\Delta+\mathcal N+M\right)^{-1/2} \left( 1+\tilde Z_{\phi_s} \right)^{-1} \left( Q_{11} + Q_{12} \right)
\end{align}
with
\begin{align*}
Q_{10} & := \left( -\Delta+\mathcal N+M\right)^{-1} \int_{\R^3} \int_{\R^3} e^{-iH_{\phi_t}(k)(t-s)} e^{i(k'+k)\cdot x} b_k^* b_{k'}^* \tilde\psi_{s} \otimes \Omega \,\frac{dk'}{|k'|}\,\frac{dk}{|k|} \,, \\
Q_{11} & :=  \left( -\Delta+\mathcal N+M\right)^{-1/2} \left( \int_{\R^3} e^{-ik''\cdot x} b_{k''} \,\frac{dk''}{|k''|} \right) \left( -\Delta+\mathcal N+M\right)^{-1} \\
& \qquad\qquad\qquad \times \int_{\R^3}\int_{\R^3} e^{-iH_{\phi_t}(k)(t-s)} e^{i(k'+k)\cdot x} b_k^* b_{k'}^* \tilde\psi_{s} \otimes \Omega \,\frac{dk'}{|k'|}\,\frac{dk}{|k|} \,, \\
Q_{12} & :=  \left( -\Delta+\mathcal N+M\right)^{-1/2} \left( \int_{\R^3} e^{ik''\cdot x} b_{k''}^* \,\frac{dk''}{|k''|} \right) \left( -\Delta+\mathcal N+M\right)^{-1} \\
& \qquad\qquad\qquad \times \int_{\R^3}\int_{\R^3} e^{-iH_{\phi_t}(k)(t-s)} e^{i(k'+k)\cdot x} b_k^* b_{k'}^* \tilde\psi_{s} \otimes \Omega \,\frac{dk'}{|k'|}\,\frac{dk}{|k|} \,.
\end{align*}
Using \eqref{eq:zinverse}, the fact that $(-\Delta+\mathcal N+M)^{-1/2} (b(\phi_s)+b^*(\phi_s))$ is bounded uniformly in $s$ as well as the estimates $\|V_{\phi_s}\|_\infty \lesssim 1$ (from \eqref{eq:aprioriv} and Proposition \ref{THM:wellposedness}), $\|\phi_s\|_2 \lesssim 1$ (from Lemma \ref{wellposedenergy}), we conclude from \eqref{eq:q1decomp} that
$$
\|Q_1\|_{\mathcal L^2\otimes\mathcal F} \lesssim \|Q_{10}\|_{\mathcal L^2\otimes\mathcal F} + \|Q_{11}\|_{\mathcal L^2\otimes\mathcal F} + \|Q_{12}\|_{\mathcal L^2\otimes\mathcal F} \,.
$$
We now bound the three terms on the right side separately.

%%%%%%%%%%%%%%%%%%%%%%%%%%%%%%%

\subsection*{Bound on $Q_{10}$}

In order to control $Q_{10}$ we prove an analogue of Lemma \ref{singular} for the case of two singularities.

\begin{lemma}\label{singular2}
For $u\in\mathcal H^2(\R^3)$, $f\in \mathcal L^2(\R^3)$ and $s\in\R$,
$$
\left\| (-\Delta+1)^{-1} \int_{\R^3} \int_{\R^3} e^{-iH_{\phi_t}(k)s} e^{i(k+k')\cdot x} b_k^* b_{k'}^* u\otimes\Omega \frac{dk'\,dk}{|k'|\,|k|} \right\|_{\mathcal L^2\otimes\mathcal F} \lesssim \alpha^{-2} \|u\|_{\mathcal H^2} \,.
$$
\end{lemma}

Before proving this lemma we show how to use it to bound $Q_{10}$. Note that, since $Q_{10}$ involves only $b_k^* b_{k'}^*\Omega$, the operator $(-\Delta+\mathcal N+M)^{-1}$ in its definition can be replaced by $(-\Delta+2\alpha^{-2} +M)^{-1}$. This observation, together with Lemma \ref{singular2} and the uniform boundedness of $\tilde\psi_s$ in $\mathcal H^2$ for $s\in[0,\alpha^2]$ (see Proposition \ref{THM:wellposedness}), proves that
\begin{equation}
\label{eq:q10}
\|Q_{10}\|_{\mathcal L^2\otimes\mathcal F} \lesssim \alpha^{-2} \,.
\end{equation}

\begin{proof}[Proof of Lemma \ref{singular2}]
We shall show that for any $\gamma\in\mathcal L^2(\R^3)\otimes\mathcal F$
$$
\left| \left\langle \gamma, (-\Delta+1)^{-1} \int_{\R^3} \int_{\R^3} e^{-iH_{\phi_t}(k)s} e^{i(k+k')\cdot x} b_k^* b_{k'}^* u\otimes\Omega \,\frac{dk'\,dk}{|k'|\,|k|} \right\rangle \right| \lesssim \alpha^{-2} \|\gamma\|_{\mathcal L^2\otimes\mathcal F} \|u\|_{\mathcal H^2} \,. 
$$
We integrate by parts twice in $x$ and use \eqref{eq:standard} with $k$ replaced by $k+k'$. A typical term that is obtained in this way in the inner product on the left side is
\begin{align*}
\left\langle   e^{iH_{\phi_t}(k)s} \partial_{x_i} \partial_{x_j} \left( -\Delta +1 \right)^{-1} \gamma,\ \int_{\R^3} \int_{\mathbb{R}^3} e^{i(k+k')\cdot x} b_k^* b_{k'}^* u\otimes \Omega \,\frac{(k_i+k_i')(k_j+k'_j)\,dk'\,dk}{|k|\,|k'|\,(1+|k+k'|^2)^2} \right\rangle.
\end{align*}
Since $\partial_{x_i}\partial_{x_j}(-\Delta+1)^{-1}$ is bounded and $e^{iH_{\phi_t}(k)s}$ is unitary, the vector on the left side of the inner product is bounded in norm by $\|\gamma\|_{\mathcal L^2\otimes\mathcal F}$. We now show that the vector on the right side of the inner product is bounded as well. We compute
\begin{align*}
& \left\| \int_{\R^3} \int_{\mathbb{R}^3} e^{i(k+k')\cdot x} b_k^* b_{k'}^* u\otimes \Omega \,\frac{(k_i+k_i')(k_j+k'_j)}{|k|\,|k'|\,(1+|k+k'|^2)^2} \,dk'\,dk \right\|_{\mathcal L^2\otimes\mathcal F}^2 \\
&  \qquad = 2\alpha^{-4} \|u\|_2^2 \int_{\R^3}\int_{\R^3} \frac{(k_i+k_i')^2(k_j+k'_j)^2}{|k|^2\,|k'|^2\,(1+|k+k'|^2)^4} \,dk'\,dk \,.
\end{align*}
The desired bound now follows from the fact that the double integral on the right side is finite. Other terms that arise in the integration by parts are controlled similarly and we omit the details. This proves the lemma.
\end{proof}

%%%%%%%%%%%%%%%%%%%%%%%%%%%%%%%

\subsection*{Bound on $Q_{11}$}

By considering the number of involved field particles we can replace $\mathcal N$ in the definition of $Q_{11}$ by numbers and obtain
\begin{align*}
Q_{11} & =  \left( -\Delta+\alpha^{-2} +M\right)^{-1/2} \left( \int_{\R^3} e^{-ik''\cdot x} b_{k''} \,\frac{dk''}{|k''|} \right) \left( -\Delta+2\alpha^{-2} +M\right)^{-1} \\
& \qquad\qquad\qquad \times \int_{\R^3}\int_{\R^3} e^{-iH_{\phi_t}(k)(t-s)} e^{i(k'+k)\cdot x} b_k^* b_{k'}^* \tilde\psi_{s} \otimes \Omega \,\frac{dk'}{|k'|}\,\frac{dk}{|k|} \,.
\end{align*}
Next, by commuting $b_{k''}$ to the right,
\begin{align*}
Q_{11} & =  \alpha^{-2} \left( -\Delta+\alpha^{-2} +M\right)^{-1/2} \int_{\R^3} \left( (i\nabla -k')^2+2\alpha^{-2} +M\right)^{-1} \\
& \qquad\qquad\qquad \times \int_{\R^3} e^{-ik'\cdot x} e^{-iH_{\phi_t}(k)(t-s)} e^{i(k'+k)\cdot x} b_k^* \tilde\psi_{s} \otimes \Omega \,\frac{dk'}{|k'|^2}\,\frac{dk}{|k|} \\
& \quad + \alpha^{-2} \left( -\Delta+\alpha^{-2} +M\right)^{-1/2} \int_{\R^3} \left( (i\nabla -k)^2+2\alpha^{-2} +M\right)^{-1} \\
& \qquad\qquad\qquad \times \int_{\R^3} e^{-ik\cdot x} e^{-iH_{\phi_t}(k)(t-s)} e^{i(k'+k)\cdot x} b_{k'}^* \tilde\psi_{s} \otimes \Omega \,\frac{dk'}{|k'|}\,\frac{dk}{|k|^2} \,.
\end{align*}
It remains to compute the norm of this expression. Since this is considerably easier than for $Q_{12}$ we omit the details and only state the final result,
\begin{equation}
\label{eq:q01}
\|Q_{11}\|_{\mathcal L^2\otimes\mathcal F} \lesssim \alpha^{-3} \,.
\end{equation}

%%%%%%%%%%%%%%%%%%%%%%%%%%%%%%%

\subsection*{Bound on $Q_{12}$}

In the same way as for $Q_{11}$, we can replace $\mathcal N$ by a number, so that
\begin{align*}
Q_{12} & =  \left( -\Delta+3\alpha^{-2} +M\right)^{-1/2} \int_{\R^3} e^{ik''\cdot x} b_{k''}^* \left( -\Delta+2\alpha^{-2} +M\right)^{-1} \\
& \qquad\qquad\qquad \times \int_{\R^3} \int_{\R^3} e^{-iH_{\phi_t}(k)(t-s)} e^{i(k'+k)\cdot x} b_{k'}^* b_k^* \tilde\psi_{s} \otimes \Omega \,\frac{dk'}{|k'|}\,\frac{dk}{|k|} \,.
\end{align*}
Next, we commute $e^{ik''\cdot x}$ and $e^{i(k'+k)\cdot x}$ to the right and obtain
\begin{align*}
Q_{12} & =  \int_{\R^3}\int_{\R^3}\int_{\R^3} b_k^* b_{k'}^* b_{k''}^* e^{i(k+k'+k'')\cdot x} \left( (i\nabla-k-k'-k'')^2+3\alpha^{-2} +M\right)^{-1/2} \\
& \qquad\qquad\quad \times \left( (i\nabla-k-k')^2+2\alpha^{-2} +M\right)^{-1} e^{-iH_{\phi_t}(-k')(t-s)} \tilde\psi_{s} \otimes \Omega \,\frac{dk''}{|k''|}\,\frac{dk'}{|k'|}\,\frac{dk}{|k|} \,.
\end{align*}
We now compute the norm of this expression. For the part of the norm over $\mathcal F$, we use the fact that
\begin{align*}
& \alpha^6 \langle\Omega,b_{k_1}b_{k_2}b_{k_3}b_{k_4}^*b_{k_5}^*b_{k_6}^*\Omega\rangle \\
& \qquad = \delta(k_1-k_4)\delta(k_2-k_5)\delta(k_3-k_6) + \delta(k_1-k_4)\delta(k_2-k_6)\delta(k_3-k_5) \\
& \qquad\quad + \delta(k_1-k_5)\delta(k_2-k_4)\delta(k_3-k_6) + \delta(k_1-k_5)\delta(k_2-k_6)\delta(k_3-k_4) \\
& \qquad\quad + \delta(k_1-k_6)\delta(k_2-k_4)\delta(k_3-k_5) + \delta(k_1-k_6)\delta(k_2-k_4)\delta(k_3-k_6)
\end{align*}
to write
\begin{equation}
\label{eq:q12expansion}
\|Q_{12}\|_{\mathcal L^2\otimes\mathcal F}^2 = \alpha^{-6} \left( X_1 + \ldots + X_6 \right) \,,
\end{equation}
where, for instance,
\begin{align*}
X_1 := & \int_{\R^3}\int_{\R^3}\int_{\R^3} \left\langle e^{-iH_{\phi_t}(-k')(t-s)} \tilde\psi_{s}, \left( (i\nabla-k-k'-k'')^2+3\alpha^{-2} +M\right)^{-1} \right. \\
& \qquad\qquad\left. \times \left( (i\nabla-k-k')^2+2\alpha^{-2} +M\right)^{-2} e^{-iH_{\phi_t}(-k')(t-s)} \tilde\psi_{s} \right\rangle \,\frac{dk''}{|k''|^2}\,\frac{dk'}{|k'|^2}\,\frac{dk}{|k|^2}
\end{align*}
and
\begin{align*}
X_2 := & \int_{\R^3}\int_{\R^3}\int_{\R^3} \left\langle e^{-iH_{\phi_t}(-k'')(t-s)} \tilde\psi_{s}, \left( (i\nabla-k-k'-k'')^2+3\alpha^{-2} +M\right)^{-1} \right. \\
& \qquad\qquad\times \left( (i\nabla-k-k'')^2+2\alpha^{-2} +M\right)^{-1} \left( (i\nabla-k-k')^2+2\alpha^{-2} +M\right)^{-1} \\
& \qquad\qquad\left. \times \ e^{-iH_{\phi_t}(-k')(t-s)} \tilde\psi_{s} \right\rangle \,\frac{dk''}{|k''|^2}\,\frac{dk'}{|k'|^2}\,\frac{dk}{|k|^2}.
\end{align*}
By the Schwarz inequality we have $|X_2|\leq X_1$ and, similarly,
\begin{equation}
\label{eq:x16}
|X_j|\leq X_1 
\qquad\text{for all}\ j=1,\ldots,6 \,.
\end{equation}
Thus it suffices to control $X_1$.

We first perform the $k''$ integral and then the $k$ integral. We make use of the following bounds.

\begin{lemma}\label{intbounds}
One has
\begin{equation}
\label{eq:hkl1}
\int_{\mathbb{R}^3} \left( (i\nabla-k'')^2+1\right)^{-1} \,\frac{dk''}{|k''|^2} \lesssim 1 \,,
\end{equation}
\begin{equation}
\label{eq:kl1}
\int_{\mathbb{R}^3} \left( (i\nabla_x-k)^2+1\right)^{-2} \frac{dk}{|k|^2} \lesssim \left(-\Delta+1\right)^{-1} \,.
\end{equation}
\end{lemma}

Before proving the lemma, let us see that they provide the desired bounds on $X_1$. First, conjugating \eqref{eq:hkl1} with $e^{i(k+k')\cdot x}$ and assuming that $M+3\alpha^2\geq 1$, we obtain, uniformly in $k,k'\in\R^3$,
\begin{equation}
\label{eq:hkl}
\int_{\mathbb{R}^3} \left( (i\nabla-k-k'-k'')^2+3\alpha^{-2}+M\right)^{-1} \,\frac{dk''}{|k''|^2} \lesssim 1 \,.
\end{equation}
Similarly, conjugating \eqref{eq:kl1} with $e^{ik'\cdot x}$, we obtain, uniformly in $k'\in\R^3$,
\begin{equation}
\label{eq:kl}
\int_{\mathbb{R}^3} \left( (i\nabla_x-k-k')^2+2\alpha^{-2}+M\right)^{-2} \frac{dk}{|k|^2} \lesssim \left((i\nabla-k')^2+1\right)^{-1} \,.
\end{equation}
Inserting \eqref{eq:hkl} and \eqref{eq:kl} into the definition of $X_1$, we obtain
$$
X_1 \lesssim \int_{\R^3} \left\langle e^{-iH_{\phi_t}(-k')(t-s)} \tilde\psi_{s},
\left( (i\nabla-k')^2+1\right)^{-1} e^{-iH_{\phi_t}(-k')(t-s)} \tilde\psi_{s} \right\rangle \,\frac{dk'}{|k'|^2}
$$
Since $(-\Delta +1)^{-1/2} (H_{\phi_t}+M)^{1/2}$ is bounded, uniformly in $t$ (by Corollary \ref{formopdom} and Lemma \ref{wellposedenergy}), we also know that $((i\nabla -k')^2+1)^{-1/2} (H_{\phi_t}(-k')+M)^{1/2}$ is bounded, uniformly in $t$. Thus,
\begin{align*}
X_1 & \lesssim \int_{\R^3} \left\langle e^{-iH_{\phi_t}(-k')(t-s)} \tilde\psi_{s},
\left( H_{\phi_t}(-k') +M\right)^{-1} e^{-iH_{\phi_t}(-k')(t-s)} \tilde\psi_{s} \right\rangle \,\frac{dk'}{|k'|^2} \\
& = \int_{\R^3} \left\langle \tilde\psi_{s},
\left( H_{\phi_t}(-k') +M\right)^{-1} \tilde\psi_{s} \right\rangle \,\frac{dk'}{|k'|^2} \\
& \lesssim \int_{\R^3} \left\langle \tilde\psi_{s},
\left( (i\nabla-k')^2 +M\right)^{-1} \tilde\psi_{s} \right\rangle \,\frac{dk'}{|k'|^2} \,.
\end{align*}
Applying \eqref{eq:kl1} again, we see that the latter expression is bounded by a constant times $\|\tilde\psi_s\|_{\mathcal L^2}^2=1$ by Lemma \ref{wellposedenergy}. This, together with \eqref{eq:q12expansion} and \eqref{eq:x16}, implies that
\begin{equation}
\label{eq:q12}
\|Q_{12}\|_{\mathcal L^2\otimes\mathcal F} \lesssim \alpha^{-3} \,.
\end{equation}

\begin{proof}[Proof of Lemma \ref{intbounds}]
We only prove \eqref{eq:kl1}, since the proof of \eqref{eq:hkl1} is similar and simpler. By applying a Fourier transform we see that we need to prove
$$
\int_{\mathbb{R}^3} \left( (p+k)^2+1\right)^{-2} \frac{dk}{|k|^2} \lesssim \left(p^2+1\right)^{-1} 
\qquad\text{for}\ p\in\R^3 \,.
$$
We split the integral into the regions $4|k|>|p|+1$ and $4|k|\leq |p|+1$. In the first region we bound $|k|^{-2} \leq 16/(|p|+1)^2$ and note that
$$
\int_{\{4|k|>|p|+1\}} \left( (p+k)^2+1\right)^{-2} \,dk \leq
\int_{\R^3} \left( (p+k)^2+1\right)^{-2} \,dk = \int_{\R^3} \left( k^2+1\right)^{-2} \,dk <\infty \,.
$$
In the second region we distinguish the cases $|p|<1$ and $|p|\geq 1$. In the first case we bound
$$
\int_{\{4|k|\leq |p|+1\}} \left( (p+k)^2+1\right)^{-2} \frac{dk}{|k|^2} \leq
\int_{\{4|k|\leq |p|+1\}} \frac{dk}{|k|^2} \leq \int_{\{|k|\leq 1/2\}} \frac{dk}{|k|^2} <\infty \,.
$$
For $|p|\geq 1$ we note that in the second region we have $2|k|\leq |p|$ and therefore $(p+k)^2\geq p^2/4 \geq k^2$. Thus,
$$
\left( (p+k)^2+1\right)^{-2} \leq (p^2/4+1)^{-1} (k^2+1)^{-1} \,.
$$
Since $(k^2+1)^{-1} |k|^{-2}$ is integrable, we obtain again a bound of the required form.
\end{proof}

%%%%%%%%%%%%%%%%%%%%%%%%%%%%%%%%

\subsection*{Bounds on $Q_2,\ldots, Q_5$}

The terms $Q_2,\ldots,Q_4$ are controlled in exactly the same way as $Q_1$. (For $Q_4$ we use the fact that $\|\partial_s\tilde\psi_s\|_{\mathcal H^2} \lesssim 1$ for $t\leq\alpha^2$ by Proposition \ref{THM:wellposedness}.) The argument for $Q_5$ is also similar. In fact, the term involving $\im(\phi_s,\alpha^2\partial_s\phi_s)$ is controlled as before. For the term involving $b^*(\alpha^2\partial_s\phi_s)$ we have to prove a simple extension of Lemma \ref{singular2} where we have operators $b^*(f)b^*_k b^*_{k'}$ with $f\in\mathcal L^2$ (similarly as the second part in Lemma \ref{singular}). Finally, the term involving $b(\alpha^2\partial_s\phi_s)$ can commuted to the right and therefore becomes a less singular term which can be controlled already with Lemma \ref{singular}. These arguments prove \eqref{eq:qbounds} and complete the proof of \eqref{eq:d111}.

%%%%%%%%%%%%%%%%%%%%%%%%%%%%%%%%

\subsection*{Bound on $D_{112}$}

The term $D_{112}$ in \eqref{eq:d112def} contains only one factor $|k'|^{-1}$ and can therefore be controlled essentially by the same method as $D_{01}$, based on Lemma \ref{singular}. In order to create a factor of $(H_{\phi_t}+M)^{-1}$, we integrate by parts in $s_1$. This, however, will create a factor of $\tilde H_{\phi_t}$ in one of the terms. When dealing with $D_{211}$ we will explain how to remove this term by integrating by parts in $s$. Since $\|g_{s,t}\|_\infty\lesssim \alpha^{-2}|t-s|$ and $\|\partial_s g_{s,t}\|_\infty = \|g_s\|_\infty \lesssim \alpha^{-2}$ by Proposition \ref{THM:wellposedness}, this factor behaves will in the bounds. When applying Lemma \ref{singular} we also use $\|\partial_s\tilde\psi_s\|_{\mathcal H^1} \lesssim 1$ from Proposition \ref{THM:wellposedness}; see also the remark at the beginning of Subsection \ref{sec:decomp} concerning the bounds on $\partial_t\tilde\psi_t$. Without going into details we state the final result,
\begin{equation}
\label{eq:d112}
\left\| D_{112} \right\|_{\mathcal L^2\otimes\mathcal F} \lesssim \alpha^{-3} t^2(1+t) \,.
\end{equation}

%%%%%%%%%%%%%%%%%%%%%%%%%%%%%%%%

\subsection*{Bound on $D_{121}$}

Also the term $D_{121}$ in \eqref{eq:d121def} contains only one factor of $|k|^{-1}$ and can be controlled as just sketched for $D_{112}$ and as explained in detail for $D_{211}$. In order to control the terms that appear when integrating by parts in $s$ we make use of $\|\partial_s\sigma_{\tilde\psi_s}\|_{\mathcal L^2} \lesssim 1$ and $\|\partial_s\tilde\psi_s\|_{\mathcal H^1}\lesssim$ from Proposition \ref{THM:wellposedness} in addition to the bounds from Lemma \ref{wellposedenergy}. Moreover, we need an obvious extension of Lemma \ref{singular} to the case with $b^*(f_1)b^*(f_2)b^*_k$, which is proved in the same way. Combining all this, we end up with
\begin{equation}
\label{eq:d121}
\left\| D_{121} \right\|_{\mathcal L^2\otimes\mathcal F} \lesssim \alpha^{-2} t(1+t) \,.
\end{equation}

%%%%%%%%%%%%%%%%%%%%%%%%%%%%%%

\subsection*{Bound on $D_{122}$}

The term $D_{122}$ contains no $|k|^{-1}$ term. Using $\|g_{s,t}\|_\infty \lesssim \alpha^{-2} |t-s|$ for $0\leq s\leq t\leq\alpha^2$ by Proposition \ref{THM:wellposedness} and $\|b(\sigma_{\tilde\psi_s})\Omega\|_{\mathcal F}= \alpha^{-1} \|\sigma_{\tilde\psi_s}\|_2 \lesssim \alpha^{-1}$ by Lemma \ref{wellposedenergy} we obtain immediately
\begin{equation}
\label{eq:d122}
\left\| D_{122} \right\|_{\mathcal L^2\otimes\mathcal F} \lesssim \alpha^{-3} t^3 \,.
\end{equation}

%%%%%%%%%%%%%%%%%%%%%%%%%%%%%%%%%%%
%%%%%%%%%%%%%%%%%%%%%%%%%%%%%%%%%%%

\section{Estimation on $D_2$}\label{sec:remainder2}

%%%%%%%%%%%%%%%%%%%%%%%%%%%%%%%%

\subsection*{Bound on $D_{211}$}

We recall equation \eqref{eq:d211def} for $D_{211}$. In this equation we commute $e^{-ik\cdot x}$ through $e^{-iH_{\phi_t}s_1}$, which introduces again the operator $H_{\phi_t}(k)$ from \eqref{eq:hphitk}, and we commute $b_k$ with $b_{k'}^*$. In this way, we obtain
\begin{align*}
D_{211}= \alpha^{-2} \int_0^{t} \int_{0}^{t-s} \int_{\R^3} e^{i\tilde{H}_{\phi_t}(s+s_1)} e^{-iH_{\phi_t}(k)s_1} W^*(\alpha^2\phi_t)W(\alpha^2\phi_s) \tilde\psi_s\otimes\Omega \,\frac{dk}{|k|^2}\,ds_1\,ds \,.
\end{align*}
The difficulty in controlling $D_{211}$ comes again from the $k$-integral. It is not enough to bound the norm of the integrand as it stands, since $|k|^{-2}$ is not integrable. Thus, we need to gain some extra decay from $e^{-iH_{\phi_t}(k)s_1}$. To get this decay, we integrate by parts in $s_1$ using
\begin{align}\label{eq:downK}
e^{-iH_{\phi_t}(k)s_1}=ie^{iM s_1} \left( H_{\phi_t}(k)+M \right)^{-1}\ \partial_{s_1}e^{-i[H_{\phi_t}(k)+M]s_1}
\end{align}
with a large constant $M>0$ independent of $\alpha$ and $t$. We obtain
\begin{align*}
D_{211}=& i\alpha^{-2} \int_0^{t} \int_{\R^3} e^{i\tilde{H}_{\phi_t}t} \left( H_{\phi_t}(k)+M\right)^{-1} e^{-iH_{\phi_t}(k)(t-s)} \\
& \qquad\qquad\times W^*(\alpha^2\phi_t)W(\alpha^2\phi_s) \tilde\psi_s\otimes\Omega \,\frac{dk}{|k|^2}\,ds \\
& -i\alpha^{-2} \int_0^t \int_{\R^3} e^{i\tilde{H}_{\phi_t}s} \left( H_{\phi_t}(k)+M\right)^{-1} W^*(\alpha^2\phi_t) \\
& \qquad\qquad\times W(\alpha^2\phi_s) \tilde\psi_s\otimes\Omega \,\frac{dk}{|k|^2}\,ds \\ 
& + \alpha^{-2} M \int_0^{t} \int_0^{t-s} \int_{\R^3} e^{i\tilde{H}_{\phi_t}(s+s_1)} \left( H_{\phi_t}(k)+M\right)^{-1} e^{-iH_{\phi_t}(k) s_1} \\
& \qquad\qquad\times W^*(\alpha^2\phi_t)W(\alpha^2\phi_s) \tilde\psi_s\otimes\Omega \,\frac{dk}{|k|^2}\,ds_1\,ds \\
& + \alpha^{-2} \int_0^{t} \int_0^{t-s} \int_{\R^3} e^{i\tilde{H}_{\phi_t}(s+s_1)} \tilde H_{\phi_t} \left( H_{\phi_t}(k)+M\right)^{-1} e^{-iH_{\phi_t}(k) s_1} \\
& \qquad\qquad\times W^*(\alpha^2\phi_t)W(\alpha^2\phi_s) \tilde\psi_s\otimes\Omega \,\frac{dk}{|k|^2}\,ds_1\,ds \\
=& D_{2111} + D_{2112} + D_{2113} + D_{2114} \,,
\end{align*} 
where $D_{211k}$, $k=1,\ldots,4$, are naturally defined.

We first show how to deal with the terms $D_{2111}$, $D_{2112}$ and $D_{2113}$. The term $D_{2114}$ is harder because of the additional factor of $\tilde H_{\phi_t}$.

The following lemma quantifies in which sense the operator $(H_{\phi_t}+M)^{-1}$ leads to additional decay in $k$.

\begin{lemma}\label{LM:k2Psi}
For $u\in\mathcal H^2(\R^3)$,
\begin{align}
\int_{\mathbb{R}^3} \left\|\left( |i\nabla + k|^2 + 1\right)^{-1} u \right\|_{2}\, \frac{dk}{|k|^2} \lesssim  \|u\|_{\mathcal{H}^2} \,.\label{eq:integrable}
\end{align}
\end{lemma}

\begin{proof}
By Fourier transform, we have
$$
\left\|\left( |i\nabla + k|^2 + 1\right)^{-1} u \right\|_{2}^2
= \int_{\R^3} \frac{1}{(1+|p+k|^2)^2 (1+|p|^2)^2} (1+|p|^2)^2|\hat u(p)|^2\,dp
$$
We now observe that
$$
\frac{1}{(1+|p+k|^2)^2 (1+|p|^2)^2} \lesssim \frac{1}{(1+|k|^2)^2} \,.
$$
This can be proved by considering separately the regions where $|p|\leq \frac{1}{2}|k|$ and $|p|\geq \frac{1}{2}|k|$. Thus,
$$
\left\|\left( |i\nabla + k|^2 + 1\right)^{-1} u \right\|_{2}^2 \lesssim \frac{1}{(1+|k|^2)^2} \|u\|_{\mathcal H^2}^2 \,,
$$
and the claimed bound follows by integration over $k$.
\end{proof}

Let us return to the terms $D_{2111}$, $D_{2112}$ and $D_{2113}$. It follows from Corollary \ref{formopdom} by conjugating with the unitary $e^{ik\cdot x}$ that there is an $M>0$ such that the operator $\left( H_{\phi_t}(k)+M\right)^{-1} \left(|i\nabla +k|^2 +1 \right)$ is uniformly bounded in $\alpha$ and $t$. This, together with the boundedness of $\psi_s$ in $\mathcal H^2$ for $s\in [0,\alpha^2]$ from Proposition \ref{THM:wellposedness}, yields
$$
\int_{\R^3} \left\| \left( H_{\phi_t}(k)+M\right)^{-1} \tilde\psi_s \right\|_2 \,\frac{dk}{|k|^2} \lesssim 1 \,, 
$$
and therefore
\begin{equation}
\label{eq:d2111-3}
\| D_{2111} \|_{\mathcal L^2\otimes\mathcal F} \lesssim \alpha^{-2} t \,,
\qquad
\| D_{2112} \|_{\mathcal L^2\otimes\mathcal F} \lesssim \alpha^{-2} t \,,
\qquad
\| D_{2113} \|_{\mathcal L^2\otimes\mathcal F} \lesssim \alpha^{-2} t^2 \,.
\end{equation}

We now turn to the term $D_{2114}$, which contains the operator $\tilde H_{\phi_t}$. The idea is to remove this operator by integrating by parts in $s$ using
\begin{equation}
\label{eq:upK}
\tilde{H}_{\phi_t} e^{i\tilde{H}_{\phi_t}s}=-i \partial_{s}e^{i\tilde{H}_{\phi_t}s} \,.
\end{equation}
This leads to
\begin{align*}
D_{2114}=&-i\alpha^{-2} \int_0^{t} \int_{\R^3} e^{i\tilde{H}_{\phi_t}t} \left( H_{\phi_t}(k)+M\right)^{-1} e^{-iH_{\phi_t}(k) (t-s_1)} \\
& \qquad\qquad\qquad\times W^*(\alpha^2\phi_t)W(\alpha^2\phi_{s_1}) \tilde\psi_{s_1}\otimes\Omega \,\frac{dk}{|k|^2}\,ds_1 \\
& +i\alpha^{-2} \int_0^{t} \int_{\R^3} e^{i\tilde{H}_{\phi_t}s_1} \left( H_{\phi_t}(k)+M\right)^{-1} e^{-iH_{\phi_t}(k) s_1} \\
& \qquad\qquad\qquad\times W^*(\alpha^2\phi_t)W(\alpha^2\phi_0) \tilde\psi_0\otimes\Omega \,\frac{dk}{|k|^2}\,ds_1 \\
& +i \alpha^{-2} \int_0^{t} \int_0^{t-s} \int_{\R^3} e^{i\tilde{H}_{\phi_t}(s+s_1)} \left( H_{\phi_t}(k)+M\right)^{-1} e^{-iH_{\phi_t}(k) s_1}\\
& \qquad\qquad\qquad\times W^*(\alpha^2\phi_t)W(\alpha^2\phi_s) \partial_s \tilde\psi_s\otimes\Omega \,\frac{dk}{|k|^2}\,ds_1\,ds \\
& + i \alpha^{-2} \int_0^{t} \int_0^{t-s} \int_{\R^3} e^{i\tilde{H}_{\phi_t}(s+s_1)} \left( H_{\phi_t}(k)+M\right)^{-1} e^{-iH_{\phi_t}(k) s_1} \\
& \qquad\qquad\qquad\times W^*(\alpha^2\phi_t) \left( \partial_s W(\alpha^2\phi_s)\right) \tilde\psi_s\otimes\Omega \,\frac{dk}{|k|^2}\,ds_1\,ds \,.
\end{align*}

The first three terms on the right side can be bounded by Lemma \ref{LM:k2Psi} together with the uniform boundedness in $\mathcal H^2$ of $\tilde\psi_s$ and $\partial_s\tilde\psi_s$ in $[0,\alpha^2]$ from Proposition \ref{THM:wellposedness}; see also the remark at the beginning of Subsection \ref{sec:decomp} concerning the bounds on $\partial_t\tilde\psi_t$. For the fourth term on the right side we use the formula \eqref{eq:firPW} for $\partial_s W(\alpha^2\phi_s)$. Then the term can be bounded by proceeding in the same way as for $D_{015}$ and using Lemma \ref{LM:k2Psi} together with the fact that $\alpha^2\partial_s\phi_s$ is uniformly bounded in $\mathcal L^2$ for all times by Lemma \ref{wellposedenergy}. To summarize, we obtain
\begin{equation}
\label{eq:d2114}
\left\| D_{2114} \right\|_{\mathcal L^2\otimes\mathcal F} \lesssim \alpha^{-2}t(1+t) \,,
\end{equation}
and, because of \eqref{eq:d2111-3},
\begin{equation}
\label{eq:d211}
\left\| D_{211} \right\|_{\mathcal L^2\otimes\mathcal F} \lesssim \alpha^{-2} t(1+t) \,.
\end{equation}

%%%%%%%%%%%%%%%%%%%%%%%%%%%%%%%%

\subsection*{Bound on $D_{212}$}

The term $D_{212}$ involves a single difficult operator $\int b_{k'}^* e^{ik'\cdot x} |k'|^{-1}\,dk'$ and can be controlled using the technique from bounding $D_{01}$. We first integrate by parts with respect to $s_1$ using \eqref{eq:downK} (with $k=0$) to create a factor of $(H_{\phi_t}+M)^{-1}$. Using this factor we can apply Lemma \ref{singular} as in the bound of $D_{01}$. In one of the terms, however, the integration by parts creates a factor $\tilde H_{\phi_t}$. We remove this operator via \eqref{eq:upK} by integrating by parts in $s$. The factor $g_{s,t}$ and its derivative $\partial_s g_{s,t}=-g_s$ are bounded by Proposition \ref{THM:wellposedness} and do not create any problems. Eventually, this shows that
\begin{equation}
\label{eq:d212}
\left\| D_{212} \right\|_{\mathcal L^2\otimes\mathcal F} \lesssim \alpha^{-3} t^2(1+t) \,.
\end{equation}

%%%%%%%%%%%%%%%%%%%%%%%%%%%%%%%%

\subsection*{Bound on $D_{221}$}

The term $D_{221}$ appears in \eqref{eq:d221def}. We use $b_k b^*(\sigma_{\tilde\psi_s})\Omega = \alpha^{-2} \sigma_{\tilde\psi_s}(k)\Omega$. By the Schwarz inequality, \eqref{eq:apriorisigma} and Lemma \ref{wellposedenergy} we have $\| |k|^{-1} \sigma_{\tilde\psi_s}(k)\Omega \|_1 \lesssim \|\sigma_{\tilde\psi_s}\|_{\mathcal L^2_{(1)}}$ $ \lesssim \|\psi_s\|_{H^1}^2\lesssim 1$. From this one easily concludes that
$$
\| D_{221} \|_{\mathcal L^2\otimes\mathcal F} \lesssim \alpha^{-2} t^2 \,.
$$

%%%%%%%%%%%%%%%%%%%%%%%%%%%%%%%%

\subsection*{Bound on $D_{222}$}

The term $D_{222}$ appears in \eqref{eq:d222def}. Using the bound on $g_{s,t}$ from Proposition \ref{THM:wellposedness} and the fact that $b(\sigma_{\tilde\psi_s})\Omega$ has norm of order $\alpha^{-1}$ by Lemma \ref{wellposedenergy} one obtains
$$
\| D_{222} \|_{\mathcal L^2\otimes\mathcal F} \lesssim \alpha^{-3} t^3 \,.
$$

%%%%%%%%%%%%%%%%%%%%%%%%%%%%%%%%%%%%%
%%%%%%%%%%%%%%%%%%%%%%%%%%%%%%%%%%%%%

\section{Bounds on $D_3$, $D_4$ and $D_5$}\label{sec:remainder3}

We recall that we have already controlled $D_{32}$, $D_{42}$ and $D_{52}$ in \eqref{eq:d32}, \eqref{eq:d42} and \eqref{eq:d52}. The remaining terms $D_{31}$, $D_{41}$ and $D_{51}$ have at most a single term $|k|^{-1}$ and can be bounded using the methods we have already developed. Therefore we will be rather brief.

For each of the terms $D_{311}$, $D_{312}$, $D_{412}$, $D_{511}$ and $D_{512}$ we first integrate by parts in $s_1$ to generate a factor of $(H_{\phi_t}+M)^{-1}$ which allows us to apply Lemma \ref{singular}. One of the terms, however, will involve a $\tilde H_{\phi_t}$, which we have to remove by integrating by parts in $s$. Using the bounds from Lemma \ref{wellposedenergy} and Proposition \ref{THM:wellposedness} we obtain
$$
\|D_{311}\|_{\mathcal L^2\otimes\mathcal F} \lesssim \alpha^{-2} t(1+t) \,,
\qquad
\|D_{312}\|_{\mathcal L^2\otimes\mathcal F} \lesssim \alpha^{-3} t^2(1+t) \,,
$$ 
$$
\|D_{412}\|_{\mathcal L^2\otimes\mathcal F} \lesssim \alpha^{-3} t^2(1+t)
$$
$$
\|D_{511}\|_{\mathcal L^2\otimes\mathcal F} \lesssim \alpha^{-3} t(1+t) \,,
\qquad
\|D_{512}\|_{\mathcal L^2\otimes\mathcal F} \lesssim \alpha^{-4} t^2(1+t+\alpha^{-1} t^2) \,,
$$
The remaining term $D_{411}$ can be immediately bounded by
$$
\|D_{411}\|_{\mathcal L^2\otimes\mathcal F} \lesssim \alpha^{-2} t^2 \,.
$$

%%%%%%%%%%%%%%%%%%%%%%%%%%%%%%%%%%%%%
%%%%%%%%%%%%%%%%%%%%%%%%%%%%%%%%%%%%%

\section{Proof of the almost orthogonality relations}

\subsection{Proof of \eqref{eq:vaToOne}}\label{sub:true1}

We recall that
\begin{align*}
&\left\langle \Omega, \  e^{-iH_{\phi_t}t} D_0\right\rangle_{\mathcal{F}} \\
& \qquad=\left\langle \Omega,\  \int_0^{t} e^{-iH_{\phi_t}(t-s)} P_{\tilde\psi_s}^\bot \int_{\R^3} \left( e^{ik\cdot x} W^*(\alpha^2\phi_t)W(\alpha^2\phi_s)\ b_k^*\ \tilde\psi_s \otimes \Omega \right) \frac{dk}{|k|}\,ds \right\rangle_{\mathcal{F}}.
\end{align*}
We commute the operator $b_k^*$ to the left and use $b_k\Omega=0$. For the commutator we obtain from Corollary \ref{cor:bbb} (with the definition \eqref{eq:gts} of $g_{s,t}$)
\begin{align*}
\left\langle \Omega, \  e^{-iH_{\phi_t}t} D_0\right\rangle_{\mathcal{F}}
= & \left\langle \Omega,\  \int_0^{t} e^{-iH_{\phi_t}(t-s)} P_{\tilde\psi_s}^\bot g_{s,t} W^*(\alpha^2\phi_t)W(\alpha^2\phi_s)\ \tilde\psi_s \otimes \Omega \,ds \right\rangle_{\mathcal{F}} \\
= & \int_0^{t} e^{-iH_{\phi_t}(t-s)} P_{\tilde\psi_s}^\bot g_{s,t}  \tilde\psi_s \left\langle \Omega,\ W^*(\alpha^2\phi_t)W(\alpha^2\phi_s)\ \Omega \right\rangle_{\mathcal{F}} ds \,.
\end{align*}
Thus,
$$
\left\| \left\langle \Omega, \  e^{-iH_{\phi_t}t} D_0\right\rangle_{\mathcal{F}} \right\|_{\mathcal L^2}
\leq t \sup_{0\leq s\leq t} \left\| g_{s,t} \right\|_\infty \left\|\tilde\psi_s\right\|_2 \,.
$$
Thus, by the bound on $g_{s,t}$ from Proposition \ref{THM:wellposedness} and the conservation of the $\mathcal L^2$ norm of $\tilde\psi_s$, we obtain the claimed bound \eqref{eq:vaToOne}.

%%%%%%%%%%%%%%%%%%%%%%%%%%%%%%%%%%%%
%%%%%%%%%%%%%%%%%%%%%%%%%%%%%%%%%%%%

\subsection{Proof of \eqref{eq:psiPerpprop}}\label{sub:true2}

For $\Phi\in\mathcal F$, let
\begin{align*}
& \Theta_\Phi(t):=\left\langle \tilde\psi_t\otimes \Phi, \ e^{-iH_{\phi_t}t} D_0 \right\rangle_{\mathcal{L}^2\otimes \mathcal{F}}\\
& \ =\left\langle \tilde\psi_t\otimes \Phi, \  \int_0^{t} e^{-iH_{\phi_t}(t-s)} P_{\tilde\psi_s}^\bot \int_{\R^3} \left( e^{ik\cdot x} W^*(\alpha^2\phi_t)W(\alpha^2\phi_s)\ b_k^*\ \tilde\psi_s \otimes \Omega \right) \frac{dk}{|k|}\,ds \right\rangle_{\mathcal{L}^2\otimes \mathcal{F}}\nonumber.
\end{align*}
We shall show that
\begin{align}\label{eq:psiPerp}
|\Theta_{\Phi}(t)|\lesssim \alpha^{-2}t^2 \left(1+\alpha^{-2} t^2\right) \|\Phi\|_{\mathcal{F}} \,,
\end{align}
which by duality implies \eqref{eq:psiPerpprop}.

Our goal will be to derive an ordinary differential equation for $\Theta_\Phi$. We use the presence of the operator $P_{\tilde\psi_s}^\bot$ to obtain (with inner products in $\mathcal L^2\otimes\mathcal F$)
\begin{align*}
& \partial_{t}\Theta_\Phi= \left\langle \partial_{t}\tilde\psi_t\otimes \Phi,\ \int_0^{t} e^{-iH_{\phi_t}(t-s)} P_{\tilde\psi_s}^\bot \int_{\R^3} \left( e^{ik\cdot x} W^*(\alpha^2\phi_t)W(\alpha^2\phi_s)\ b_k^*\ \tilde\psi_s \otimes \Omega \right) \frac{dk}{|k|}\,ds \right\rangle \nonumber\\
&\ +\left\langle \tilde\psi_t\otimes \Phi, \  \int_0^{t} \left( \partial_t e^{-iH_{\phi_t}(t-s)} \right) P_{\tilde\psi_s}^\bot \int_{\R^3} \left( e^{ik\cdot x} W^*(\alpha^2\phi_t)W(\alpha^2\phi_s)\ b_k^*\ \tilde\psi_s \otimes \Omega \right) \frac{dk}{|k|}\,ds \right\rangle \nonumber\\
&\ +\left\langle \tilde\psi_t\otimes \Phi, \  \int_0^{t} e^{-iH_{\phi_t}(t-s)} P_{\tilde\psi_s}^\bot \int_{\R^3} \left( e^{ik\cdot x} \left(\partial_t W^*(\alpha^2\phi_t) \right) W(\alpha^2\phi_s)\ b_k^*\ \tilde\psi_s \otimes \Omega \right) \frac{dk}{|k|}\,ds \right\rangle.
\end{align*}
For the first term we use equation \eqref{eq:defParticlemod} for $\partial_t\tilde\psi_t$. In the second term, we compute, using Duhamel's formula,
\begin{align*}
& \partial_t e^{-iH_{\phi_t}(t-s)} = -iH_{\phi_t} e^{-iH_{\phi_t}(t-s)} - i \int_0^{t-s} e^{-iH_{\phi_t}(t-s-s_1)} \left(\partial_t H_{\phi_t}\right) e^{-iH_{\phi_t}s_1} \,ds_1 \\
& \quad = -i\left( H_{\phi_t} +(t-s) \partial_t \|\phi_t\|_2^2 \right) e^{-iH_{\phi_t}(t-s)} - i \int_0^{t-s} e^{-iH_{\phi_t}(t-s-s_1)} \left(\partial_t V_{\phi_t}\right) e^{-iH_{\phi_t}s_1} \,ds_1 \,.
\end{align*}
Note that the part involving $H_{\phi_t}$ will cancel the contribution from the first term, except for part of the constant $\omega(t)$. Finally, for the third term we use Lemma \ref{LM:derivative} and Lemma \ref{LM:commRe} to obtain
\begin{align*}
\partial_t W^*(\alpha^2\phi_t) W(\alpha^2\phi_s) = & \alpha^2 W^*(\alpha^2\phi_t) \left[ b(\partial_t\phi_t) - b^*(\partial_t\phi_t) + i \im \left( \phi_t,\partial_t\phi_t \right) \right] W(\alpha^2\phi_s) \\
= & \alpha^2 W^*(\alpha^2\phi_t)W(\alpha^2\phi_s) \left[ b(\partial_t\phi_t) - b^*(\partial_t\phi_t) \right. \\
& \qquad\qquad\qquad\qquad \left. + 2i\im \left(\partial_t\phi_t,\phi_s\right) + i \im \left( \phi_t,\partial_t\phi_t \right) \right] \\
= & \alpha^2 W^*(\alpha^2\phi_t)W(\alpha^2\phi_s) \left[ b(\partial_t\phi_t) - b^*(\partial_t\phi_t) \right. \\
& \qquad\qquad\qquad\qquad \left. + 2i\im \left(\partial_t\phi_t,\phi_s-\phi_t\right) + i \im \left( \partial_t\phi_t,\phi_t \right) \right].
\end{align*}
Putting all this into the above formula, we obtain
$$
\partial_{t}\Theta_\Phi= M_1 + M_2 + M_3 \,,
$$
where the terms $M_1$, $M_2$ and $M_3$ are defined, using the notation
$$
\Phi_{s,t} : =  W^*(\alpha^2\phi_s) W(\alpha^2\phi_t)\Phi \,,
$$
by
\begin{align*}
M_1(t) & := -i \int_0^t \int_0^{t-s} \left\langle \tilde\psi_t\otimes \Phi_{s,t},\ e^{-iH_{\phi_t}(t-s-s_1)} \left(\partial_t V_{\phi_t}\right) e^{-iH_{\phi_t}s_1} \right.\\
& \qquad\qquad\qquad\qquad\qquad \left. P_{\tilde\psi_s}^\bot \int_{\R^3} \left( e^{ik\cdot x} \ b_k^*\ \tilde\psi_s \otimes \Omega \right) \frac{dk}{|k|} \right\rangle ds_1\,ds \,, \\
M_2(t) & := \alpha^2 \int_0^{t} \left\langle \tilde\psi_t\otimes \Phi_{s,t}, \ e^{-iH_{\phi_t}(t-s)} \right. \\
& \qquad\qquad\qquad\qquad\qquad \left. P_{\tilde\psi_s}^\bot \int_{\R^3} \left( e^{ik\cdot x} \left( b(\partial_t\phi_t) - b^*(\partial_t\phi_t) \right) b_k^*\ \tilde\psi_s \otimes \Omega \right) \frac{dk}{|k|} \right\rangle ds \,,\\
M_3(t) & := \int_0^t m(s,t) \left\langle \tilde\psi_t\otimes \Phi_{s,t},\  e^{-iH_{\phi_t}(t-s)} P_{\tilde\psi_s}^\bot \int_{\R^3} \left( e^{ik\cdot x} \ b_k^*\ \tilde\psi_s \otimes \Omega \right) \frac{dk}{|k|} \right\rangle ds
\end{align*}
with
$$
m(s,t) := -i (t-s)\partial_t \|\phi_t\|_2^2 + 2i \alpha^2 \im\left(\partial_t\phi_t,\phi_s-\phi_t\right) \,.
$$
Since $\Theta_\Phi(0)=0$, we conclude that
\begin{equation}
\label{eq:thetaint}
\Theta_\Phi(t) = \int_0^t \left( M_1(s)+ M_2(s)+M_3(s) \right)ds \,.
\end{equation}
Below we shall show that
\begin{equation}
\label{eq:mbounds}
|M_1(t)|\lesssim \alpha^{-3} t^2 \|\Phi\|_{\mathcal F} \,,
\quad
|M_2(t)|\lesssim \alpha^{-2} t \|\Phi\|_{\mathcal F} \,,
\quad
|M_3(t)|\lesssim \alpha^{-3} t^2 \|\Phi\|_{\mathcal F} \,.
\end{equation}
Together with \eqref{eq:thetaint} this will prove \eqref{eq:psiPerp} and therefore \eqref{eq:psiPerpprop}.

\subsection*{Bound on $M_1$}

Using the fact that $P_{\tilde\psi_s}^\bot=1-|\tilde\psi_s\rangle\langle\tilde\psi_s|$ (see the proof of Lemma \ref{LM:newform}) we decompose
$$
M_1 = M_{11} - M_{12} \,,
$$
where
\begin{align*}
M_{11}(t) & := -i \int_0^t \int_0^{t-s} \left\langle \tilde\psi_t\otimes \Phi_{s,t},\ e^{-iH_{\phi_t}(t-s-s_1)} \left(\partial_t V_{\phi_t}\right) e^{-iH_{\phi_t}s_1} \right. \\
& \qquad\qquad\qquad\qquad\qquad\qquad \left. \times\int_{\R^3} \left( e^{ik\cdot x} \ b_k^*\ \tilde\psi_s \otimes \Omega \right) \frac{dk}{|k|} \right\rangle_{\mathcal{L}^2\otimes \mathcal{F}} ds_1\,ds
\end{align*}
and, with $\sigma_{\tilde\psi_s}$ from \eqref{eq:forcing},
$$
M_{12}(t) := -i \int_0^t \!\int_0^{t-s} \!\! \left\langle \tilde\psi_t,\ e^{-iH_{\phi_t}(t-s-s_1)} \left(\partial_t V_{\phi_t}\right) e^{-iH_{\phi_t}s_1} \tilde\psi_s\right\rangle_{\mathcal L^2}\! \left\langle \Phi_{s,t}, b^*(\sigma_{\tilde\psi_s})\ \Omega \right\rangle_{\mathcal{F}} ds_1\,ds.
$$
The second term is easy to control. In fact, the a-priori bounds from Lemma \ref{wellposedenergy} together with $\|\partial_t V_{\phi_t}\|_\infty \lesssim \alpha^{-2}$ from \eqref{eq:vpartialtbounds} imply
$$
\left| \left\langle \tilde\psi_t,\ e^{-iH_{\phi_t}(t-s-s_1)} \left(\partial_t V_{\phi_t}\right) e^{-iH_{\phi_t}s_1} \tilde\psi_s\right\rangle_{\mathcal L^2} \right| \lesssim \alpha^{-2}
$$
and
$$
\left| \left\langle \Phi_{s,t}, b^*(\sigma_{\tilde\psi_s})\ \Omega \right\rangle_{\mathcal{F}} \right| \lesssim \alpha^{-1} \|\Phi\|_{\mathcal F} \,.
$$
This yields a bound of the form \eqref{eq:mbounds}.

We now bound the integrand in $M_{11}$. We have
\begin{align*}
& \left| \left\langle \tilde\psi_t\otimes \Phi_{s,t},\ e^{-iH_{\phi_t}(t-s-s_1)} \left(\partial_t V_{\phi_t}\right) e^{-iH_{\phi_t}s_1} \int_{\R^3} \left( e^{ik\cdot x} \ b_k^*\ \tilde\psi_s \otimes \Omega \right) \frac{dk}{|k|} \right\rangle_{\mathcal{L}^2\otimes \mathcal{F}} \right| \\
& \qquad \leq \left\| ( H_{\phi_t}+M)^{1/2} \tilde\psi_t\otimes \Phi_{s,t} \right\| \left\| (H_{\phi_t}+M)^{-1/2} \left(\partial_t V_{\phi_t}\right) (H_{\phi_t}+M)^{1/2} \right\| \\
& \qquad \qquad \times \left\| (H_{\phi_t} +M)^{-1/2}  \int_{\R^3} \left( e^{ik\cdot x} \ b_k^*\ \tilde\psi_s \otimes \Omega \right) \frac{dk}{|k|} \right\|
\end{align*}
By Corollary \ref{formopdom} and an easy modification of its proof, for $M$ sufficiently large (but independent of $t$ and $\alpha$), the operators $(H_{\phi_t}+M)^{\pm 1/2} (-\Delta+1)^{\mp 1/2}$ are both bounded uniformly in $t$. Therefore Lemma \ref{singular} and the a-priori bounds from Lemma \ref{wellposedenergy} yield
\begin{align*}
& \left| \left\langle \tilde\psi_t\otimes \Phi_{s,t},\ e^{-iH_{\phi_t}(t-s-s_1)} \left(\partial_t V_{\phi_t}\right) e^{-iH_{\phi_t}s_1} \int_{\R^3} \left( e^{ik\cdot x} \ b_k^*\ \tilde\psi_s \otimes \Omega \right) \frac{dk}{|k|} \right\rangle_{\mathcal{L}^2\otimes \mathcal{F}} \right| \\
& \qquad \lesssim \alpha^{-1} \|\tilde\psi_t\|_{\mathcal H^1} \|\Phi\|_{\mathcal F} \left\| (-\Delta+1)^{-1/2} \left(\partial_t V_{\phi_t}\right) (-\Delta+1)^{1/2} \right\| \|\psi_s\|_{\mathcal H^1} \\
& \qquad \lesssim \alpha^{-1} \|\Phi\|_{\mathcal F} \left\| (-\Delta+1)^{-1/2} \left(\partial_t V_{\phi_t}\right) (-\Delta+1)^{1/2} \right\| \,.
\end{align*}
Finally, using the fact that $\|\nabla\partial_t V_{\phi_t}\|_\infty \lesssim \alpha^{-2}$ (see \eqref{eq:vpartialtbounds}), we obtain that the operator appearing in this bound has norm $\lesssim \alpha^{-2}$. Thus, we finally obtain
$$
\left| \left\langle \tilde\psi_t\otimes \Phi_{s,t},\ e^{-iH_{\phi_t}(t-s-s_1)} \left(\partial_t V_{\phi_t}\right) e^{-iH_{\phi_t}s_1} \int_{\R^3} \left( e^{ik\cdot x} \ b_k^*\ \tilde\psi_s \otimes \Omega \right) \frac{dk}{|k|} \right\rangle_{\mathcal{L}^2\otimes \mathcal{F}} \right| \lesssim \alpha^{-3} \,,
$$
which, when integrated over $s_1$ and $s$, leads to the bound in \eqref{eq:mbounds}.

\subsection*{Bound on $M_2$}

As for $M_1$, we use $P_{\tilde\psi_s}^\bot=1-|\tilde\psi_s\rangle\langle\tilde\psi_s|$ to decompose
$$
M_2 = M_{21} - M_{22}
$$
with
$$
M_{21}(t) := \alpha^2\! \int_0^{t}\! \left\langle \tilde\psi_t\otimes \Phi_{s,t}, \ e^{-iH_{\phi_t}(t-s)}\! \int_{\R^3}\! \left( e^{ik\cdot x} \left( b(\partial_t\phi_t) - b^*(\partial_t\phi_t) \right) b_k^*\ \tilde\psi_s \otimes \Omega \right) \frac{dk}{|k|} \right\rangle\! ds
$$
and, with $\sigma_{\tilde\psi_s}$ from \eqref{eq:forcing},
$$
M_{22}(t) := \alpha^2 \int_0^{t} \left\langle \tilde\psi_t , \ e^{-iH_{\phi_t}(t-s)} \tilde\psi_s\right\rangle_{\mathcal{L}^2} \left\langle \Phi_{s,t}, \left( b(\partial_t\phi_t) - b^*(\partial_t\phi_t) \right) b^*(\sigma_{\tilde\psi_s})\ \Omega \right\rangle_{\mathcal{F}} \,ds \,.
$$
Once again the bound on $M_{22}$ is straightforward. Namely, we commute $b^*(\sigma_{\tilde\psi_s})$ to the left through $b(\partial_t\phi_t) - b^*(\partial_t\phi_t)$ and obtain
\begin{align*}
& \left\langle \Phi_{s,t}, \left( b(\partial_t\phi_t) - b^*(\partial_t\phi_t) \right) b^*(\sigma_{\tilde\psi_s})\ \Omega \right\rangle_{\mathcal{F}} \\
& \qquad 
= - \left\langle \Phi_{s,t}, b^*(\sigma_{\tilde\psi_s}) b^*(\partial_t\phi_t) \Omega \right\rangle_{\mathcal{F}} 
+ \alpha^{-2} (\partial_t\phi_t,\sigma_{\tilde\psi_s}) \left\langle \Phi_{s,t}, \Omega \right\rangle_{\mathcal{F}} \,.
\end{align*}
By similar computations as for instance in the bound on $D_{32}$ and by the a-priori bounds from Lemma \ref{wellposedenergy} we obtain
$$
\left| \left\langle \Phi_{s,t}, \left( b(\partial_t\phi_t) - b^*(\partial_t\phi_t) \right) b^*(\sigma_{\tilde\psi_s})\ \Omega \right\rangle_{\mathcal{F}} \right| \lesssim \alpha^{-2}  \|\Phi\|_{\mathcal F}\|\sigma_{\tilde\psi_s}\| \|\partial_t\phi_t\|
\lesssim \alpha^{-4} \|\Phi\|_{\mathcal F} \,.
$$
By the conservation of the $\mathcal L^2$ norm of $\tilde\psi_t$ we conclude
$$
\left|M_{22}(t)\right| \lesssim \alpha^{-2} t \|\Phi\|_{\mathcal F} \,,
$$
which is of the form claimed in \eqref{eq:mbounds}.

We now discuss $M_{21}$. Again we commute $b_k^*$ to the left through $b(\partial_t\phi_t) - b^*(\partial_t\phi_t)$ and obtain
$$
M_{21} = M_{211} + M_{212} \,,
$$
where
$$
M_{211}(t) := - \alpha^2 \int_0^{t} \left\langle \tilde\psi_t\otimes \Phi_{s,t}, \ e^{-iH_{\phi_t}(t-s)} \int_{\R^3} \left( e^{ik\cdot x} b_k^* \ b^*(\partial_t\phi_t) \ \tilde\psi_s \otimes \Omega \right) \frac{dk}{|k|} \right\rangle_{\mathcal{L}^2\otimes \mathcal{F}} \,ds
$$
and, with $g_s$ from \eqref{eq:gs},
$$
M_{212}(t) := \int_0^{t} \left\langle \tilde\psi_t, \ e^{-iH_{\phi_t}(t-s)} g_s \tilde\psi_s \right\rangle_{\mathcal{L}^2}
\left\langle \Phi_{s,t}, \ \Omega \right\rangle_{\mathcal{F}}
 \,ds \,.
$$
Since $\|g_s\|_\infty \lesssim \alpha^{-2}$ by Proposition \ref{THM:wellposedness}, we obtain immediately
$$
\left|M_{212}(t)\right| \lesssim \alpha^{-2} t \|\Phi\|_{\mathcal F} \,.
$$
To control $M_{211}$ we bound
\begin{align*}
& \left| \left\langle \tilde\psi_t\otimes \Phi_{s,t}, \ e^{-iH_{\phi_t}(t-s)} \int_{\R^3} \left( e^{ik\cdot x} b_k^* \ b^*(\partial_t\phi_t) \ \tilde\psi_s \otimes \Omega \right) \frac{dk}{|k|} \right\rangle_{\mathcal{L}^2\otimes \mathcal{F}} \right| \\
& \qquad \leq \left\| (H_{\phi_t}+M)^{1/2} \tilde\psi_t\otimes \Phi_{s,t} \right\| \left\| (H_{\phi_t}+M)^{-1/2} \int_{\R^3} \left( e^{ik\cdot x} b_k^* \ b^*(\partial_t\phi_t) \ \tilde\psi_s \otimes \Omega \right) \frac{dk}{|k|} \right\| \,.
\end{align*}
As for $M_{11}$ we use Lemma \ref{wellposedenergy} and Corollary \ref{formopdom} (and a simple extension of its proof) to choose $M$ large enough, but independent of $t$ and $\alpha$, so that $(H_{\phi_t}+M)^{\pm 1/2} (-\Delta+1)^{\mp 1/2}$ are both bounded uniformly in $t$. Therefore Lemma \ref{singular} and the a-priori bounds from Lemma \ref{wellposedenergy} yield
\begin{align*}
& \left| \left\langle \tilde\psi_t\otimes \Phi_{s,t}, \ e^{-iH_{\phi_t}(t-s)} \int_{\R^3} \left( e^{ik\cdot x} b_k^* \ b^*(\partial_t\phi_t) \ \tilde\psi_s \otimes \Omega \right) \frac{dk}{|k|} \right\rangle_{\mathcal{L}^2\otimes \mathcal{F}} \right| \\
& \qquad \lesssim \alpha^{-2} \|\tilde \psi_t\|_{\mathcal H^1} \|\Phi\|_{\mathcal F} \|\partial_t\phi_t\|_{\mathcal L^2} \|\tilde\psi_s\|_{\mathcal H^1} \\
& \qquad \lesssim \alpha^{-4} \|\Phi\|_{\mathcal F} \,.
\end{align*}
This, when integrated over $s$ and multiplied by $\alpha^2$, leads to the bound in \eqref{eq:mbounds}.

\subsection*{Bound on $M_3$}

The a-priori bounds from Lemma \ref{wellposedenergy} yield
$$
|m(s,t)| \lesssim \alpha^{-2} |t-s| \,.
$$
Moreover, applying Lemma \ref{singular} as in the bound on $M_{21}$ we find that the absolute value of the inner product in the integral defining $M_3$ is bounded by a constant times $\alpha^{-1}\|\Phi\|_{\mathcal F}$. This yields the bound in \eqref{eq:mbounds}.

\medskip

This concludes the proof of \eqref{eq:psiPerpprop}.

%%%%%%%%%%%%%%%%%%%%%%%%%%%%%%%%%%%
%%%%%%%%%%%%%%%%%%%%%%%%%%%%%%%%%%%

\appendix

%%%%%%%%%%%%%%%%%%%%%%%%%%%%%%%%%%%%%%%%%%%%%%%%%

\section{Some Properties of the Weyl operators}

In this appendix we collect some standard properties of the Weyl operators $W(f)$ defined in \eqref{eq:defW} in terms of $b(f)$ and $b^*(f)$. They are well-known, but we provide proofs for the sake of completeness. We recall that the commutation relations for $b_k$ and $b_k^*$ involve a factor $\alpha^{-2}$.

\begin{lemma}\label{LM:commRe} 
$b_{k}$, $b_k^*$ and $W(f)$ satisfy the following relations,
\begin{align}
 b_k W(f)=W(f)\left( b_k+\alpha^{-2} f(k) \right)
 \qquad\text{and}\qquad
 b_k^* W(f)=& W(f) \left( b_k^*+\alpha^{-2}\bar f(k) \right).
\end{align}
\end{lemma}

\begin{proof}
For $t>0$ we consider the operators 
\begin{align}
F_{t}:=W(t\phi)=e^{t\left(b^*(f)-b(f)\right)} \,,\label{eq:Fts}
\end{align}
which satisfy
\begin{align}
\partial_{t}F_t=\left(b^*(f)-b(f)\right) F_t \,,
\qquad F_0=\text{Id} \,.\nonumber
\end{align}
Multiplying by $b_k$ and using the commutation relations we obtain the following equation for $b_k F_t$,
\begin{align}
\partial_{t} b_k F_t= \left( b^*(f)-b(f)\right) b_{k}F_t+f(k)F_t\,,
\qquad b_k F_0=& b_k \,. \nonumber
\end{align}
Therefore, by Duhamel's principle applied to the latter equation,
\begin{align}
 b_k F_t=e^{t\left(b^*(f)-b(f)\right)} b_k+ f(k) \int_{0}^{t}e^{(t-s)\left(b^*(f)-b(f)\right)} F_s\, ds \,. \nonumber
\end{align}
Recalling the definition of $F_t$ in \eqref{eq:Fts} we can rewrite this as
\begin{align}
 b_k F_t=F_t b_k+tf(k) F_t \,.
\end{align} 
At $t=1$ we obtain the first identity in the lemma. The second one is proved similarly.
\end{proof}

By applying Lemma \ref{LM:commRe} twice, we obtain

\begin{corollary}\label{cor:bbb}
\begin{align*}
\left[b_k^*, \ W^*(f)W(g)\right]= & -\alpha^{-2} \left(\bar f(k)-\bar g(k)\right) \ W^*(f)W(g) \,,\\
\left[b_k, \ W^*(f)W(g)\right]= & \alpha^{-2} \left(f(k)-g(k)\right)\ W^*(f) W(g) \,.
\end{align*}
\end{corollary}

Next, we'll consider the case where $f$ depends (differentiably) on a parameter.

\begin{lemma}\label{LM:derivative}
\begin{align}
\partial_{t}W(f_t) = &\frac{\alpha^{-2}}{2} \left( (f_t,\partial_t f_t) - (\partial_t f_t, f_t) \right) \ W(f_t) + W(f_t) \ \left( b^*(\partial_t f_t)- b(\partial_t f_t)\right) , \label{eq:firPW} \\
\partial_{t}W(f_t)= &-\frac{\alpha^{-2}}{2} \left( (f_t,\partial_t f_t) - (\partial_t f_t, f_t) \right) \ W(f_t) + \left( b^*(\partial_t f_t) - b(\partial_t f_t) \right) \ W(f_t)  \,.\label{eq:secPW}
\end{align}
\end{lemma}

\begin{proof}
For $s>0$ we consider the operators
\begin{align}
F(s,t):=W(s f_t) \,,
\end{align} 
which satisfy
\begin{align}
\partial_{s}F(s,t)=& \left( b^*(f_t)-b(f_t) \right) F(s,t) \,,
\qquad F(0,t)=\text{Id} \,. \nonumber
\end{align}
We differentiate this equation with respect to $t$ and obtain
\begin{align*}
\partial_{s}\partial_{t} F(s,t)=& \left( b^*(f_t)-b(f_t) \right) \partial_t F(s,t)+ \left(b^*(\partial_t f_t)-b(\partial_t f_t) \right) F(s,t) \,,\\
\partial_t F(0,t)=&0 \,.
\end{align*}
Therefore, by Duhamel's principle,
\begin{align}
\partial_{t}F(s,t)=& \int_{0}^{s} e^{\left( b^*(f_t)- b(f_t) \right) (s-s_1)} \left( b^*(\partial_t f_t)-b(\partial_t f_t) \right) F(s_1,t)\, ds_1 \nonumber\\
=& \int_{0}^{s} W((s-s_1) f_t) \left( b^*(\partial_t f_t)-b(\partial_t f_t) \right) W(s_1 f_t) \ ds_1 \,. \nonumber
\end{align}
In order to simplify the integrand we now use Lemma \ref{LM:commRe} and obtain
\begin{align*}
\left(b^*(\partial_t f_t)-b(\partial_t f_t) \right) W(s_1 f_t)=& \alpha^{-2} W(s_1 f_t)\ s_1 \left( (f_t,\partial_t f_t) -(\partial_t f_t,f_t) \right) \nonumber\\
&+W(s_1 f_t)\ \left( b^*(\partial_t f_t)-b(\partial_t f_t) \right) \,.
\end{align*}
If we insert this into the above formula for $\partial_t F(s,t)$, we obtain
$$
\partial_{t}F(s,t)= \alpha^{-2} \frac{s^2}{2} W(s f_t) \ \left( (f_t,\partial_t f_t) -(\partial_t f_t,f_t) \right)
+s W(s f_t)\ \left( b^*(\partial_t f_t)-b(\partial_t f_t) \right) \,.
$$
At $s=1$ we obtain the first identity in the lemma. The second one is proved similarly.
\end{proof}

\begin{lemma}\label{overlap}
For any $f,g\in\mathcal F$,
$$
\langle\Omega,W^*(g)W(f)\Omega\rangle = e^{i\alpha^{-2} \mathrm{Im} (g,f) -\alpha^{-2}\|f-g\|^2/2} \,.
$$
\end{lemma}

\begin{proof}
Let $f_t:=tf+(1-t)g$ and $F(t):= \langle\Omega,W^*(g)W(f_t)\Omega\rangle$. By Lemma \ref{LM:derivative}, using that $\im(f_t,\partial_t f_t) = \im(f_t,f-g) = \im(g,f)$,
$$
\partial_t F(t) = \langle\Omega,W^*(g)W(f_t) \left( b^*(f-g) +i\alpha^{-2} \im(g,f)\right) \Omega\rangle \,.
$$
Next, by Corollary \ref{cor:bbb}, since $(g-f_t,f-g)=-t \|f-g\|^2$,
\begin{align*}
W^*(g)W(f_t) b^*(f-g) =  & b^*(f-g)W^*(g)W(f_t) + \alpha^{-2}(g-f_t,f-g) W^*(g)W(f_t) \,,
\end{align*}
so
\begin{align*}
\partial_t F(t) & = \left( -\alpha^{-2}t\|f-g\|^2 + i\alpha^{-2}\im(g,f) \right) F(t) \,.
\end{align*}
Since $F(0)=1$, we conclude that
$$
F(t) = e^{-\alpha^{-2}t^2\|f-g\|^2/2 + i\alpha^{-2}t\im(g,f)} \,,
$$
which, at $t=1$, gives the assertion.
\end{proof}

%%%%%%%%%%%%%%%%%%%%%%%%%%%%
%%%%%%%%%%%%%%%%%%%%%%%%%%%%

\section{The effective Schr\"odinger operator}

In this appendix we investigate the operator and form domains of the effective Schr\"odinger operator $H_\phi$ from \eqref{eq:effso} with potential $V_\phi$ from \eqref{eq:effpot}.

\begin{lemma}\label{infiniformop}
For every $A>0$ and $\epsilon>0$ there is an $M>0$ such that if $\|\phi\|\leq A$, then for all $\psi\in\mathcal H^1(\R^3)$
$$
\left\| |V_\phi|^{1/2} \psi\right\| \leq \epsilon \left\|(-\Delta+M)^{1/2}\psi\right\|
$$
and for all $\psi\in\mathcal H^2(\R^3)$
$$
\left\| V_\phi \psi\right\| \leq \epsilon \left\|(-\Delta+M)\psi\right\| \,.
$$
\end{lemma}

\begin{proof}
As in \cite[Sec.~2.1]{FrankSchlein2013}, the Hardy--Littlewood--Sobolev inequality implies that
\begin{equation}
\label{eq:hls}
\| V_\phi \|_{6} \lesssim \|\phi\|_2 \,.
\end{equation}
This implies, by the H\"older and Sobolev inequalities,
$$
\int_{\R^3} |V_\phi| |\psi|^2 \,dx \leq \| V_\phi \|_6 \| \psi \|_{12/5}^2
\lesssim \|\phi\|_2 \|\nabla\psi\|_2^{1/2} \|\psi\|_2^{3/2} 
$$
and
$$
\int_{\R^3} |V_\phi|^2 |\psi|^2 \,dx \leq \| V_\phi \|_6^2 \| \psi \|_{3}^2
\lesssim \|\phi\|_2^2 \|\Delta\psi\|_2^{1/2} \|\psi\|_2^{3/2} \,. 
$$
These bounds easily imply the assertions of the lemma.
\end{proof}

\begin{corollary}\label{formopdom}
For every $A>0$ there are $M>0$ and $C>0$ such that if $\|\phi\|_2\leq A$ then for all $f\in\mathcal L^2(\R^3)$
$$
\left\| \left( H_\phi + M\right)^{-1/2} f \right\|_2 \leq C \left\| \left( -\Delta + 1\right)^{-1/2} f \right\|_2 
$$
and
$$
\left\| \left( H_\phi + M\right)^{-1} f \right\|_2 \leq C \left\| \left( -\Delta + 1\right)^{-1} f \right\|_2 \,.
$$
\end{corollary}

\begin{proof}
To prove the first assertion, we write
$$
\left( H_\phi +M \right)^{-1} = \left(-\Delta+M\right)^{-\frac12} \left( 1+ \left(-\Delta+M\right)^{-\frac12} V_\phi \left(-\Delta+M\right)^{-\frac12} \right)^{-1} \left(-\Delta+M\right)^{-\frac12}
$$
and note that according to Lemma \ref{infiniformop} we can choose $M$ such that $\|\phi\|\leq A$ implies $\|\left(-\Delta+M\right)^{-1/2} V_\phi \left(-\Delta+M\right)^{-1/2}\| \leq \epsilon^2$. Similarly, for the second assertion we write
$$
\left( H_\phi +M \right)^{-1} = \left( 1+ \left(-\Delta+M\right)^{-1} V_\phi \right)^{-1} \left(-\Delta+M\right)^{-1}
$$
and choose $M$ such that $\|\phi\|\leq A$ implies $\|\left(-\Delta+M\right)^{-1} V_\phi\|\leq\epsilon$.
\end{proof}

%%%%%%%%%%%%%%%%%%%%%%%%%%%%
%%%%%%%%%%%%%%%%%%%%%%%%%%%%

\section{Well-posedness of the Landau--Pekar equations}\label{sec:wellposedness}

In this appendix we prove Lemma \ref{wellposedenergy} and Proposition \ref{THM:wellposedness}. Recall that the weighted spaces $\mathcal L^2_{(m)} = \mathcal L^2(\R^3;(1+k^2)^m\,dk)$ were introduced in \eqref{eq:weighted}. We begin with some bounds on the coupling terms $V_\phi$ and $\sigma_\psi$ introduced in \eqref{eq:effpot} and \eqref{eq:forcing}.

\begin{lemma}
We have
\begin{equation}
\label{eq:aprioriv}
\left\|\partial^\beta V_\phi\right\|_\infty \lesssim \|\phi\|_{\mathcal L^2_{|\beta|+1}}
\qquad
\text{for all}\ \beta\in\mathbb N_0^3 \,,
\end{equation}
\begin{equation}
\label{eq:apriorisigma}
\|\sigma_\psi\|_{L^2_{(1)}} \lesssim \|\psi\|_{\mathcal H^1}^2 \,,
\qquad
\|\sigma_\psi\|_{L^2_{(3)}} \lesssim \|\psi\|_{\mathcal H^2}^2 \,.
\end{equation}
\end{lemma}

\begin{proof}
By the Schwarz inequality,
$$
|\partial^\beta V_\phi(x)|\leq 2 \int_{\R^3} |k|^{|\beta|-1} |\phi(k)|\,dk \leq 2 \|\phi\|_{\mathcal L^2_{|\beta|+1}} \left( \int_{\R^3} \frac{|k|^{2(|\beta|-1)}\,dk}{(1+k^2)^{2(|\beta|+1)}} \right)^{1/2}
$$
and the last integral is finite.

We have
$$
\|\sigma_\psi \|_2^2 = \left\| \frac{1}{|k|}\int_{\mathbb{R}^3} |\psi(x)|^2e^{ik\cdot x} \,dx \right\|_2^2
= 2\pi^2 \iint_{\R^3\times\R^3} \frac{|\psi(x)|^2\, |\psi(y)|^2}{|x-y|} \,dx\,dy \,.
$$
By the Hardy--Littlewood--Sobolev inequality, this is bounded by a constant times $\| |\psi|^2 \|_{6/5}^2=\|\psi\|_{12/5}^4$, which, by the Sobolev embedding theorem, is bounded by a constant times $\|\psi\|^4_{\mathcal H^1}$. Moreover, by Plancherel,
$$
\|\sigma_\psi\|_{\mathcal L^2(|k|^{2m})}^2 
= \int_{\R^3} |k|^{2(m-1)} \left|\int_{\R^3} |\psi|^2 e^{ik\cdot x}\,dx \right|^2 \,dk
= (2\pi)^3 (|\psi|^2, (-\Delta)^{m-1} |\psi|^2) \,.
$$
In particular, for $m=1$ we get $\|\psi\|_4^4$, which by Sobolev is controlled by $\|\psi\|_{\mathcal H^1}^2$. For $m=3$, the claimed bound follows easily using $\|\psi\|_\infty \lesssim \|\psi\|_{\mathcal H^2}$ and again Sobolev.
\end{proof}

\begin{proof}[Proof of Lemma \ref{wellposedenergy}]
Local well-posedness in $\mathcal H^1\times\mathcal L^2$ follows by a standard fixed-point argument and one sees that $\|\psi_t\|_2$ and $\mathcal E(\psi_t,\phi_t)$ are conserved. One can use \eqref{eq:hls} and the Sobolev inequality to show that \cite[Sec.~2.1]{FrankSchlein2013},
\begin{equation}
\label{eq:energybound}
\mathcal E(\psi,\phi) \geq \|\nabla\psi\|_2^2 + \|\phi\|_2^2 - C \|\phi\|_2 \|\nabla\psi\|_2^{1/2} \|\psi\|_2^{3/2}
\end{equation}
for some universal constant $C>0$. This, together with conservation of $\mathcal E(\psi_t,\phi_t)$, yields global well-posedness as well as the uniform bounds \eqref{eq:energycons}.

According to \eqref{eq:apriorisigma} and the first bound in \eqref{eq:energycons} we have $\|\sigma_{\psi_t}\|\lesssim \|\psi_t\|_{\mathcal H^1}^2 \lesssim 1$, which is the third bound in \eqref{eq:slowPhi1}.

By equation \eqref{eq:defField} for $\phi_t$ we have
$$
\|\alpha^2 \partial_t\phi_t\|_2 \leq \|\phi_t\|_2 + \left\| \sigma_{\psi_t} \right\|_2
$$
and therefore, by the second bound in \eqref{eq:energycons} and the third bound in \eqref{eq:slowPhi1} we obtain the first bound in \eqref{eq:slowPhi1}.

Finally, $\phi_t-\phi_s = \int_s^t \partial_{s_1}\phi_{s_1}\,ds_1$, so for $t>s$ by the first bound in \eqref{eq:slowPhi1}
$$
\|\phi_t - \phi_s\|_2 \leq \int_s^t \|\partial_{s_1}\phi_{s_1}\|_2 \,ds_1 \lesssim \alpha^{-2} |t-s| \,.
$$
This proves the second bound in \eqref{eq:slowPhi1} and completes the proof of the lemma.
\end{proof}

Before dealing with $\mathcal H^4\times \mathcal L^2_{(3)}$-regularity in Proposition \ref{THM:wellposedness}, we need to establish $\mathcal H^2\times \mathcal L^2_{(1)}$-regularity.

\begin{proposition}\label{wellposedh2}
If $(\psi_0,\phi_0)\in\mathcal H^2(\R^3)\times\mathcal L^2_{(1)}(\R^3)$, then $(\psi_t,\phi_t)\in\mathcal H^2(\R^3)\times\mathcal L^2_{(1)}(\R^3)$ for all $t\in\R$ and
$$
\|\psi_t\|_{\mathcal H^2} \lesssim 1+ \alpha^{-2} |t| \,,
\qquad
\|\phi_t\|_{\mathcal L^2_{(1)}(\R^3)} \lesssim 1+ \alpha^{-2} |t|
$$
with implicit constants depending only on the initial data. Moreover,
\begin{equation}
\label{eq:psiH2der}
\|\partial_t\psi_t\|_{\mathcal L^2} \lesssim 1+\alpha^{-2} |t| \,,
\qquad
\|\partial_t\sigma_{\psi_t}\|_{\mathcal L^2} \lesssim 1+\alpha^{-2} |t| \,.
\end{equation}
If, in addition, $\phi_0\in\mathcal L^2_{(m)}(\R^3)$, $m=2,3$, then $\phi_t\in\mathcal L^2_{(m)}(\R^3)$ for all $t\in\R$ and 
$$
\|\phi_t\|_{\mathcal L^2_{(m)}(\R^3)} \lesssim 1+ \alpha^{-6} |t|^3 \,.
$$
\end{proposition}

\begin{proof}
By a standard fixed point argument one can show local existence of solutions in $\mathcal H^2\times\mathcal L^2_{(1)}$. In the following we will construct an functional, which is equivalent to the $\mathcal H^2$ norm of $\psi$ and which grows in a controlled way as time increases. This will prove, in particular, that $\psi_t$ belongs to $\mathcal H^2$ for all times.

We claim that for every $A>0$ there is a constant $M>0$ such that
$$
\mathcal E^{(2)}(\psi,\phi) := \left\| (-\Delta+V_\phi+M) \psi \right\|_2^2
$$
satisfies
\begin{equation}
\label{eq:equivh2}
(1/2)\|\psi\|_{\mathcal H^2} \leq \left( \mathcal E^{(2)}(\psi,\phi) \right)^{1/2} \leq (3/2) \|\psi\|_{\mathcal H^2}
\end{equation}
for all $\psi\in\mathcal H^2$ and all $\phi$ satisfying $\|\phi\|_2 \leq A$. In fact, similarly as in the proof of Corollary \ref{formopdom}, we have
$$
\left| \left\| (-\Delta+V_\phi+M) \psi \right\|_2 - \left\| (-\Delta+M) \psi \right\|_2 \right|
\leq \left\| V_\phi(-\Delta+M)^{-1} \right\| \left\| (-\Delta+M) \psi \right\|_2
$$
and according to Lemma \ref{infiniformop} we can choose $M$ such that the first factor on the right side is less than $\epsilon$ for $\|\phi\|_2\leq A$.

According to Lemma \ref{wellposedenergy} there is an $A>0$ (depending only on $\|\psi_0\|_{\mathcal H^1}$ and $\|\phi_0\|_{\mathcal L^2}$) such that $\|\phi_t\|_{\mathcal L^2}\leq A$ for all $t$. We choose $M$ corresponding to this value of $A$ and compute, using the equation for $\psi_t$,
\begin{align*}
\partial_t \mathcal E^{(2)}(\psi_t,\phi_t) = & 2\re \left((-\Delta+V_{\phi_t}+M)\psi_t,(-\Delta+V_{\phi_t}+M)\partial_t \psi_t\right) \\
& + 2\re \left((-\Delta+V_{\phi_t}+M)\psi_t,(\partial_t V_{\phi_t}) \psi_t\right) \\
= & 2\re \left((-\Delta+V_{\phi_t}+M)\psi_t,(\partial_t V_{\phi_t}) \psi_t\right).
\end{align*}
By the Schwarz and the H\"older inequality,
$$
\partial_t \mathcal E^{(2)}(\psi_t,\phi_t) \leq 2 \left( \mathcal E^{(2)}(\psi_t,\phi_t) \right)^{1/2} \left\|\partial_t V_{\phi_t} \right\|_6 \|\psi_t\|_3
$$
By \eqref{eq:hls} and Lemma \ref{wellposedenergy}, $\|\partial_t V_{\phi_t}\|_6\lesssim \|\partial_t\phi_t\|_2 \lesssim \alpha^{-2}$, and by the Sobolev inequality and Lemma \ref{wellposedenergy}, $\|\psi_t\|_3 \lesssim \|\psi_t\|_{\mathcal H^1}\lesssim 1$. Thus,
$$
\partial_t \mathcal E^{(2)}(\psi_t,\phi_t) \lesssim \alpha^{-2} \left( \mathcal E^{(2)}(\psi_t,\phi_t) \right)^{1/2} \,,
$$
which implies that $\left(\mathcal E^{(2)}(\psi_t,\phi_t)\right)^{1/2} \lesssim 1+ \alpha^{-2} |t|$. According to \eqref{eq:equivh2} this implies the claimed bound on $\|\psi_t\|_{\mathcal H^2}$.

The remaining bounds are proved in a straightforward way. We have
$$
\|\partial_t\psi_t\|_{2} \leq \|-\Delta\psi_t\|_{2} + \|V_{\phi_t}\psi_t\|_{2}
\leq \|\psi_t\|_{\mathcal H^2} + \|V_{\phi_t}\|_6 \|\psi_t\|_3
$$
By the bound on $\|\psi_t\|_{\mathcal H^2}$ together with \eqref{eq:hls} and the bounds from Lemma \ref{wellposedenergy} we obtain the first bound in \eqref{eq:psiH2der}. Moreover,
$$
\partial_t \sigma_{\psi_t} = 2|k|^{-1} \int_{\R^3} \re\left(\overline{\psi_t}\partial_t\psi_t\right) e^{ik\cdot x}\,dx
$$
and so, by the Hardy--Littlewood--Sobolev inequality as in \eqref{eq:hls},
$$
\left\|\partial_t \sigma_{\psi_t}\right\|_2 \lesssim \|\psi_t\partial_t\psi_t\|_{6/5}\leq \|\psi_t\|_3 \|\partial_t\psi_t\|_2 \,.
$$
By the first bound in \eqref{eq:psiH2der} and Lemma \ref{wellposedenergy} we obtain the second bound in \eqref{eq:psiH2der}.

In order to deduce the bounds on $\phi_t$ we use Duhamel's formula
\begin{equation}
\label{eq:duhamelphi}
\phi_t(k) = e^{-it/\alpha^2}\phi_0(k) - i\alpha^{-2} \int_0^t e^{-i(t-s)/\alpha^2}\sigma_{\psi_s}(k)\,ds \,.
\end{equation}
If $\phi_0\in\mathcal L^2_{(m)}$, $m=1,2,3$, we deduce that $\phi_t\in\mathcal L^2_{(m)}$ provided we can bound $\|\sigma_{\psi_s}\|_{\mathcal L^2_{(m)}}$. This quantity can by controlled by Sobolev norms of $\psi_s$ according to \eqref{eq:apriorisigma}. This proves the proposition.
\end{proof}

\begin{proof}[Proof of Proposition \ref{THM:wellposedness}]
The basic strategy is the same as in the proof of Lemma \ref{wellposedh2}, except that verifying the properties of the functional is more complicated in this case. Again we do not give the details of the local existence via a fixed point argument.

We claim that for every $A>0$ there is a constant $M>0$ such that
$$
\mathcal E^{(4)}(\psi,\phi) := \left\| (-\Delta+V_\phi+M)^2 \psi \right\|_2^2
$$
satisfies
\begin{equation}
\label{eq:equivh4}
(1/2)\|\psi\|_{\mathcal H^4} \leq \left( \mathcal E^{(4)}(\psi,\phi) \right)^{1/2} \leq (3/2) \|\psi\|_{\mathcal H^4}
\end{equation}
for all $\psi\in\mathcal H^4$ and all $\phi$ satisfying $\|\phi\|_{\mathcal L^2_{(3)}}\leq A$. To show this, we first observe that, as in the proof of Lemma \ref{wellposedh2},
\begin{align*}
& \left| \left\| (-\Delta+V_\phi+M)^2 \psi \right\|_2 - \left\| (-\Delta+M)(-\Delta+V_\phi+M) \psi \right\|_2 \right| \\ 
& \qquad \leq \left\| V_\phi (-\Delta+M)^{-1} \right\| \left\| (-\Delta+M)(-\Delta+V_\phi+M) \psi \right\|_2
\end{align*}
and that $\left\| V_\phi (-\Delta+M)^{-1} \right\|$ can be made arbitrarily small for $\|\phi\|_{\mathcal L^2}$ bounded by choosing $M$ large. Thus, it suffices to show that $\left\| (-\Delta+M)(-\Delta+V_\phi+M) \psi \right\|_2$ is equivalent to $\|(-\Delta+M)^2\psi\|_2$. We compute
\begin{align*}
& \left| \left\| (-\Delta+M)(-\Delta+V_\phi+M) \psi \right\|_2 - \left\| (-\Delta+V_\phi+M) (-\Delta+M) \psi \right\|_2  \right| \\ 
& \qquad \leq \left\| \left( 2\nabla V_\phi\cdot\nabla + \Delta V_\phi \right) (-\Delta+M)^{-1} \right\| \left\| (-\Delta+M) \psi \right\|_2 \,.
\end{align*}
According to \eqref{eq:aprioriv}, the first factor on the right side can be made arbitrarily small for $\|\phi\|_{\mathcal L^2_{(3)}}$ bounded by choosing $M$ large. We finally apply the argument in Lemma \ref{wellposedh2} again to compare $ \left\| (-\Delta+V_\phi+M) (-\Delta+M) \psi \right\|_2$ to  $\left\| (-\Delta+M)^2 \psi \right\|_2$. This proves the claim.

According to Lemma \ref{wellposedh2} for every $\tau>0$ there is an $A>0$ (depending only on $\|\psi_0\|_{\mathcal H^2}$, $\|\phi_0\|_{\mathcal L^2_{(3)}}$ and $\tau$) such that $\|\phi_t\|_{\mathcal L^2_{(3)}}\leq A$ for all $|t|\leq \tau \alpha^2$. We choose $M$ corresponding to this value of $A$ and compute, using the equation for $\psi_t$,
\begin{align*}
\partial_t \mathcal E^{(4)}(\psi_t,\phi_t) = & 2\re \left((-\Delta+V_{\phi_t}+M)^2\psi_t,(-\Delta+V_{\phi_t}+M)^2\partial_t \psi_t\right) \\
& + 2\re \left((-\Delta+V_{\phi_t}+M)^2\psi_t,(\partial_t V_{\phi_t}) (-\Delta+V_{\phi_t}+M)\psi_t\right) \\
& + 2\re \left((-\Delta+V_{\phi_t}+M)^2\psi_t, (-\Delta+V_{\phi_t}+M)(\partial_t V_{\phi_t})\psi_t\right) \\
= & 4\re \left((-\Delta+V_{\phi_t}+M)^2\psi_t,(\partial_t V_{\phi_t}) (-\Delta+V_{\phi_t}+M)\psi_t\right) \\
& - 2\re \left((-\Delta+V_{\phi_t}+M)^2\psi_t, (2\nabla\partial_t V_{\phi_t}\cdot\nabla + \Delta\partial_tV_{\phi_t})\psi_t\right).
\end{align*}
Therefore, by the Schwarz inequality,
\begin{align*}
\partial_t \mathcal E^{(4)}(\psi_t,\phi_t)
\leq & 2 \left( \mathcal E^{(4)}(\psi_t,\phi_t) \right)^{1/2} \left(2 \|\partial_t V_{\phi_t}\|_\infty \|(-\Delta+V_{\phi_t}+M)\psi_t\|_2 \right. \\
& \qquad\qquad\qquad\qquad \left. + 2\|\nabla\partial V_{\phi_t}\|_\infty \|\nabla\psi_t\|_2 + \|\nabla\partial_t V_{\phi_t}\|_\infty \|\psi_t\|_2 \right)
\end{align*}
According to Lemma \ref{wellposedh2} and \eqref{eq:equivh2} all terms involving $\psi_t$ here are bounded by a constant for $|t|\leq \tau\alpha^2$. Assume that we can prove that all terms involving $\phi_t$ here are bounded by a constant times $\alpha^{-2}$ for $|t|\leq \tau\alpha^2$. Then we will have shown that
$$
\partial_t \mathcal E^{(4)}(\psi_t,\phi_t) \lesssim \alpha^{-2} \left( \mathcal E^{(4)}(\psi_t,\phi_t) \right)^{1/2}
$$
for $|t|\leq \tau\alpha^2$, which implies that $\left(\mathcal E^{(4)}(\psi_t,\phi_t)\right)^{1/2} \lesssim 1+\alpha^{-2} |t|\lesssim 1$ for $|t|\leq\tau\alpha^2$. According to \eqref{eq:equivh4}, this proves that $\|\psi_t\|_{\mathcal H^4}\lesssim 1$ for $|t|\leq \tau\alpha^2$.

Thus, it remains to prove that
\begin{equation}
\label{eq:vpartialtbounds}
\|\partial_x^\beta\partial_t V_{\phi_t}\|_\infty \lesssim \alpha^{-2}
\qquad\text{for}\ |t|\leq\tau\alpha^{-2} \,.
\end{equation}
If we insert the equation of $\phi_t$ into the definition of $V_{\phi_t}$, we find
\begin{align}
\partial_{t}V_{\phi_t}(x)= -i\alpha^{-2} \int_{\R^3} \left( e^{-ik\cdot x}\phi_t(k) - e^{ik\cdot x} \overline{\phi_t(k)} \right) \frac{dk}{|k|} \,.
\end{align}
(Note that the contribution from $\sigma_{\psi_t}$ cancels.) Using this formula, we obtain
$$
\|\partial_x^\beta\partial_t V_{\phi_t}\|_\infty \lesssim \alpha^{-2} \|\phi_t\|_{\mathcal L^2_{|\beta|+1}}
$$
in the same way as we obtained \eqref{eq:aprioriv}. This implies \eqref{eq:vpartialtbounds} in view of the bounds on $\phi_t$ from Lemma \ref{wellposedh2}.

It is straightforward to deduce the remaining bounds claimed in the proposition. The bound on $\|\phi_t\|_{\mathcal L^2_{(3)}}$ follows from Proposition \ref{wellposedh2}. Because of the equation for $\psi_t$, we have
$$
\|\partial_t\psi_t\|_{\mathcal H^2} \leq \|-\Delta\psi_t\|_{\mathcal H^2} + \|V_{\phi_t} \psi_t \|_{\mathcal H^2}
\lesssim \|\psi_t\|_{\mathcal H^4} + \sum_{|\beta|\leq 2} \|\partial^\beta V_{\phi_t}\|_\infty \|\psi_t\|_{\mathcal H^2} 
$$
Using the fact that $\|\psi_t\|_{\mathcal H^4}\lesssim 1$ and $\|\phi_t\|_{\mathcal L^2_{(3)}} \lesssim 1$, which by \eqref{eq:aprioriv} controls $\|\partial^\beta V_{\phi_t}\|_\infty$ for $|\beta|\leq 2$, we conclude that $\|\partial_t\psi_t\|_{\mathcal H^2}\lesssim 1$. The second bound in \eqref{eq:psiH4der} follows from Proposition \ref{wellposedh2}.

Finally, we need to prove the bounds on $g_s$ and $g_{s,t}$. By the Schwarz inequality as in the proof of \eqref{eq:aprioriv} together with the equation for $\phi_s$ we find
$$
\|g_s\|_\infty \lesssim \|\partial_s\phi_s\|_{\mathcal L^2_{(1)}}
\leq \alpha^{-2} \left( \|\phi_s\|_{\mathcal L^2_{(1)}} + \|\sigma_{\psi_s}\|_{\mathcal L^2_{(1)}} \right).
$$
According to \eqref{eq:apriorisigma} and Lemma \ref{wellposedenergy} we have $\|\sigma_{\psi_{s}}\|_{\mathcal L^2_{(1)}}\lesssim\|\psi_s\|_{H^1}^2\lesssim 1$. Moreover, if $|t|,|s|\leq\tau\alpha^2$, then Proposition \ref{wellposedh2} implies $\|\phi_{s}\|_{\mathcal L^2_{(1)}}\lesssim 1$. Thus,
$$
\|g_s\|_\infty \lesssim \alpha^{-2} \,,
$$
as claimed. Moreover, $g_{s,t} = \int_s^t g_{s_1} \,ds_1$, so for $t>s$
$$
\|g_{s,t}\|_\infty \leq \int_s^t \|g_{s_1}\|_\infty \,ds_1 \lesssim \alpha^{-2} (t-s) \,.
$$
This proves \eqref{eq:slowPhi}.
\end{proof}

%%%%%%%%%%%%%%%%%%%%%%%%%%%%%%%%%
%%%%%%%%%%%%%%%%%%%%%%%%%%%%%%%%%

\section{Reduced density matrices}\label{sec:reducedabstract}

Here we show how the approximation of $e^{-i\tilde{H}_{\alpha}^{F}t}\psi_0\otimes W(\alpha^2\phi_0)\Omega$ in Theorem \ref{THM:main} yields approximations to its reduced density matrices in Theorem \ref{reduceddensity}. The argument relies on the following abstract lemma.

\begin{lemma}
Let $\mathcal H_1$ and $\mathcal H_2$ be Hilbert spaces, let $\Psi, \Phi \in\mathcal H_1\otimes\mathcal H_2$ and $f\in\mathcal H_1$ and $g\in\mathcal H_2$ such that
$$
\Psi = f\otimes g + \Phi
$$
and
$$
\|f\|_{\mathcal H_1} \leq C\,,\quad \|g\|_{\mathcal H_2} \leq C \,,\quad \| \Phi \|_{\mathcal H_1\otimes\mathcal H_2} \leq C \epsilon
$$
and
$$
\| \langle g,\Phi\rangle_{\mathcal H_2} \|_{\mathcal H_1} \leq C \epsilon^2\,,\quad \| \langle f,\Phi\rangle_{\mathcal H_1} \|_{\mathcal H_2} \leq C\epsilon^2
$$
for some $C>0$ and $\epsilon>0$. Define
$$
\gamma_1 = \tr_{\mathcal H_2} |\Psi\rangle\langle\Psi| \,,\qquad
\gamma_2 = \tr_{\mathcal H_1} |\Psi\rangle\langle\Psi| \,.
$$
Then
$$
\tr_{\mathcal H_1} \left|\gamma_1 - \|g\|^2_{\mathcal H_2}\, |f\rangle\langle f|\right| \leq 3 C^2 \epsilon^2 \,,\quad
\tr_{\mathcal H_2} \left|\gamma_2 - \|f\|^2_{\mathcal H_1}\, |g\rangle\langle g|\right| \leq 3 C^2 \epsilon^2 \,.
$$
\end{lemma}

Before proving this lemma, let us use it to derive Theorem \ref{reduceddensity} from Theorem \ref{THM:main}. We apply the lemma with $\mathcal H_1=\mathcal L^2(\R^3)$, $\mathcal H_2=\mathcal F$, $f=e^{-i\int_0^t \omega(s)\,ds} \psi_t$, $g=\Omega$,
$$
\Psi = W^*(\alpha^2\phi_t) e^{-i\tilde{H}_{\alpha}^{F}t}\psi_0\otimes W(\alpha^2\phi_0)\Omega \,,
\quad
\Phi = W^*(\alpha^2\phi_t)R(t) \,.
$$
Then Theorem \ref{THM:main} implies that the assumptions of the lemma are satisfied with $\epsilon=\alpha^{-1}(1+|t|)$. 
We have $\|f\|^2 = \|\psi_t\|^2=\|\psi_0\|^2=1$, $\|g\|^2 = \|\Omega\|^2= 1$ and $|f\rangle\langle f|= |\psi_t\rangle\langle\psi_t|$. Moreover,
$$
\tr_{\mathcal H_2} |\Psi\rangle\langle\Psi| = \gamma_t^{\mathrm{particle}} \,,
\qquad
\tr_{\mathcal H_1} |\Psi\rangle\langle\Psi| = W^*(\alpha^2\phi_t) \gamma_t^{\mathrm{field}} W(\alpha^2\phi_t) \,.
$$
Thus, the conclusion of Theorem \ref{reduceddensity} follows from the lemma.

We now turn to the proof of the lemma. It relies on the bound
\begin{equation}
\label{eq:traceineq}
\tr_{\mathcal H_1} | \tr_{\mathcal H_2} |\Psi_1\rangle\langle\Psi_2| | \leq \|\Psi_1\|_{\mathcal H_1\otimes\mathcal H_2} \|\Psi_2\|_{\mathcal H_1\otimes\mathcal H_2} 
\end{equation}
valid for any vectors $\Psi_1,\Psi_2\in\mathcal H_1\otimes\mathcal H_2$. For the proof of \eqref{eq:traceineq} recall the variational characterization of the trace norm,
$$
\tr_{\mathcal H_1} |K| = \sup_{(e_j),(e_j')} \re \sum_j \langle e_j,K e_j'\rangle_{\mathcal H_1} \,,
$$
where the supremum is over all orthonormal systems $(e_j)$ and $(e_j')$ in $\mathcal H_1$. Thus, if $(b_k)$ is an orthonormal basis in $\mathcal H_2$, then
\begin{align*}
& \re \sum_j \left\langle e_j, \left( \tr_{\mathcal H_2} |\Psi_1\rangle\langle\Psi_2| \right) e_j'\right\rangle_{\mathcal H_1} = \sum_{j,k} \langle e_j\otimes b_k, \Psi_1\rangle_{\mathcal H_1\otimes\mathcal H_2} \langle \Psi_2,e_j'\otimes b_k\rangle_{\mathcal H_1\otimes\mathcal H_2} \\
& \qquad\leq \left( \sum_{j,k} |\langle e_j\otimes b_k, \Psi_1\rangle_{\mathcal H_1\otimes\mathcal H_2} |^2 \right)^{1/2} \left( \sum_{j,k} |\langle \Psi_2,e_j'\otimes b_k\rangle_{\mathcal H_1\otimes\mathcal H_2}|^2 \right)^{1/2} \\
& \qquad\leq \|\Psi_1\|_{\mathcal H_1\otimes\mathcal H_2} \|\Psi_2\|_{\mathcal H_1\otimes\mathcal H_2} \,,
\end{align*}
where the last inequality comes from the orthonormality of $(e_j\otimes b_k)$ and $(e_j'\otimes b_k)$. Therefore the variational characterization of the trace norm yields \eqref{eq:traceineq}.

\begin{proof}
Since $\tr_{\mathcal H_2} |f\otimes g\rangle\langle \Phi| = |f\rangle\langle \langle g,\Phi\rangle_{\mathcal H_2}|$, we have
\begin{align*}
\gamma_1 - \|g\|^2_{\mathcal H_2}\, |f\rangle\langle f| 
= &  |f\rangle\langle \langle g,\Phi\rangle_{\mathcal H_2}| + | \langle \Phi,g\rangle_{\mathcal H_2}\rangle\langle f| + \tr_2|\Phi\rangle\langle\Phi| \,.
\end{align*}
By \eqref{eq:traceineq} and the assumptions the trace norm of each one of the three operators on the right side is bounded by $C^2 \epsilon^2$. This proves the first inequality in the lemma. The second one is proved similarly.
\end{proof}

Finally, we show that the $\alpha^{-2}$ error bound in Theorem \ref{reduceddensity} (for times of order one) is due to the fact that $\phi_t$ is time-dependent. The proof makes use of the fact that for arbitrary normalized vectors $a$ and $b$ in a Hilbert space $\mathcal H$ one has
\begin{equation}
\label{eq:rankone}
\tr_{\mathcal H} \left| |a\rangle\langle a| - |b\rangle\langle b| \right| = 2 \left( 1- \left|\langle a,b\rangle\right|^2 \right)^{1/2},
\end{equation}
as is easily verified.

\begin{proof}[Proof of Lemma \ref{opt}]
Because of Theorem \ref{reduceddensity} it suffices to prove that there are $\epsilon>0$ and $c>0$ such that for all $|t|\leq\epsilon$ and all $\alpha\geq 1$,
$$
\tr_{\mathcal F} \left| \left|W(\alpha^2\phi_t)\Omega\right\rangle\left\langle W(\alpha^2\phi_t)\Omega\right| - \left|W(\alpha^2\phi_0)\Omega\right\rangle\left\langle W(\alpha^2\phi_0)\Omega\right| \right| \geq c\alpha^{-1} |t| \,. 
$$
According to Lemma \ref{overlap} and \eqref{eq:rankone} this is equivalent to
$$
1- e^{-\alpha^2\|\phi_t-\phi_0\|_2^2} = 
1-\left|\langle \Omega,W^*(\alpha^2\phi_0) W(\alpha^2\phi_t)\Omega \rangle\right|^2 \geq (c^2/4) \alpha^{-2} t^2 \,.
$$
Since $\|\phi_t-\phi_0\|_2 \lesssim \alpha^{-2} |t|$ by Lemma \ref{wellposedenergy}, it suffices to prove that there are $\epsilon>0$ and $c'>0$ such that for all $|t|\leq \epsilon$ and all $\alpha\geq 1$,
$$
\|\phi_t-\phi_0\|_2 \geq c'\alpha^{-2} |t| \,.
$$
Since $\phi_0+\sigma_{\psi_0}\not\equiv 0$, this will clearly follow if we can prove that for all $|t|\leq\alpha^2$ and $\alpha\geq 1$
\begin{equation}
\label{eq:optproof}
\left\|\phi_t - \phi_0 + i\alpha^{-2} t \left(\phi_0+\sigma_{\psi_0}\right) \right\|_2 \leq C \alpha^{-2} t^2 \,.
\end{equation}
To prove this, we use equation \eqref{eq:defParticle} for $\phi_t$ to write
\begin{align*}
\phi_t-\phi_0 = \int_0^t \partial_s\phi_s\,ds = -i\alpha^{-2} \int_0^t \left(\phi_s + \sigma_{\psi_s}\right)ds = -i\alpha^{-2} t \left(\phi_0+\sigma_{\psi_0}\right) + r_t
\end{align*}
with
$$
r_t := -i\alpha^{-2} \int_0^t\int_0^s \left(\partial_{s_1}\phi_{s_1} + \partial_{s_1}\sigma_{\psi_{s_1}}\right) ds_1\,ds \,. 
$$
By Lemma \ref{wellposedenergy} and Proposition \ref{THM:wellposedness} the $\mathcal L^2$-norm of the integrand of $r_t$ is bounded by a constant uniformly in $|s_1|\leq\alpha^2$ and $\alpha\geq 1$. This yields \eqref{eq:optproof} and completes the proof.
\end{proof}

%%%%%%%%%%%%%%%%%%%%%%%%%%%%%%%%%%%%%%%%%%%%%%%%%

%%%%%%%%%%%%%%%%%%%%%%%%%%%%%%%%%%%%

%%%%%%%%%%%%%%%%%%%%%%%%%%%%%%%%%%%%
  
\bibliographystyle{amsalpha}

\begin{thebibliography}{16}

\bibitem{AmFa}
Z. Ammari and M. Falconi.
\newblock Wigner measures approach to the classical limit of the Nelson model: Convergence of dynamics and ground state energy.
\newblock {\em J. Stat. Phys.}, 157(2):330--364, 2014. 

\bibitem{AmNi1}
Z. Ammari and F. Nier.
\newblock Mean field limit for bosons and infinite dimensional phase-space analysis.
\newblock {\em Ann. Henri Poincar\'e}, 9(8):1503--1574, 2008.

\bibitem{AmNi2}
Z. Ammari and F. Nier. 
\newblock Mean field limit for bosons and propagation of Wigner measures. 
\newblock {\em J. Math. Phys.}, 50(4):042107, 2009.

\bibitem{PolaronReview}
J.~T. Devreese and A.~S. Alexandrov.
\newblock Fr\"{o}hlich polaron and bipolaron: recent developments.
\newblock {\em Reports on Progress in Physics}, 72(6):066501, 2009.

\bibitem{MR709647}
M.~D. Donsker and S.~R.~S. Varadhan.
\newblock Asymptotics for the polaron.
\newblock {\em Comm. Pure Appl. Math.}, 36(4):505--528, 1983.

\bibitem{Fa}
M. Falconi.
\newblock Classical limit of the Nelson model with cutoff. 
\newblock {\em J. Math. Phys.}, 54(1):012303, 2013.

\bibitem{FrankSchlein2013}
R.~Frank and B.~Schlein.
\newblock Dynamics of a strongly coupled polaron.
\newblock {\em Lett. Math. Phys.}, 104(8):911--929, 2014.

\bibitem{Fr} 
H. Fr\"ohlich.
\newblock Theory of electrical breakdown in ionic crystals. 
\newblock {\em Proc. R. Soc. Lond. A}, 160:230--241, 1937.

\bibitem{GiNiVe}
J. Ginibre, F. Nironi and G. Velo.
\newblock Partially classical limit of the Nelson model. 
\newblock {\em Ann. Henri Poincar\'e}, 7(1):21--43, 2006.

\bibitem{He}
K. Hepp.
\newblock The classical limit for quantum mechanical correlation functions. 
\newblock {\em Comm. Math. Phys.}, 35:265--277, 1974. 

\bibitem{LaPe}
L. D. Landau and S. I. Pekar.
\newblock Effective mass of a polaron.
\newblock {\em Zh. Eksp. Teor. Fiz.} 18(5):419--423, 1948.

\bibitem{MR1462224}
E.~H. Lieb and L.~E. Thomas.
\newblock Exact ground state energy of the strong-coupling polaron.
\newblock {\em Comm. Math. Phys.}, 183(3):511--519, 1997.

\bibitem{LiYa}
E.~H. Lieb and K.~Yamazaki.
\newblock Ground-state energy and effective mass of the polaron.
\newblock {\em Phys. Rev.}, 111:728--733, 1958.

\bibitem{pekar1946polaron}
S.~Pekar.
\newblock {\em Zh. Eksp. Teor. Fiz.}, 16:335, 1946.

\bibitem{pekar1954polaron}
S.~Pekar.
\newblock Research in electron theory of crystals.
\newblock Gostekhizdat, Moskva, 1951. English transl.: US AEC Tech. Rep. AEC-tr-5575, Washington, DC, 1963.

\bibitem{Spohn1987}
H.~Spohn.
\newblock Effective mass of the polaron: a functional integral approach.
\newblock {\em Ann. Physics}, 175(2):278--318, 1987.

\end{thebibliography}

\end{document}